\documentclass[a4paper,10pt]{article}

\usepackage{latexsym}
\usepackage{amsmath}
\usepackage{amsthm}
\usepackage{amssymb}
\usepackage{enumerate}
\usepackage{multicol}
\usepackage{graphicx}

\allowdisplaybreaks

\theoremstyle{plain}
\newtheorem{theorem}{Theorem}[section]
\newtheorem{proposition}[theorem]{Proposition}
\newtheorem{lemma}[theorem]{Lemma}
\newtheorem{corollary}[theorem]{Corollary}

\theoremstyle{definition}
\newtheorem{definition}[theorem]{Definition}
\newtheorem{example}[theorem]{Example}
\newtheorem{remark}[theorem]{Remark}
\numberwithin{equation}{section}

\def\C{{\mathbb{C}}}% \C == \mathbb{C}

\def\T{{\mathbb{T}}}% \T == \mathbb{T}
\def\R{{\mathbb{R}}}% \R == \mathbb{R}
\def\N{{\mathbb{N}}}% \N == \mathbb{N}
\def\S{{\mathbb{S}}}% \S == \mathbb{S}
\def\X{{\mathbb{X}}}
\def\Z{{\mathbb{Z}}}% \Z == \mathbb{Z}
\def\cB{{\mathcal{B}}}
\def\cD{{\mathcal{D}}}
\def\cE{{\mathcal{E}}}
\def\cH{{\mathcal{H}}}

\def\cO{{\mathcal{O}}}
\def\cU{{\mathcal{U}}}

\def\<{{\langle}}% \< == \langle
\def\>{{\rangle}}% \> == \rangle

\def\Tr{\mathop{\mathrm{Tr}}}
\def\Arg{\mathop{\mathrm{Arg}}}

\def\dis{\mathop{\mathrm{dis}}\nolimits}

\def\sgn{\mathop{\mathrm{sgn}}\nolimits}

\def\Im{\mathop{\mathrm{Im}}}
\def\Re{\mathop{\mathrm{Re}}}

\def\b0{{\mathbf{0}}}
\def\ba{{\mathbf{a}}}
\def\bb{{\mathbf{b}}}
\def\bB{{\mathbf{B}}}

\def\be{{\mathbf{e}}}

\def\bG{{\mathbf{G}}}
\def\bk{{\mathbf{k}}}

\def\bx{{\mathbf{x}}}
\def\by{{\mathbf{y}}}
\def\bz{{\mathbf{z}}}
\def\bu{{\mathbf{u}}}
\def\bv{{\mathbf{v}}}
\def\bw{{\mathbf{w}}}

\def\bs{{\mathbf{s}}}
\def\bS{{\mathbf{S}}}

\def\tn{{\tilde{n}}}
\def\hm{{\hat{m}}}

\def\hX{{\hat{X}}}
\def\hY{{\hat{Y}}}
\def\hbx{{\hat{\mathbf{x}}}}
\def\hby{{\hat{\mathbf{y}}}}
\def\hphi{{\hat{\phi}}}
\def\hPhi{{\hat{\Phi}}}

\def\talpha{{\tilde{\alpha}}}
\def\opsi{{\overline{\psi}}}

\def\spin{{\{\uparrow,\downarrow\}}}
\def\ua{{\uparrow}}
\def\da{{\downarrow}}

\def\o{{\omega}}

\def\hxi{\hat{\xi}}
\def\hXi{\hat{\Xi}}
\def\eps{{\varepsilon}}% \eps == \varepsilon
\def\g{{\gamma}}% \g == \gamma
\def\G{{\Gamma}}% \G == \Gamma
\def\s{{\sigma}}% \s == \sigma

\def\D{{\Delta}}% \D == \Delta

\begin{document}

\title{Exponential Decay of Correlation Functions in Many-Electron Systems}
\author{Yohei Kashima \medskip \\
Institut f\"ur Theoretische Physik, Universit\"at Heidelberg\\
Philosophenweg 19, 69120 Heidelberg, Germany\\ 
y.kashima@thphys.uni-heidelberg.de}
\date{}

\maketitle

\begin{abstract}
\noindent
For a class of tight-binding many-electron models on hyper-cubic lattices the equal-time correlation functions at non-zero temperature are proved to decay exponentially in the distance between the center of positions of the electrons and the center of positions of the holes. The decay bounds hold in any space dimension in the thermodynamic limit if the interaction is sufficiently small depending on temperature. The proof is based on the $U(1)$-invariance property and volume-independent perturbative bounds of the finite dimensional Grassmann integrals formulating the correlation functions. 
\end{abstract}

\section{Introduction}
A Hamiltonian governing the total energy of many electrons hopping and interacting on a finite lattice is defined as a self-adjoint operator on the finite dimensional Hilbert space of all states of electrons. Correlation functions in the system at non-zero temperature are formulated as a quotient of trace operations over the Hilbert space. Despite the explicitness of their mathematical definitions, to rigorously analyze the behavior of the correlation functions still requires restrictive assumptions and remains to be solved in a general setting.

The method using Bogoliubov's inequality initiated by Hohenberg \cite{H}, Mermin and Wagner \cite{MW} has been applied to prove decay of various order parameters in the one- and two-dimensional Hubbard models. It has been shown that the magnetic order parameters (\cite{WR},\cite{G},\cite{U}) and the electron-pairing order parameters (\cite{SSZ},\cite{SS}) vanish in the thermodynamic limit as the amplitude of the corresponding external field goes to zero. In these theories the application of Bogoliubov's inequality demonstrates that the thermodynamic limit of the order parameter under consideration is bounded from above by a constant times the inverse of an integral of the form
$$
\prod_{j=1}^d\int_{-\pi}^{\pi}dk_j\frac{1}{\sum_{j=1}^dk_j^2+\lambda}
$$
with the space dimension $d$ and the amplitude $\lambda>0$ of the external field. If this integral diverges as the amplitude $\lambda$ is sent to zero, which is true for $d=1,2$, not true for $d\ge 3$, one concludes that the order parameter accordingly converges to zero. In order to reach this conclusion, thus, these theories require the space dimension to be less than or equal to $2$.

On the other hand, the method proposed by McBryan and Spencer \cite{MS} to prove decay properties of correlations in classical spin systems has been extended to bound the correlation functions in the one- and two-dimensional Hubbard models in \cite{KT}, \cite{MR}. These theories make use of $U(1)$-symmetry property of the system and deduce that the electron pairing-pairing correlation function for $2$ separate sites $\bx$, $\by$ is bounded from above by 
$$e^{-C_1(\theta_{\bx}-\theta_{\by})+C_2\sum_{\bu,\bv}|t_{\bu,\bv}|(\cosh(\theta_{\bu}-\theta_{\bv})-1)},
$$
where $C_1,C_2>0$ are constants, $\{\theta_{\bu}\}$ are arbitrary taken real parameters indexed by every site on the lattice, $(t_{\bu,\bv})_{\bu,\bv}$ is the hopping matrix and the sum with respect to $\bu,\bv$ is taken over all the sites. A suitable choice of the parameters $\{\theta_{\bu}\}$ yields a decaying upper bound on the correlation function. As in \cite{MS} Koma and Tasaki \cite{KT} took $\{\theta_{\bu}\}$ to be the fundamental solution of a Laplace equation on the lattice and concluded that the pairing-pairing correlation function decays as $|\bx-\by|\to +\infty$. Macris and Ruiz \cite{MR} referred to a list of the possible parameters $\{\theta_{\bu}\}$ summarized in \cite{MMR} and extended the decay bounds obtained by Koma and Tasaki to be valid for the Hubbard models with long range hopping matrix as well. Since appropriate choices of $\{\theta_{\bu}\}$ have been found in one and two dimensions, this approach verifies the decay of the correlation functions in these low dimensional cases at present, to the best of the author's knowledge. 

Apart from these analysis to bound the correlation functions in low dimensions, Kubo and Kishi \cite{KK} proved that the susceptibilities for the Hubbard models at non-zero temperature are bounded from above by the inverse of modulus of the coupling constant in any space dimension under a few assumptions on the sign of the coupling constant and the lattice.

In this paper we consider the equal-time correlation functions for a class of the Hubbard models at non-zero temperature and show that the correlation functions decay exponentially in the distance between the center of positions of the electrons and the center of positions of the holes in any space dimension if a norm of the interaction term of the Hamiltonian is sufficiently small (see Theorem \ref{thm_exponential_decay} and Theorem \ref{thm_exponential_decay_hubbard_model} in Section \ref{sec_models_results} for the precise statements). As in \cite{KT}, \cite{MR} our approach essentially uses the $U(1)$-symmetry property of the model. We start from characterizing the correlation functions as a limit of Grassmann integrals over finite Grassmann variables. The Grassmann integral formulations called the Schwinger functions are mathematical objects defined on a rigorous base. The $U(1)$-symmetry property of the model is simply implemented in the Grassmann integral and enables us to convert the formulation into a form of multi-contour integral of the Schwinger function with respect to newly introduced complex variables contained in the covariance matrix. Our objective is, thus, set to find an upper bound on the Schwinger function inside the multi-contour integral. The evaluation of the Schwinger function is done perturbatively. We find a volume-independent upper bound on each term of the Taylor expansion of the Schwinger function with respect to the interaction around zero. We require the interaction to be small so that the perturbation series of the Schwinger function converges.

One advantage of this approach is that the space dimension of the system causes  no technical difficulty and the resulting decay bounds on the correlation functions are valid in any dimension. However, as we need to go through the perturbation theory with respect to the interaction, the additional assumption is imposed on the magnitude of the interaction.  

Rigorous frameworks have been developed to control the perturbation theory for many-Fermion on lattice independently of the volume factor. In \cite{PS} Pedra and Salmhofer extended the notion of Gram's inequality to be applicable to the determinant of the covariance matrices appearing in many-Fermion systems. Their abstract theorem \cite[\mbox{Theorem 1.3}]{PS} is general enough to cover the modified covariance coming into play in our construction and bounds its determinant independently of volume and temperature. The tree formula for partial derivatives of logarithm of the Grassmann Gaussian integral summarized in \cite{SW} coupled with the determinant bound of the covariance makes it feasible to find volume-independent upper bounds on the perturbative expansion of the Schwinger functions. Though in this paper we employ the tree formula to bound the Schwinger functions, the same goal can also be achieved via a concise representation of the Schwinger functions established by Feldman, Kn\"orrer and Trubowitz in \cite{FKTrep}. Indeed, their bound \cite[\mbox{Theorem I.9}]{FKTrep} proved in a quite general context needs only the determinant bound and the $L^1$-bound of the covariance, which are the same inputs to return the upper bounds on the Schwinger functions by means of the tree formula.

We often refer to the recent article \cite{K} for basic lemmas needed in our construction. In fact, this work should be regarded as a continuation of \cite{K}, which intended to explain mathematical tools to analyze the perturbation theory for many electrons in detail.

As in \cite{K} we directly treat the perturbation theory of the original model without introducing any multi-scale technique. Though the results can be presented explicitly in a simple manner through the single scale analysis, it costs the temperature dependency of the interaction. In our analysis the norm of the interaction is restricted to be less than a constant times $\beta^{-d-1}$ for the inverse temperature $\beta$ in $d$-dimensional case. As a renormalization group analysis in the theoretical front let us remark that for a wide class of the Fermionic lattice models various mathematical properties of the correlation functions have been intensively studied by Pedra \cite{P} in a generalized form. Pedra's multi-scale analysis remarkably concludes that the correlation functions can be qualitatively analyzed if a norm of the interaction is bounded by a constant times $(\log \beta)^{-1}$ in a $2$-dimensional case, or less than a constant times $\beta^n$ with some $n\in (-d-1,-d/2+1/2]$ in $d$-dimensional cases $(d\ge 2)$ as well.

This paper is organized as follows. In Section \ref{sec_models_results} we define the model Hamiltonians, prepare notations and state the main results of this paper. In Section \ref{sec_formulation} the correlation functions are formulated into a limit of the finite dimensional Grasssmann integrals. In Section \ref{sec_perturbative_bound} upper bounds on each term of the Taylor series of the Grassmann integral formulations are obtained. In Section \ref{sec_exponential_decay_correlation} we first prove that the covariance matrix fulfills the necessary requirements for the Grassmann integral formulation to be bounded perturbatively. By using the perturbative bounds on the Grassmann integral formulations, we then complete the proof of our main results on the exponential decay property of the correlation functions. Appendix \ref{app_anti_symmetricity} shows that the coefficient function defining the interaction term can be replaced by a unique anti-symmetric function. Appendix \ref{app_thermodynamic_limit} presents a proof of existence of the thermodynamic limit of the correlation functions.
  
\section{Model Hamiltonians and main results}\label{sec_models_results}
In this section we define Hamiltonian operators and correlation functions together with notations used in this paper and state the main results. We use the standard terminology concerning the Fermionic Fock space without providing the definitions. They are documented, e.g, in the book \cite{BR} or briefly in \cite[\mbox{Appendix A}]{K}.

\subsection{The Hamiltonian operator}
We are going to define our Hamiltonian operator on the Fermionic Fock space on the $d$-dimensional hyper-cubic lattice $\G=(\Z/(L\Z))^d$ $(L,d\in\N)$ and the spin coordinate $\spin$. The lattice $\G^*$ of momentum is given by $\G^*:=((2\pi\Z/L)/(2\pi\Z))^d$. We admit conventions that Kronecker's delta $\delta_{\bx,\by}$ takes $1$ if the element $\bx$ is identical with $\by$ in the set they belong to, $0$ otherwise and the function $1_P$ of a proposition $P$ returns $1$ if $P$ is true, $0$ otherwise. For any vectors $\ba=(a_1,\cdots,a_n)$, $\bb=(b_1,\cdots,b_n)$ of algebra, let $\<\ba,\bb\>$ denote the sum $\sum_{j=1}^na_jb_j$. Let $\|\cdot\|_{\R^n}$ be the Euclidean norm in $\R^n$, $\<\cdot,\cdot\>_{\C^n}$ be the inner product of $\C^n$ and $\|\cdot\|_{\C^n}$ denote the norm of $\C^n$ induced by $\<\cdot,\cdot\>_{\C^n}$. For any finite set $S$ let $\sharp S$ denote the number of elements of $S$.

With the creation operator $\psi^*_{\bx\xi}$ and the annihilation operator $\psi_{\bx\xi}$ for any $\bx\in\G$ and $\xi\in\spin$ the free part $H_0$ is defined by 
$$H_0:= \sum_{\bx,\by\in \G}\sum_{\xi,\phi\in\spin}T(\bx\xi,\by\phi)
\psi_{\bx\xi}^*\psi_{\by\phi},$$
where the short range hopping matrix $\{T(\bx\xi,\by\phi)\}_{(\bx,\xi),(\by,\phi)\in \G\times\spin}$ is given by
\begin{equation*}
\begin{split}
T&(\bx\xi,\by\phi):=\delta_{\xi,\phi}\Bigg(-t\sum_{j=1}^d(\delta_{\bx,\by-\be_j}+\delta_{\bx,\by+\be_j})\\
&-t'\cdot 1_{d\ge 2}\sum_{j,k=1 \atop j<k}^d(\delta_{\bx,\by-\be_j-\be_k}+\delta_{\bx,\by-\be_j+\be_k}+\delta_{\bx,\by+\be_j-\be_k}+\delta_{\bx,\by+\be_j+\be_k})-\mu \delta_{\bx,\by}\Bigg).
\end{split}
\end{equation*}  
The real parameters $t,t'$ and $\mu$ are called the nearest neighbor hopping amplitude, the next to nearest neighbor hopping amplitude, and the chemical potential, respectively. The vectors $\be_j\in\G$ $(j\in\{1,\cdots,d\})$ are defined by $\be_j(l):=\delta_{j,l}$ for all $j,l\in \{1,\cdots,d\}$.

Since our problem becomes trivial otherwise (see Remark \ref{rem_H0_trivialcase}), we assume that 
\begin{equation}\label{eq_nonzero_hopping}
|t|+|t'|1_{d\ge 2}\neq 0
\end{equation}
throughout the paper.

To define the interacting part of the Hamiltonian we introduce functions $U_l:(\Z^d)^l\times \spin^{2l}\to \C$ $(l\in\{1,\cdots,\tn\})$ satisfying the equality
\begin{equation}\label{eq_condition_U}
\begin{split}
&\overline{U_l((\bx_{1},\cdots,\bx_{l}),(\xi_{1},\cdots,\xi_{l}),(\phi_{1},\cdots,\phi_{l}))}\\
&\quad =U_l((\bx_{1},\cdots,\bx_{l}),(\phi_{1},\cdots,\phi_{l}),(\xi_{1},\cdots,\xi_{l}))\ (\forall l\ge 1)
\end{split}
\end{equation}
and the translation invariance
\begin{equation}\label{eq_translation_invariance}
\begin{split}
&U_l((\bx_1,\bx_2,\cdots,\bx_l),(\xi_1,\xi_2,\cdots,\xi_l),(\phi_1,\phi_2,\cdots,\phi_l))\\
&\quad = U_l((\bx_1+\by,\bx_2+\by,\cdots,\bx_l+\by),(\xi_1,\xi_2,\cdots,\xi_l),(\phi_1,\phi_2,\cdots,\phi_l))\ (\forall l\ge 2)
\end{split}
\end{equation}
for all $(\bx_1,\cdots,\bx_l)\in (\Z^d)^l$, $(\xi_{1},\cdots,\xi_{l}),(\phi_{1},\cdots,\phi_{l})\in \spin^l$ and $\by\in  \Z^d$.

We define a restriction of $U_l$ by periodicity in the following way. Let $\lfloor L/2 \rfloor$ denote the largest integer not exceeding $L/2$. For any $\bx\in\Z^d$ there uniquely exists $\bx^L\in (\{-\lfloor L/2 \rfloor,-\lfloor L/2 \rfloor+1,\cdots,-\lfloor L/2 \rfloor+L-1\})^d$ such that $\bx=\bx^L$ in $\G$. By using this identification we define $U_{L,l}:(\Z^d)^l\times \spin^{2l}\to \C$ $(l\in\{1,\cdots,\tn\})$ by
\begin{equation}\label{eq_reduction_periodicity}
\begin{split}
&U_{L,1}(\bx,\xi,\phi):=U_{1}(\bx^L,\xi,\phi),\\
&U_{L,l}((\bx_1,\bx_2,\cdots,\bx_{l-1},\bx_l),(\xi_1,\cdots,\xi_l),(\phi_1,\cdots,\phi_l))\\
&\quad:=U_{l}(((\bx_1-\bx_l)^L,(\bx_2-\bx_l)^L,\cdots,(\bx_{l-1}-\bx_l)^L,\b0),(\xi_1,\cdots,\xi_l),(\phi_1,\cdots,\phi_l)) 
\end{split}
\end{equation}
for $l\in\{2,\cdots,\tn\}$. Note that $U_{L,l}$ is periodic with respect to the spacial variables and obeys \eqref{eq_condition_U}-\eqref{eq_translation_invariance} and 
\begin{equation}\label{eq_limit_equivalence}
\begin{split}
&\lim_{L\to +\infty\atop L\in\N}U_{L,l}((\bx_1,\cdots,\bx_l),(\xi_1,\cdots,\xi_l),(\phi_1,\cdots,\phi_l))\\
&\quad = U_{l}((\bx_1,\cdots,\bx_l),(\xi_1,\cdots,\xi_l),(\phi_1,\cdots,\phi_l))\end{split}
\end{equation}
for all $(\bx_1,\cdots,\bx_l)\in \Z^l$, $(\xi_1,\cdots,\xi_l),(\phi_1,\cdots,\phi_l)\in\G^l$ $(\forall l\in\{1,\cdots,\tn\})$.

The interacting part $V$ is defined as follows.
\begin{equation}\label{eq_definition_V}
\begin{split}
V:=\sum_{l=1}^{\tn}\sum_{\bx_{j}\in\G\atop\forall j\in\{1,\cdots,l\}}&\sum_{\xi_{j},\phi_{j}\in\spin \atop \forall j\in\{1,\cdots,l\}}U_{L,l}((\bx_{1},\cdots,\bx_{l}),(\xi_{1},\cdots,\xi_{l}),(\phi_{1},\cdots,\phi_{l}))\\
&\cdot\psi_{\bx_{1}\xi_{1}}^*\psi_{\bx_{2}\xi_{2}}^*\cdots\psi_{\bx_{l}\xi_{l}}^*\psi_{\bx_{l}\phi_{l}}\psi_{\bx_{l-1}\phi_{l-1}}\cdots\psi_{\bx_{1}\phi_{1}}.
\end{split}
\end{equation}
Note that the condition \eqref{eq_condition_U} makes $V$ self-adjoint. The following examples motivate us to generalize the interacting part of the Hamiltonian as defined above.

\begin{example}[the density-density interaction]
With real functions $U^{dd}_l:(\Z^d)^l\times\spin^l\to \R\ (l\in\{1,\cdots,n\})$ satisfying the translation invariance 
$$U^{dd}_l((\bx_1,\cdots,\bx_l),(\xi_1,\cdots,\xi_l))=U^{dd}_l((\bx_1+\by,\cdots,\bx_l+\by),(\xi_1,\cdots,\xi_l))$$
and $U^{dd}_l((\bx_1,\cdots,\bx_l),(\xi_1,\cdots,\xi_l))=0$ if $(\bx_j,\xi_j)= (\bx_k,\xi_k)$ with $j\neq k$ $(\forall l\ge 2)$, 
let us define the density-density interaction $V_{dd}$ by 
\begin{equation*}
V_{dd}:=\sum_{l=1}^n\sum_{\bx_{j}\in\G\atop \forall j\in\{1,\cdots,l\}}\sum_{\xi_{j}\in\spin\atop \forall j\in\{1,\cdots,l\}}U_{L,l}^{dd}((\bx_{1},\cdots,\bx_{l}),(\xi_{1},\cdots,\xi_{l}))\prod_{j=1}^l\psi_{\bx_{j}\xi_{j}}^*\psi_{\bx_{j}\xi_{j}},
\end{equation*}
where the function $U_{L,l}^{dd}$ is derived from $U_{l}^{dd}$ in the same way as in \eqref{eq_reduction_periodicity}. The operator $V_{dd}$ can be rewritten as 
\begin{equation*}
\begin{split}
V_{dd}=&\sum_{l=1}^n\sum_{\bx_{j}\in\G\atop \forall j\in\{1,\cdots,l\}}\sum_{\xi_{j},\phi_{j}\in\spin\atop \forall j\in\{1,\cdots,l\}}U_{L,l}^{dd}((\bx_{1},\cdots,\bx_{l}),(\xi_{1},\cdots,\xi_{l}))\prod_{j=1}^l\delta_{\xi_{j},\phi_{j}}\\
&\cdot\psi_{\bx_{1}\xi_{1}}^*\psi_{\bx_{2}\xi_{2}}^*\cdots\psi_{\bx_{l}\xi_{l}}^*\psi_{\bx_{l}\phi_{l}}\psi_{\bx_{l-1}\phi_{l-1}}\cdots\psi_{\bx_{1}\phi_{1}},
\end{split}
\end{equation*}
which shows that $V_{dd}$ has the form \eqref{eq_definition_V}.
\end{example} 

\begin{example}[the spin operator coupled with a local magnetic field]\label{ex_spin_field}
Introduce a local magnetic field $\bB_{\bx}=(B_{\bx}^{(1)},B_{\bx}^{(2)},B_{\bx}^{(3)}):\Z^d\to \R^3$ and let $\bB_{L,\bx}=$\\
$(B_{L,\bx}^{(1)},B_{L,\bx}^{(2)},B_{L,\bx}^{(3)}):\G\to \R^3$ be the restriction of $\bB_{\bx}$ on $\G$ by periodicity as defined in \eqref{eq_reduction_periodicity}. With the spin operator $\bS_{\bx}=(S_{\bx}^{(1)},S_{\bx}^{(2)},S_{\bx}^{(3)})$ $(\bx\in\G)$ given by
$$S_{\bx}^{(l)}:=\frac{1}{2}\sum_{\xi,\phi\in\{\ua,\da\}}P_{\xi,\phi}^{(l)}\psi_{\bx\xi}^*\psi_{\bx\phi}\ (l\in\{1,2,3\})$$
with the Pauli matrices
$$P^{(1)}=\left(\begin{array}{cc} 0 & 1 \\ 1 & 0 \end{array}\right),\ P^{(2)}=\left(\begin{array}{cc} 0 & -i \\ i & 0 \end{array}\right),\ P^{(3)}=\left(\begin{array}{cc} 1 & 0 \\ 0 & -1 \end{array}\right),$$
we define a self-adjoint operator $V_s$ by
$$V_s:=\sum_{\bx\in\G}\<\bB_{L,\bx},\bS_{\bx}\>.$$
Since 
$$V_s=\sum_{\bx\in\G}\sum_{\xi,\phi\in\spin}\left(\frac{1}{2}\sum_{l=1}^3B_{L,\bx}^{(l)}P_{\xi,\phi}^{(l)}\right)\psi_{\bx\xi}^*\psi_{\bx\phi}$$
and $\overline{P_{\xi,\phi}^{(l)}}=P_{\phi,\xi}^{(l)}$ $(\forall l\in\{1,2,3\},\forall \xi,\phi\in\spin)$, the operator $V_s$ provides one example of the interactions of the form \eqref{eq_definition_V}.
\end{example}

\begin{example}[the spin-spin interaction]
With the spin operator $\bS_{\bx}$ introduced in Example \ref{ex_spin_field} and a real function $w(\bx):\Z^d\to\R$, we define the spin-spin interaction $V_{ss}$ by
$$V_{ss}:=\sum_{\bx,\by\in\G}w_L(\bx-\by)\<\bS_{\bx},\bS_{\by}\>,$$
where the coefficient $w_L$ is the restriction of $w$ on $\G$ by periodicity. 
A calculation shows that
\begin{equation*}
\begin{split}
&V_{ss}=\sum_{\bx\in\G}\sum_{\xi,\phi\in\spin}\left(\frac{w_L(\b0)}{4}\sum_{l=1}^3\sum_{\tau\in\spin}P_{\xi,\tau}^{(l)}P_{\tau,\phi}^{(l)}\right)\psi_{\bx\xi}^*\psi_{\bx\phi}\\
&\quad+\sum_{\bx_1,\bx_2\in\G}\sum_{\xi_1,\xi_2,\phi_1,\phi_2\in\spin}\left(\frac{w_L(\bx_1-\bx_2)}{4}\sum_{l=1}^3P_{\xi_1,\phi_1}^{(l)}P_{\xi_2,\phi_2}^{(l)}\right)\psi_{\bx_1\xi_1}^*\psi_{\bx_2\xi_2}^*\psi_{\bx_2\phi_2}\psi_{\bx_1\phi_1}.\end{split}
\end{equation*}
Hence $V_{ss}$ can be written in the form \eqref{eq_definition_V}. Consequently, our $V$ covers the interaction of the form $V_{dd}+V_s+V_{ss}$.
\end{example}
In this paper we treat the Hamiltonian operator $H=H_0+V$.

\subsection{Main results}
We employ norms $\|\cdot\|_{L,l}$ and $\|\cdot\|_l$ to measure the magnitude of the interaction.
\begin{equation*}
\begin{split}
&\|U_{L,1}\|_{L,1}:=\max_{\bx\in\G}\max_{\xi\in\spin}\sum_{\phi\in\spin}|U_{L,1}(\bx,\xi,\phi)|,\\
&\|U_{L,l}\|_{L,l}:=\max_{j\in\{1,\cdots,l\}}\max_{\xi_j\in\spin}\sum_{\bx_k\in\G\atop \forall k\in\{1,\cdots,l-1\}}\sum_{\xi_k\in\spin\atop \forall k\in\{1,\cdots,l\}\backslash \{j\}}\sum_{\phi_k\in\spin\atop \forall k\in\{1,\cdots,l\}}\\
&\qquad\qquad\qquad\cdot|U_{L,l}((\bx_1,\cdots,\bx_{l-1},\b0),(\xi_1,\cdots,\xi_l),(\phi_1,\cdots,\phi_l))|\ (\forall l\ge 2),\\
&\|U_{1}\|_1:=\sup_{\bx\in\Z^d}\max_{\xi\in\spin}\sum_{\phi\in\spin}|U_1(\bx,\xi,\phi)|,\\
&\|U_{l}\|_l:=\max_{j\in\{1,\cdots,l\}}\max_{\xi_j\in\spin}\sum_{\bx_k\in\Z^d\atop \forall k\in\{1,\cdots,l-1\}}\sum_{\xi_k\in\spin\atop \forall k\in\{1,\cdots,l\}\backslash \{j\}}\sum_{\phi_k\in\spin\atop \forall k\in\{1,\cdots,l\}}\\
&\quad\qquad\qquad\cdot|U_{l}((\bx_1,\cdots,\bx_{l-1},\b0),(\xi_1,\cdots,\xi_l),(\phi_1,\cdots,\phi_l))|\ (\forall l\ge 2).
\end{split} 
\end{equation*}

For any operator $\cO$ we define the thermal average $\<\cO\>_L$ by
$$\<\cO\>_L:=\frac{\Tr(e^{-\beta H}\cO)}{\Tr e^{-\beta H}},$$
where the trace operation is taken over the Fermionic Fock space on $\G\times \spin$ and the positive constant $\beta$ is proportional to the inverse of temperature.

Define the function $F_{t,t',d}:\R\to\R_{>0}$ by 
\begin{equation}\label{eq_definition_keyfunction}
F_{t,t',d}(r):=\frac{r}{2(|t|+2(d-1)|t'|)}+\sqrt{\frac{r^2}{4(|t|+2(d-1)|t'|)^2}+1}.
\end{equation}

For $\bx\in\Z^d$ let us define $\psi_{\bx\xi}^{(*)}$ by considering $\bx$ as a site of $\G$ by periodicity.

The main results of this paper are stated as follows.
\begin{theorem}\label{thm_exponential_decay}
Assume that there exists $R\in (0,1)$ such that 
\begin{equation}\label{eq_assumption_U_convergence}
\sum_{l=1}^{\tn}l 16^{l}\|U_l\|_{l}<\beta^{-1} \left(\frac{F_{t,t',d}\left(\frac{\pi}{2\beta}\right)^{1/(2e\pi d)}+1}{F_{t,t',d}\left(\frac{\pi}{2\beta}\right)^{1/(2e\pi d)}-1}\right)^{-d}R.
\end{equation}
For any $\hm\in\N$ and any $\hbx_j,\hby_j\in\Z^d$, $\hxi_j,\hphi_j\in\spin$ $(\forall j\in\{1,\cdots,\hm\})$ the thermodynamic limit 
\begin{equation}\label{eq_thermodynamic_limit}
\begin{split}
\lim_{L\to +\infty\atop L\in\N}&\<\psi_{\hbx_1\hxi_1}^*\psi_{\hbx_2\hxi_2}^*\cdots\psi_{\hbx_{\hm}\hxi_{\hm}}^*\psi_{\hby_{\hm}\hphi_{\hm}}\psi_{\hby_{\hm-1}\hphi_{\hm-1}}\cdots\psi_{\hby_{1}\hphi_{1}}\\
&\quad+\psi_{\hby_{1}\hphi_{1}}^*\psi_{\hby_{2}\hphi_{2}}^*\cdots\psi_{\hby_{\hm}\hphi_{\hm}}^*\psi_{\hbx_{\hm}\hxi_{\hm}}\psi_{\hbx_{\hm-1}\hxi_{\hm-1}}\cdots\psi_{\hbx_1\hxi_1}\>_L
\end{split}
\end{equation}
exists and satisfies the following inequality.
\begin{equation}\label{eq_main_exponential_decay}
\begin{split}
\lim_{L\to +\infty\atop L\in\N}&\Big|\<\psi_{\hbx_1\hxi_1}^*\psi_{\hbx_2\hxi_2}^*\cdots\psi_{\hbx_{\hm}\hxi_{\hm}}^*\psi_{\hby_{\hm}\hphi_{\hm}}\psi_{\hby_{\hm-1}\hphi_{\hm-1}}\cdots\psi_{\hby_{1}\hphi_{1}}\\
&\quad+\psi_{\hby_{1}\hphi_{1}}^*\psi_{\hby_{2}\hphi_{2}}^*\cdots\psi_{\hby_{\hm}\hphi_{\hm}}^*\psi_{\hbx_{\hm}\hxi_{\hm}}\psi_{\hbx_{\hm-1}\hxi_{\hm-1}}\cdots\psi_{\hbx_1\hxi_1}\>_L\Big|\\
&\le (4^{\hm+1}-\hm 4^{2\hm+1}\log(1-R))\cdot F_{t,t',d}\left(\frac{\pi}{2\beta}\right)^{-\frac{1}{4ed}\|\sum_{j=1}^{\hm}\hbx_j-\sum_{j=1}^{\hm}\hby_j\|_{\R^d}}.
\end{split}
\end{equation}
\end{theorem}

In the case that the interaction $V$ is the on-site interaction \\
$U\sum_{\bx\in\G}\psi_{\bx\ua}^*\psi_{\bx\da}^*\psi_{\bx\da}\psi_{\bx\ua}$ $(U\in\R)$, 
the 4 point correlation functions can be bounded more sharply as follows.
\begin{theorem}\label{thm_exponential_decay_hubbard_model}
Assume that the Hamiltonian operator $H$ is given by
$$H=H_0+U\sum_{\bx\in\G}\psi_{\bx\ua}^*\psi_{\bx\da}^*\psi_{\bx\da}\psi_{\bx\ua}$$
with the coupling constant $U\in\R$ satisfying
\begin{equation}\label{eq_assumption_U_convergence_hubbard_model}
|U|\le (108\beta)^{-1} \left(\frac{F_{t,t',d}\left(\frac{\pi}{2\beta}\right)^{1/(2e\pi d)}+1}{F_{t,t',d}\left(\frac{\pi}{2\beta}\right)^{1/(2e\pi d)}-1}\right)^{-d}.
\end{equation}
For any $\hbx_1,\hbx_2,\hby_1,\hby_2\in\Z^d$ the thermodynamic limit 
\begin{equation}\label{eq_thermodynamic_limit_hubbard}
\lim_{L\to +\infty\atop L\in\N}\<\psi_{\hbx_1\ua}^*\psi_{\hbx_2\da}^*\psi_{\hby_2\da}\psi_{\hby_1\ua}+\psi_{\hby_1\ua}^*\psi_{\hby_2\da}^*\psi_{\hbx_2\da}\psi_{\hbx_1\ua}\>_L
\end{equation}
 exists and satisfies
\begin{equation*}
\begin{split}
\lim_{L\to +\infty\atop L\in\N}&|\<\psi_{\hbx_1\ua}^*\psi_{\hbx_2\da}^*\psi_{\hby_2\da}\psi_{\hby_1\ua}+\psi_{\hby_1\ua}^*\psi_{\hby_2\da}^*\psi_{\hbx_2\da}\psi_{\hbx_1\ua}\>_L|\\
&\quad \le 324\cdot F_{t,t',d}\left(\frac{\pi}{2\beta}\right)^{-\frac{1}{4ed}\|\hbx_1+\hbx_2-\hby_1-\hby_2\|_{\R^d}}.
\end{split}
\end{equation*}
\end{theorem}
Theorem \ref{thm_exponential_decay} and Theorem \ref{thm_exponential_decay_hubbard_model} will be proved in Section \ref{sec_exponential_decay_correlation}.

\begin{remark}
Theorem \ref{thm_exponential_decay_hubbard_model} does not follow Theorem \ref{thm_exponential_decay}. To see this, we write
\begin{equation*}
\begin{split}
&U\sum_{\bx\in\G}\psi_{\bx\ua}^*\psi_{\bx\da}^*\psi_{\bx\da}\psi_{\bx\ua}\\
&\quad=\sum_{\bx_j,\by_j\in\G\atop \forall j\in\{1,2\}}\sum_{\xi_j,\phi_j\in\spin\atop \forall j\in\{1,2\}}f_{c}((\bx_1,\xi_1),(\bx_2,\xi_2),(\by_1,\phi_1),(\by_2,\phi_2))\psi_{\bx_1\xi_1}^*\psi_{\bx_2\xi_2}^*\psi_{\by_1\phi_1}\psi_{\by_2\phi_2}
\end{split}
\end{equation*}
with 
\begin{equation*}
\begin{split}
&f_c((\bx_1,\xi_1),(\bx_2,\xi_2),(\by_1,\phi_1),(\by_2,\phi_2))\\
&\quad:=-\frac{U}{4}\delta_{\bx_1,\bx_2}\delta_{\by_1,\by_2}\delta_{\bx_1,\by_1}(\delta_{\xi_1,\ua}\delta_{\xi_2,\da}-\delta_{\xi_1,\da}\delta_{\xi_2,\ua})(\delta_{\phi_1,\ua}\delta_{\phi_2,\da}-\delta_{\phi_1,\da}\delta_{\phi_2,\ua}).
\end{split}
\end{equation*}
The function $f_c $ satisfies the anti-symmetricity \eqref{eq_anti_symmetricity} and 
\begin{equation*}
\begin{split}
\max\Bigg\{&\max_{(\bx_1,\xi_1)\in\G\times\spin}\sum_{(\bx_2,\xi_2),(\by_j,\phi_j)\in\G\times\spin\atop \forall j\in\{1,2\}}|f_{c}((\bx_1,\xi_1),(\bx_2,\xi_2),(\by_1,\phi_1),(\by_2,\phi_2))|,\\
&\max_{(\by_1,\phi_1)\in\G\times\spin}\sum_{(\by_2,\phi_2),(\bx_j,\xi_j)\in\G\times\spin\atop \forall j\in\{1,2\}}|f_{c}((\bx_1,\xi_1),(\bx_2,\xi_2),(\by_1,\phi_1),(\by_2,\phi_2))|\Bigg\}\\
&=\frac{|U|}{2}.
\end{split}
\end{equation*}
Thus, Lemma \ref{lem_anti_symmetrization} proved in Appendix \ref{app_anti_symmetricity} ensures that if a function $U_c:(\Z^d)^2\times \spin^4\to \C$ satisfies \eqref{eq_condition_U}-\eqref{eq_translation_invariance} and its restriction $U_{c,L}:\G^2\times \spin^4\to \C$ by periodicity obeys
\begin{equation}\label{eq_interaction_hubbard_rewritten}
\begin{split}
&U\sum_{\bx\in\G}\psi_{\bx\ua}^*\psi_{\bx\da}^*\psi_{\bx\da}\psi_{\bx\ua}\\
&\quad=\sum_{\bx_1,\bx_2\in\G}\sum_{\xi_1,\xi_2,\phi_1,\phi_2\in\spin}U_{c,L}((\bx_1,\bx_2),(\xi_1,\xi_2),(\phi_1,\phi_2))\psi_{\bx_1\xi_1}^*\psi_{\bx_2\xi_2}^*\psi_{\bx_2\phi_2}\psi_{\bx_1\phi_1},
\end{split}
\end{equation}
the inequalities 
\begin{equation}\label{eq_norm_inequality}
\frac{|U|}{2}\le\|U_{c,L}\|_{L,2}\le \|U_{c}\|_{2}
\end{equation}
must hold. 

 Let us apply Theorem \ref{thm_exponential_decay} to the interaction \eqref{eq_interaction_hubbard_rewritten}. The inequalities \eqref{eq_assumption_U_convergence} and \eqref{eq_norm_inequality} imply 
$$
|U|<(256\beta)^{-1} \left(\frac{F_{t,t',d}\left(\frac{\pi}{2\beta}\right)^{1/(2e\pi d)}+1}{F_{t,t',d}\left(\frac{\pi}{2\beta}\right)^{1/(2e\pi d)}-1}\right)^{-d}R,
$$
which is a stricter constraint than \eqref{eq_assumption_U_convergence_hubbard_model}. Moreover, for $R$ close to $1$, the coefficient $4^3-2\cdot4^5\log(1-R)$ is larger than $324$.
\end{remark}

\begin{remark}\label{rem_beta_dependency}One can prove that for any $b>0$ there exist constants $C_{t,t',d,b}>0$ depending only on $t,t',d, b$ and $C_{t,t',d}'>0$ depending only on $t,t',d$ such that
$$C_{t,t',d,b}\beta^{-d-1}\le \beta^{-1} \left(\frac{F_{t,t',d}\left(\frac{\pi}{2\beta}\right)^{1/(2e\pi d)}+1}{F_{t,t',d}\left(\frac{\pi}{2\beta}\right)^{1/(2e\pi d)}-1}\right)^{-d}\le C_{t,t',d}'\beta^{-d-1}$$
for any $\beta \ge b$. Hence, the interaction needs to be small to claim the decay bounds in low temperatures.
\end{remark}

\begin{remark}
To generalize the results to many-Fermion systems with finite coordinates of colors is straightforward. We present the results only for the spins $\spin$ in order to refer to proved materials for many-electron systems summarized in \cite{K}.
\end{remark}

\begin{remark}
We use the translation invariance \eqref{eq_translation_invariance} to prove the existence of the thermodynamic limit \eqref{eq_thermodynamic_limit} in Lemma \ref{lem_thermodynamic_limit}. Without assuming \eqref{eq_translation_invariance} we can also prove the inequality of the form \eqref{eq_main_exponential_decay} with $\limsup_{L\to +\infty, L\in\N}$ in the left side instead of $\lim_{L\to +\infty,L\in\N}$ under an appropriate modification of the norm of $U_l$.
\end{remark}

\begin{remark}\label{rem_H0_trivialcase}
If $|t|+|t'|1_{d\ge 2}=0$, the correlation functions decay trivially. To prove this, let us take any real parameters $\{\theta_{\bx}\}_{\bx\in\G}$ and define the unitary operator $A_{\theta}$ by
$$A_{\theta}:=\prod_{(\bx,\xi)\in\G\times\spin}e^{i\theta_{\bx}\psi_{\bx\xi}^*\psi_{\bx\xi}}.$$
In this case, $A_{\theta}HA_{\theta}^*=H$ and thus
\begin{equation}\label{eq_remark_unitary}
\begin{split}
&\<\psi_{\hbx_1\hxi_1}^*\cdots\psi_{\hbx_{\hm}\hxi_{\hm}}^*\psi_{\hby_{\hm}\hphi_{\hm}}\cdots\psi_{\hby_{1}\hphi_{1}}+\psi_{\hby_{1}\hphi_{1}}^*\cdots\psi_{\hby_{\hm}\hphi_{\hm}}^*\psi_{\hbx_{\hm}\hxi_{\hm}}\cdots\psi_{\hbx_{1}\hxi_{1}}\>_L\\
&=\<A_{\theta}\psi_{\hbx_1\hxi_1}^*\cdots\psi_{\hbx_{\hm}\hxi_{\hm}}^*\psi_{\hby_{\hm}\hphi_{\hm}}\cdots\psi_{\hby_{1}\hphi_{1}}A_{\theta}^*\>_L\\
&\qquad\qquad+\<A_{\theta}^*\psi_{\hby_{1}\hphi_{1}}^*\cdots\psi_{\hby_{\hm}\hphi_{\hm}}^*\psi_{\hbx_{\hm}\hxi_{\hm}}\cdots\psi_{\hbx_{1}\hxi_{1}}A_{\theta}\>_L\\
&= e^{i\sum_{j=1}^{\hm}(\theta_{\hbx_j}-\theta_{\hby_j})}\\
&\qquad\cdot\<\psi_{\hbx_1\hxi_1}^*\cdots\psi_{\hbx_{\hm}\hxi_{\hm}}^*\psi_{\hby_{\hm}\hphi_{\hm}}\cdots\psi_{\hby_{1}\hphi_{1}}+\psi_{\hby_{1}\hphi_{1}}^*\cdots\psi_{\hby_{\hm}\hphi_{\hm}}^*\psi_{\hbx_{\hm}\hxi_{\hm}}\cdots\psi_{\hbx_{1}\hxi_{1}}\>_L.
\end{split}
\end{equation}
If $\sum_{j=1}^{\hm}\hbx_j\neq\sum_{j=1}^{\hm}\hby_j$ in $\G$, we can choose $\{\theta_{\bx}\}_{\bx\in\G}$ to satisfy $\sum_{j=1}^{\hm}\theta_{\hbx_j}-\sum_{j=1}^{\hm}\theta_{\hby_j}\neq 0$ in $\R/(2\pi\Z)$ and the equality \eqref{eq_remark_unitary} implies that 
\begin{equation*}
\<\psi_{\hbx_1\hxi_1}^*\cdots\psi_{\hbx_{\hm}\hxi_{\hm}}^*\psi_{\hby_{\hm}\hphi_{\hm}}\cdots\psi_{\hby_{1}\hphi_{1}}+\psi_{\hby_{1}\hphi_{1}}^*\cdots\psi_{\hby_{\hm}\hphi_{\hm}}^*\psi_{\hbx_{\hm}\hxi_{\hm}}\cdots\psi_{\hbx_{1}\hxi_{1}}\>_L=0.
\end{equation*}
\end{remark}

\section{Grassmann integral formulation of the correlation functions}\label{sec_formulation}
In this section we formulate the correlation functions as a limit of finite dimensional Grassmann integrals. To attain this goal, we follow steps. As a preliminary let us fix the way to abbreviate the notations.

\subsection{The correlation functions}\label{subsec_correlation_functions}
 To simplify presentations, from now we write $X^l=(\bx_1,\cdots,\bx_l),Y^l=(\by_1,\cdots,\by_l)\in(\Z^d)^l$, $\Xi^l=(\xi_1,\cdots,\xi_l),\Phi^l=(\phi_1,\cdots,\phi_l)\in\spin^l$ $(\forall l\in\N)$. To indicate the sites and the spins on which the correlation functions are defined, we use the notation $\hat{\cdot}$ and write 
\begin{equation*}
\begin{split}
&\hX^{\hm}=(\hbx_1,\cdots,\hbx_{\hm}),\ \hY^{\hm}=(\hby_1,\cdots,\hby_{\hm})\in(\Z^d)^{\hm},\\
&\hXi^{\hm}=(\hxi_1,\cdots,\hxi_{\hm}),\ \hPhi^{\hm}=(\hphi_1,\cdots,\hphi_{\hm})\in\spin^{\hm}.
\end{split}
\end{equation*}

We identify $X^l\in(\Z^d)^l$ as an element of $\G^l$ by periodicity without remarking when we are considering a problem defined on $\G^l$.

To derive the correlation functions systematically, we introduce real parameters \\
$\{\lambda(X^{\hm},Y^{\hm},\Xi^{\hm},\Phi^{\hm})\}_{X^{\hm},Y^{\hm}\in\G^{\hm},\Xi^{\hm},\Phi^{\hm}\in\spin^{\hm}}$ and define the coefficient \\
$U_{\lambda,l}(X^l,Y^l,\Xi^l,\Phi^l)$ by
\begin{equation*}
\begin{split}
U_{\lambda,l}(X^l,Y^l,\Xi^l,\Phi^l):=&1_{l\le \tn}U_{L,l}(X^l, \Xi^l,\Phi^l)\delta_{X_l,Y_l} \\
&+ 1_{l=\hm}(\lambda(X^{\hm},Y^{\hm},\Xi^{\hm},\Phi^{\hm})+\lambda(Y^{\hm},X^{\hm},\Phi^{\hm},\Xi^{\hm}))
\end{split}
\end{equation*}
for all $X^{l},Y^{l}\in\G^{l},\ \Xi^{l},\Phi^{l}\in\spin^{l}$, $l\in \{1,\cdots,\max\{\hm,\tn\}\}$. We see that
\begin{equation}\label{eq_condition_U_lambda}
\overline{U_{\lambda,l}(X^l,Y^l,\Xi^l,\Phi^l)}=U_{\lambda,l}(Y^l,X^l,\Phi^l,\Xi^l).
\end{equation}

Let us modify the interaction $V$ to contain the coefficients $\{U_{\lambda,l}(X^l,Y^l,\Xi^l,\Phi^l)\}$ and define
\begin{equation*}
\begin{split}
V_{\lambda}:=&\sum_{l=1}^{\max\{\hm,\tn\}}\sum_{(X^l,Y^l,\Xi^l,\Phi^l)\in\G^{2l}\times\spin^{2l}}U_{\lambda,l}(X^l,Y^l,\Xi^l,\Phi^l)\\
&\cdot\psi_{\bx_{1}\xi_{1}}^*\cdots\psi_{\bx_{l}\xi_{l}}^*\psi_{\by_{l}\phi_{l}}\cdots\psi_{\by_{1}\phi_{1}}.
\end{split}
\end{equation*}
We set $H_{\lambda}:=H_0+V_{\lambda}$, which is self-adjoint by the equality \eqref{eq_condition_U_lambda}. Also note that
\begin{equation*}
H_{\lambda}\Big|_{\lambda(X^{\hm},Y^{\hm},\Xi^{\hm},\Phi^{\hm})=0\atop \forall (X^{\hm},Y^{\hm}, \Xi^{\hm},\Phi^{\hm})\in\G^{2\hm}\times\spin^{2\hm}}=H.
\end{equation*}
\begin{lemma}\label{lem_derivative_correlation}
\begin{equation}\label{eq_derivative_correlation}
\begin{split}
&\<\psi_{\hbx_1\hxi_1}^*\cdots\psi_{\hbx_{\hm}\hxi_{\hm}}^*\psi_{\hby_{\hm}\hphi_{\hm}}\cdots\psi_{\hby_{1}\hphi_{1}}+\psi_{\hby_{1}\hphi_{1}}^*\cdots\psi_{\hby_{\hm}\hphi_{\hm}}^*\psi_{\hbx_{\hm}\hxi_{\hm}}\cdots\psi_{\hbx_1\hxi_1}\>_L\\
&\qquad=-\frac{1}{\beta}\frac{\partial}{\partial \lambda(\hX^{\hm},\hY^{\hm},\hXi^{\hm},\hPhi^{\hm})}\log\left(\frac{\Tr e^{-\beta H_{\lambda}}}{\Tr e^{-\beta H_0}}\right)\Bigg|_{\lambda(X^{\hm},Y^{\hm},\Xi^{\hm},\Phi^{\hm})=0\atop \forall (X^{\hm},Y^{\hm}, \Xi^{\hm},\Phi^{\hm})\in\G^{2\hm}\times\spin^{2\hm}}.
\end{split}
\end{equation}
\end{lemma}
\begin{proof} The proof is parallel to that of \cite[\mbox{Lemma 2.1}]{K}, based on \cite[\mbox{Lemma 2.3}]{K}.
\end{proof}

\subsection{The perturbation series}
In order to characterize the correlation functions as a limit of finite dimensional Grassmann integrals, we proceed in the following steps. Firstly we expand the partition function $\Tr e^{-\beta H_{\lambda}}/\Tr e^{-\beta H_0}$ into a perturbation series of the variables \\
$\{U_{\lambda,l}(X^l,Y^l,\Xi^l,\Phi^l)\}$. Secondly we replace the integrals over $[0,\beta)$ contained in the perturbation series by the Riemann sums to derive a fully discrete analog of the perturbation series. We then show that the discretized perturbation series converges to the original one by passing the parameter defining the Riemann sums to infinity. The discretized perturbation series is formulated into the Grassmann Gaussian integral involving only finite Grassmann variables. Combining the Grassmann integral formulation of the discretized partition function with the equality \eqref{eq_derivative_correlation} completes the characterization. 

The first step results in the following proposition.
\begin{proposition}\label{prop_perturbation_series}
\begin{equation}\label{eq_perturbation_series}
\begin{split}
&\frac{\Tr e^{-\beta H_{\lambda}}}{\Tr e^{-\beta H_0}}=1+\sum_{m=1}^{\infty}\frac{(-1)^m}{m!}\\
&\cdot\prod_{k=1}^m\left(\sum_{l_k=1}^{\max\{\hm,\tn\}}\sum_{(X_k^{l_k},Y_k^{l_k}, \Xi_k^{l_k},\Phi_k^{l_k})\in\G^{2l_k}\times\spin^{2l_k}}U_{\lambda,l_k}(X_k^{l_k},Y_k^{l_k}, \Xi_k^{l_k},\Phi_k^{l_k})\int_0^{\beta}ds_k\right)\\
&\cdot\det(C((\widetilde{\bx\xi s})_p,(\widetilde{\by\phi s})_q))_{1\le p,q\le \sum_{k=1}^ml_k},
\end{split}
\end{equation}
where $X_k^{l_k}:=(\bx_{k,1},\bx_{k,2},\cdots,\bx_{k,l_k})$, $Y_k^{l_k}:=(\by_{k,1},\by_{k,2},\cdots,\by_{k,l_k})$,\\
$\Xi_k^{l_k}:=(\xi_{k,1},\xi_{k,2},\cdots,\xi_{k,l_k})$, $\Phi_k^{l_k}:=(\phi_{k,1},\phi_{k,2},\cdots,\phi_{k,l_k})$, and 
\begin{equation}\label{eq_covariance_argument}
(\widetilde{\bx\xi s})_p:=(\bx_{v+1,u},\xi_{v+1,u},s_{v+1}),\ (\widetilde{\by\phi s})_p:=(\by_{v+1,u},\phi_{v+1,u},s_{v+1})
\end{equation}
for $p=\sum_{k=1}^vl_k+u$, $u\in\{1,\cdots,l_{v+1}\}$, $v\in\{0,\cdots,m-1\}$.
The covariance $C(\bx\xi x,\by \phi y)$ is given by
\begin{equation}\label{eq_covariance}
C(\bx\xi x,\by\phi y):=\frac{\delta_{\xi,\phi}}{L^d}\sum_{\bk\in\G^*}e^{i\<\bk,\by-\bx\>}e^{-(y-x)E_{\bk}}\left(\frac{1_{y-x\le
0}}{1+e^{\beta E_{\bk}}}-\frac{1_{y-x>
0}}{1+e^{-\beta E_{\bk}}}\right)
\end{equation}
$(\forall (\bx,\xi, x),(\by, \phi, y)\in\G\times\spin\times [0,\beta))$ with the dispersion relation
\begin{equation}\label{eq_dispersion_relation}
E_{\bk}:=-2t\sum_{j=1}^d\cos(\<\bk,\be_j\>) - 4t'\cdot 1_{d\ge 2}\sum_{j,k = 1\atop j<k}^d\cos(\<\bk,\be_j\>)\cos(\<\bk,\be_k\>)-\mu.
\end{equation}
\end{proposition}

\begin{proof}
For any operator $\cO$ defined on the Fermionic Fock space, let $\cO(s)$ denote $e^{sH_0}\cO e^{-sH_0}$ for $s\in\R$ and $\<\cO\>_{0,L}$ denote $\Tr (e^{-\beta H_0}\cO)/\Tr e^{-\beta H_0}$. We can apply \cite[\mbox{Lemma B.3}]{K} to derive the equality
$$e^{-\beta H_{\lambda}}=e^{-\beta H_0}+e^{-\beta H_0}\sum_{m=1}^{\infty}(-1)^m\int_{[0,\beta)^m}ds_1\cdots ds_m1_{s_1>\cdots> s_m}V_{\lambda}(s_1)\cdots V_{\lambda}(s_m),$$
which leads to 
\begin{equation}\label{eq_partition_expansion}
\begin{split}
&\frac{\Tr e^{-\beta H_{\lambda}}}{\Tr e^{-\beta H_0}}=1+\sum_{m=1}^{\infty}(-1)^m\int_{[0,\beta)^m}ds_1\cdots ds_m1_{s_1>\cdots>s_m}\\
&\qquad\cdot\prod_{k=1}^m\left(\sum_{l_k=1}^{\max\{\hm,\tn\}}\sum_{(X_k^{l_k},Y_k^{l_k}, \Xi_k^{l_k},\Phi_k^{l_k})\in\G^{2l_k}\times\spin^{2l_k}}U_{\lambda,l_k}(X_k^{l_k},Y_k^{l_k}, \Xi_k^{l_k},\Phi_k^{l_k})\right)\\
&\qquad\cdot\<\psi_{\bx_{1,1}\xi_{1,1}}^*(s_1)\cdots\psi_{\bx_{1,l_1}\xi_{1,l_1}}^*(s_1)\psi_{\by_{1,l_1}\phi_{1,l_1}}(s_1)\cdots \psi_{\by_{1,1}\phi_{1,1}}(s_1)\\
&\qquad\quad \cdots \psi_{\bx_{m,1}\xi_{m,1}}^*(s_m)\cdots\psi_{\bx_{m,l_m}\xi_{m,l_m}}^*(s_m)\psi_{\by_{m,l_m}\phi_{m,l_m}}(s_m)\cdots \psi_{\by_{m,1}\phi_{m,1}}(s_m)\>_{0,L}.
\end{split}
\end{equation}
By recalling the definition \cite[\mbox{Definition B.2}]{K} of the temperature-ordering operator and its properties \cite[\mbox{Lemma B.7, Lemma B.9}]{K}, we can deduce \eqref{eq_perturbation_series} from \eqref{eq_partition_expansion}. The covariance $C(\bx\xi x,\by \phi y)$ and the dispersion relation $E_{\bk}$ have been proved to have the forms \eqref{eq_covariance}-\eqref{eq_dispersion_relation} in \cite[\mbox{Lemma B.10}]{K}.
\end{proof}

As the second step toward the Grassmann integral formulation, we introduce the discrete analog of the expansion \eqref{eq_perturbation_series}. By taking a parameter $h\in 2\N/\beta$, we define the discrete sets $[0,\beta)_h$ and $[-\beta,\beta)_h$ by
$$[0,\beta)_h:=\left\{0,\frac{1}{h},\frac{2}{h},\cdots,\beta - \frac{1}{h}\right\},\ [-\beta,\beta)_h:=\left\{-\beta,-\beta+\frac{1}{h},\cdots,\beta - \frac{1}{h}\right\}.
$$  
Note that $\sharp[0,\beta)_h=\beta h$, $\sharp[-\beta,\beta)_h=2\beta h$. If the temperature variables $x,y$ are confined in $[0,\beta)_h$, we write 
$$C_h(\bx\xi x,\by \phi y)=C(\bx\xi x,\by \phi y).$$
We then define $(\Tr e^{-\beta H_{\lambda}}/\Tr e^{-\beta H_0})_h$ by
\begin{equation}\label{eq_discrete_perturbation_series}
\begin{split}
&\left(\frac{\Tr e^{-\beta H_{\lambda}}}{\Tr e^{-\beta H_0}}\right)_h:=1+\sum_{m=1}^{2L^d\beta h}\frac{(-1)^m}{m!}\\
&\cdot\prod_{k=1}^m\left(\sum_{l_k=1}^{\max\{\hm,\tn\}}\sum_{(X_k^{l_k},Y_k^{l_k}, \Xi_k^{l_k},\Phi_k^{l_k})\in\G^{2l_k}\times\spin^{2l_k}}U_{\lambda,l_k}(X_k^{l_k},Y_k^{l_k}, \Xi_k^{l_k},\Phi_k^{l_k})\frac{1}{h}\sum_{s_k\in [0,\beta)_h}\right)\\
&\qquad \cdot\det(C_h((\widetilde{\bx\xi s})_p,(\widetilde{\by\phi s})_q))_{1\le p,q\le \sum_{k=1}^ml_k},
\end{split}
\end{equation}
where the variables $(\widetilde{\bx\xi s})_p,(\widetilde{\by\phi s})_p\in\G\times\spin\times[0,\beta)_h$ are defined by the same rule as \eqref{eq_covariance_argument}. Since any determinant made up of the elements $C_h(\bx\xi x, \by\phi y)$ vanishes if the size of the matrix exceeds $2L^d\beta h (=\sharp \G\times\spin\times [0,\beta)_h)$, the sum with respect to $m$ is taken only up to $m=2L^d\beta h$ in \eqref{eq_discrete_perturbation_series}.

\begin{remark}
The diagonalization of the covariance matrix \\
$(C_h(\bx \xi x,\by \phi y))_{(\bx,\xi,x),(\by,\phi,y)\in\G\times\spin \times [0,\beta)_h}$ was presented in \cite[\mbox{Appendix C}]{K} for any $h\in 2\N/\beta$. To refer to this result we take $h$ from $2\N/\beta$.
\end{remark} 

The following lemma states that the partition function $\Tr e^{-\beta H_{\lambda}}/\Tr e^{-\beta H_0}$ in Lemma \ref{lem_derivative_correlation} can be replaced by $(\Tr e^{-\beta H_{\lambda}}/\Tr e^{-\beta H_0})_h$.

\begin{lemma}\label{lem_correlation_h_limit}
For any $r>0$ there exists $N_0\in\N$ such that $\Re (\Tr
 e^{-\beta H_{\lambda}}/\Tr e^{-\beta H_0})_h$\\
$>0$ for all $h\in
 2\N/\beta$ with $h\ge 2N_0/\beta$ and all $\lambda(X^{\hm},Y^{\hm},\Xi^{\hm},\Phi^{\hm})\in \R$ with \\
$|\lambda(X^{\hm},Y^{\hm},\Xi^{\hm},\Phi^{\hm})|\le r$. Moreover,
\begin{equation}\label{eq_correlation_h_limit}
\begin{split}
&\<\psi_{\hbx_1\hxi_1}^*\cdots\psi_{\hbx_{\hm}\hxi_{\hm}}^*\psi_{\hby_{\hm}\hphi_{\hm}}\cdots\psi_{\hby_{1}\hphi_{1}}+\psi_{\hby_{1}\hphi_{1}}^*\cdots\psi_{\hby_{\hm}\hphi_{\hm}}^*\psi_{\hbx_{\hm}\hxi_{\hm}}\cdots\psi_{\hbx_1\hxi_1}\>_L\\
&=-\frac{1}{\beta}\lim_{h\to +\infty\atop h\in 2\N/\beta}\frac{\partial}{\partial \lambda(\hX^{\hm},\hY^{\hm},\hXi^{\hm},\hPhi^{\hm})}\log\left(\frac{\Tr e^{-\beta H_{\lambda}}}{\Tr e^{-\beta H_0}}\right)_h\Bigg|_{\lambda(X^{\hm},Y^{\hm},\Xi^{\hm},\Phi^{\hm})=0\atop\forall (X^{\hm},Y^{\hm}, \Xi^{\hm},\Phi^{\hm})\in\G^{2\hm}\times\spin^{2\hm}},
\end{split}
\end{equation}
where for $z\in\C$ with $\Re z>0$, $\log z$ is defined by taking the principal value; 
$$\log z:= \log|z| +i\Arg z,\ \Arg z\in (-\pi/2,\pi/2).$$
\end{lemma}
\begin{proof}
By considering
 $\{U_{\lambda,l}(X^l,Y^l,\Xi^l,\Phi^l)\}_{l\in\{1,\cdots,\max\{\hm,\tn\}\}, (X^l,Y^l,\Xi^l,\Phi^l)\in\G^{2l}\times\spin^{2l}}$ as mutually independent, complex multi-variables, we define the functions \\
$P(\{U_{\lambda,l}(X^l,Y^l,\Xi^l,\Phi^l)\})$ and 
$P_h(\{U_{\lambda,l}(X^l,Y^l,\Xi^l,\Phi^l)\})$ by the right hand side of
 \eqref{eq_perturbation_series} and
 \eqref{eq_discrete_perturbation_series}, respectively. We show that
 $P_h$ converges to $P$ as $h\to +\infty$ locally uniformly with respect
 to the variables.
 
Pedra-Salmhofer's determinant bound \cite[\mbox{Theorem 2.4}]{PS} (see also Proposition \ref{prop_extended_determinant_bound} below for an extended statement) implies that
\begin{equation}\label{eq_determinant_bound_appl}
|\det(C(\bx_p\xi_ps_p,\by_q\phi_qt_q))_{1\le p,q\le n}|\le 4^{n}
\end{equation}
for any $(\bx_p,\xi_p,s_p),(\by_p,\phi_p,t_p)\in \G\times\spin \times
 [0,\beta)$ $(\forall p\in\{1,\cdots,n\})$. Thus, if all the variables satisfy the inequality $|U_{\lambda,l}(X^l,Y^l,\Xi^l,\Phi^l)|\le r$ for an $r>0$, we have\begin{equation}\label{eq_perturbation_bound}
\begin{split}
\Bigg|&\frac{(-1)^m}{m!}\prod_{k=1}^m\Bigg(\sum_{l_k=1}^{\max\{\hm,\tn\}}\sum_{(X_k^{l_k},Y_k^{l_k}, \Xi_k^{l_k},\Phi_k^{l_k})\in\G^{2l_k}\times\spin^{2l_k}}U_{\lambda,l_k}(X_k^{l_k},Y_k^{l_k}, \Xi_k^{l_k},\Phi_k^{l_k})\Bigg)\\
&\quad\cdot\Bigg(\prod_{k=1}^m\int_0^{\beta}dt_k\det(C((\widetilde{\bx\xi t})_p,(\widetilde{\by\phi t})_q))_{1\le p,q\le \sum_{k=1}^ml_k}\\
&\quad\qquad-\prod_{k=1}^m\frac{1}{h}\sum_{s_k\in [0,\beta)_h}\det(C_h((\widetilde{\bx\xi s})_p,(\widetilde{\by\phi s})_q))_{1\le p,q\le \sum_{k=1}^ml_k}\Bigg)\Bigg|\\
&\le\frac{2}{m!}\prod_{k=1}^{m}\left(\beta r\sum_{l_k=1}^{\max\{\hm,\tn\}}(2L^d)^{2l_k}\right) 4^{\sum_{k=1}^ml_k},
\end{split}
\end{equation}
where $(\widetilde{\bx\xi t})_p,(\widetilde{\by\phi t})_q,(\widetilde{\bx\xi
 s})_p,(\widetilde{\by\phi s})_q$ are defined in the same way as in
 \eqref{eq_covariance_argument}. The right hand side of \eqref{eq_perturbation_bound} is summable over $m\in\N$.

Fix any $(X_k^{l_k},Y_k^{l_k}, \Xi_k^{l_k},\Phi_k^{l_k})\in\G^{2l_k}\times\spin^{2l_k}$ $(\forall k\in\{1,\cdots,m\})$. We define a function $g_h:[0,\beta)^m\to \C$ by
$$g_h(t_1,\cdots,t_m):=\det(C_h((\widetilde{\bx\xi s})_p,(\widetilde{\by\phi s})_q))_{1\le p,q\le \sum_{k=1}^ml_k},$$
where $s_j\in [0,\beta)_h$ satisfies $t_j\in [s_j,s_j+1/h)$ $(\forall j\in\{1,\cdots,m\})$ and
 $(\widetilde{\bx\xi s})_p$, $(\widetilde{\by\phi s})_q$ \\
$(p,q\in \{1,\cdots,\sum_{k=1}^ml_k\})$ are defined by the same rule as \eqref{eq_covariance_argument}.

Since $C(\bx\xi s,\by \phi t)$ is continuous with respect to the variables
 $(s,t)\in[0,\beta)^2$ almost everywhere, so is $\det(C((\widetilde{\bx\xi t})_p,(\widetilde{\by\phi t})_q))_{1\le p,q\le \sum_{k=1}^ml_k}$ with respect to \\
$(t_1,\cdots,t_m)\in[0,\beta)^m$, and thus  
$$\lim_{h\to+\infty\atop h\in 2\N/\beta}g_h(t_1,\cdots,t_m)=\det(C((\widetilde{\bx\xi t})_p,(\widetilde{\by\phi t})_q))_{1\le p,q\le \sum_{k=1}^ml_k}$$
for a.e. $(t_1,\cdots,t_m)\in [0,\beta)^m$. Hence, by \eqref{eq_determinant_bound_appl} and the dominated convergence theorem for $L^1([0,\beta)^m)$, we have
\begin{equation}\label{eq_convergence_determinant}
\begin{split}
&\lim_{h\to +\infty\atop h\in 2\N/\beta}\Bigg|\prod_{k=1}^m\int_0^{\beta}dt_k\det(C((\widetilde{\bx\xi t})_p,(\widetilde{\by\phi t})_q))_{1\le p,q\le \sum_{k=1}^ml_k}\\
&\qquad\qquad\qquad-\prod_{k=1}^m\frac{1}{h}\sum_{s_k\in [0,\beta)_h}\det(C_h((\widetilde{\bx\xi s})_p,(\widetilde{\by\phi s})_q))_{1\le p,q\le \sum_{k=1}^ml_k}\Bigg|\\
&=\lim_{h\to +\infty\atop h\in 2\N/\beta}\Bigg|\prod_{k=1}^m\int_0^{\beta}dt_k\left(\det(C((\widetilde{\bx\xi t})_p,(\widetilde{\by\phi t})_q))_{1\le p,q\le \sum_{k=1}^ml_k}-g_h(t_1,\cdots,t_m)\right)\Bigg|\\
&=0.
\end{split}
\end{equation}

By \eqref{eq_perturbation_bound} and \eqref{eq_convergence_determinant}
 we can again apply the dominated convergence theorem for $l^1(\N)$ to show that\begin{equation}\label{eq_convergence_partition_function}
\begin{split}
\lim_{h\to +\infty\atop h\in 2\N/\beta}\sup_{\forall
 U_{\lambda,l}(X^l,Y^l,\Xi^l,\Phi^l)\in\C\atop \text{with }|U_{\lambda,l}(X^l,Y^l, \Xi^l,\Phi^l)|\le r}\Bigg|&P\left(\{U_{\lambda,l}(X^l,Y^l, \Xi^l,\Phi^l)\}\right)\\
&-P_h\left(\{U_{\lambda,l}(X^l,Y^l, \Xi^l,\Phi^l)\}\right)\Bigg|=0
\end{split}
\end{equation}
for all $r>0$. 

Since the multi-variable function $P-P_h$ is entirely
 analytic, the convergence property
 \eqref{eq_convergence_partition_function} and Cauchy's integral formula
 prove that  
\begin{equation}\label{eq_convergence_derivative}
\begin{split}
\lim_{h\to +\infty\atop h\in 2\N/\beta}&\sup_{\forall
 U_{\lambda,l}(X^l,Y^l,\Xi^l,\Phi^l)\in\C\atop\text{with } |U_{\lambda,l}(X^l,Y^l, \Xi^l,\Phi^l)|\le r}\\
&\cdot\Bigg|\frac{\partial}{\partial U_{\lambda,k}(X^k,Y^k, \Xi^k,\Phi^k)}\Bigg( P\left(\{U_{\lambda,l}(X^l,Y^l, \Xi^l,\Phi^l)\}\right)\\
&\quad-P_h\left(\{U_{\lambda,l}(X^l,Y^l, \Xi^l,\Phi^l)\}\right)\Bigg)\Bigg|=0
\end{split}
\end{equation}
for all $r>0$ and any $k\in\{1,\cdots,\max\{\hm,\tn\}\}$, $(X^{k},Y^{k}, \Xi^{k},\Phi^{k})\in\G^{2k}\times\spin^{2k}$.  

Since $\Tr e^{-\beta H_{\lambda}}/\Tr e^{-\beta H_0}>0$, the uniform
 convergence property \eqref{eq_convergence_partition_function} verifies the first statement of the lemma on $\Re (\Tr
 e^{-\beta H_{\lambda}}/\Tr e^{-\beta H_0})_h$.

Note the equality that
\begin{equation}\label{eq_equality_differential_operator}
\begin{split}
&\frac{\partial}{\partial
 \lambda(\hX^{\hm},\hY^{\hm},\hXi^{\hm},\hPhi^{\hm})}\log\left(\frac{\Tr
 e^{-\beta H_{\lambda}}}{\Tr e^{-\beta H_0}}\right)_h\\
&\quad=\frac{\left(\frac{\partial}{\partial
 U_{\lambda,\hm}(\hX^{\hm},\hY^{\hm},\hXi^{\hm},\hPhi^{\hm})}+\frac{\partial}{\partial
 U_{\lambda,\hm}(\hY^{\hm},\hX^{\hm},\hPhi^{\hm},\hXi^{\hm})}\right)P_h(\{U_{\lambda,l}(X^l,Y^l, \Xi^l,\Phi^l)\})}{P_h(\{U_{\lambda,l}(X^l,Y^l, \Xi^l,\Phi^l)\})}.
\end{split}
\end{equation}

Combining the convergence properties \eqref{eq_convergence_partition_function}-\eqref{eq_convergence_derivative} and the equality \eqref{eq_equality_differential_operator} with \eqref{eq_derivative_correlation} yields the equality \eqref{eq_correlation_h_limit}.
\end{proof}

\subsection{The Grassmann integral formulation}
As the third step we formulate $(\Tr e^{-\beta H_{\lambda}}/\Tr e^{-\beta H_0})_h$ into a Grassmann Gaussian integral on the finite dimensional Grassmann algebra \\
$\{\opsi_{\bx\xi x},\psi_{\bx\xi x}\}_{(\bx,\xi, x)\in \G\times\spin \times[0,\beta)_h}$. Basic properties of the finite dimensional Grassmann integral have been summarized in the books \cite{FKT}, \cite{S}. We assume that each element $(\bx,\xi,x)\in\G\times\spin\times[0,\beta)_h$ is numbered so that we can write
$$\G\times\spin\times[0,\beta)_h=\{(\bx\xi x)_j\ |\ j\in\{1,\cdots,N\}\}$$
with $N:=2L^d\beta h$.

The Grassmann integral $\int\cdot d\psi_{(\bx\xi x)_N}\cdots
d\psi_{(\bx\xi x)_1}d\opsi_{(\bx\xi x)_N}\cdots d\opsi_{(\bx\xi x)_1}$
is a linear functional on the complex linear space $\C[\opsi_{(\bx\xi
x)_j},\psi_{(\bx\xi x)_j}\ |\ j\in\{1,\cdots, N\}]$ of Grassmann monomials and satisfies that
\begin{equation}\label{eq_definition_grassmann_integral}
\begin{split}
&\int\opsi_{(\bx\xi x)_1}\cdots\opsi_{(\bx\xi x)_N}\psi_{(\bx\xi
 x)_1}\cdots\psi_{(\bx\xi x)_N}\\
&\qquad \cdot d\psi_{(\bx\xi x)_N}\cdots d\psi_{(\bx\xi x)_1}d\opsi_{(\bx\xi x)_N}\cdots d\opsi_{(\bx\xi x)_1}=1,\\
&\int\opsi_{(\bx\xi x)_{j_1}}\cdots\opsi_{(\bx\xi x)_{j_k}}\psi_{(\bx\xi
 x)_{p_1}}\cdots\psi_{(\bx\xi x)_{p_q}}\\
&\qquad \cdot d\psi_{(\bx\xi x)_N}\cdots d\psi_{(\bx\xi x)_1}d\opsi_{(\bx\xi
 x)_N}\cdots d\opsi_{(\bx\xi x)_1}=0,
\end{split}
\end{equation}
if $k<N$ or $q< N$.

 For simplicity let $\opsi_X$, $\psi_X$ denote the vectors of the Grassmann variables \\
$(\opsi_{(\bx\xi x)_1},\cdots,\opsi_{(\bx\xi x)_N})$, $(\psi_{(\bx\xi x)_1},\cdots,\psi_{(\bx\xi x)_N})$, respectively. We also write $d\opsi_X=d\opsi_{(\bx\xi x)_N}\cdots d\opsi_{(\bx\xi x)_1}$ and $d\psi_X=d\psi_{(\bx\xi x)_N}\cdots d\psi_{(\bx\xi x)_1}$. 

For any $f(\opsi_X,\psi_X)\in \C[\opsi_{(\bx\xi x)_j},\psi_{(\bx\xi
x)_j}\ |\ j\in\{1,\cdots, N\}]$ the value of the Grassmann integral
$\int f(\opsi_X,\psi_X)d\psi_Xd\opsi_X$ can be computed by linearity
and the anti-commutation relations of the Grassmann variables. 

For any $N\times N$ matrix $G_h$ with $\det G_h\neq 0$, let $\bG_h$ denote the $2N\times2N$ skew symmetric matrix 
\begin{equation*}
\left(\begin{array}{cc} 0 & G_h \\ -G_h^t & 0 \end{array}\right).
\end{equation*}
From the definition and the assumption that $h\in 2\N/\beta$ we see that
\begin{equation*}
\begin{split}
\int e^{-\frac{1}{2}\<(\opsi_X,\psi_X)^t,\bG_h^{-1}(\opsi_X,\psi_X)^t\>}d\psi_Xd\opsi_X&=\int e^{-\<\psi_X^t,G_h^{-1}\opsi_X^t\>}d\psi_Xd\opsi_X\\
&=(-1)^{N(N-1)/2}(\det G_h)^{-1}=(\det G_h)^{-1}.
\end{split}
\end{equation*}

\begin{definition}\label{def_grassmann_gaussian_integral}
For any $N\times N$ matrix $G_h$ with $\det G_h\neq 0$ the Grassmann Gaussian integral $\int\cdot d\mu_{G_h}(\opsi_X,\psi_X):\C[\opsi_{(\bx\xi x)_j},\psi_{(\bx\xi x)_j}\ |\ j\in\{1,\cdots,N\}]\to\C$ is defined by
\begin{equation*}
\int f(\opsi_X,\psi_X)d\mu_{G_h}(\opsi_X,\psi_X):=\frac{\int f(\opsi_X,\psi_X)e^{-\frac{1}{2}\<(\opsi_X,\psi_X)^t,\bG_h^{-1}(\opsi_X,\psi_X)^t\>}d\psi_Xd\opsi_X}{\int e^{-\frac{1}{2}\<(\opsi_X,\psi_X)^t,\bG_h^{-1}(\opsi_X,\psi_X)^t\>}d\psi_Xd\opsi_X}.
\end{equation*}
\end{definition}

Transforming $(\Tr e^{-\beta H_{\lambda}}/\Tr e^{-\beta H_0})_h$ into a
Grassmann Gaussian integral essentially relies on the following equality.

\begin{lemma}\label{lem_grassmann_gaussian_equality}
\begin{equation*}
\begin{split}
&\int\opsi_{(\bx\xi x)_{j_k}}\cdots \opsi_{(\bx\xi x)_{j_1}}
 \psi_{(\bx\xi x)_{p_1}}\cdots\psi_{(\bx\xi
 x)_{p_k}}d\mu_{G_h}(\opsi_X,\psi_X)\\
&\quad =\det(G_h((\bx\xi x)_{j_u},(\bx\xi x)_{p_v}))_{1\le u,v\le k}.
\end{split}
\end{equation*}
\end{lemma}
\begin{proof}
This equality follows \cite[\mbox{Problem I.13}]{FKT} by replacing
 the notations $\psi_X$, $\opsi_X$ by $\opsi_X$, $\psi_X$, respectively.
\end{proof}

\begin{remark}
Our definition of the Grassmann Gaussian integral Definition
 \ref{def_grassmann_gaussian_integral} differs from that summarized in
 \cite[\mbox{Problem I.13}]{FKT} or its follower \cite[\mbox{Section
 3.2}]{K} and corresponds to the statement derived by changing $\psi_X$ for $\opsi_X$ and $\opsi_X$ for $\psi_X$ respectively in \cite[\mbox{Problem I.13}]{FKT} and \cite[\mbox{Section 3.2}]{K}. In \cite[\mbox{Proposition 3.7}]{K} the discretized partition function of the Hubbard model was formulated in a Grassmann Gaussian integral. The formulation \cite[\mbox{Proposition 3.7}]{K} required the symmetry assumption on the coefficients $U_{\bx,\by,\bz,\bw}$ $(\bx,\by,\bz,\bw\in\G)$, namely $U_{\bx,\by,\bz,\bw}=U_{\bz,\bw,\bx,\by}$. Under Definition \ref{def_grassmann_gaussian_integral}, however, we do not need any additional assumption on the coefficients $U_{\lambda,l}(X^l,Y^l,\Xi^l,\Phi^l)$ to complete the desired formulation below.
\end{remark}

\begin{lemma}\label{lem_grassmann_gaussian_partition}
\begin{equation}\label{eq_grassmann_gaussian_partition}
\begin{split}
&\left(\frac{\Tr e^{-\beta H_{\lambda}}}{\Tr e^{-\beta H_0}}\right)_h=\\
&\qquad\int
 e^{\sum_{l=1}^{\max\{\hm,\tn\}}\sum_{(X^l,Y^l,\Xi^l,\Phi^l)\in\G^{2l}\times\spin^{2l}}U_{\lambda,l}(X^l,Y^l,\Xi^l,\Phi^l)V_{h,X^l,Y^l,\Xi^l,\Phi^l}^l(\opsi_X,\psi_X)}\\
&\qquad\qquad\cdot d\mu_{C_h}(\opsi_X,\psi_X),
\end{split}
\end{equation}
where 
\begin{equation*}
\begin{split}
&V_{h,X^l,Y^l,\Xi^l,\Phi^l}^l(\opsi_X,\psi_X)\\
&\quad
 :=-\frac{1}{h}\sum_{x\in[0,\beta)_h}\opsi_{\bx_{1}\xi_{1}x}\opsi_{\bx_{2}\xi_{2}x}\cdots \opsi_{\bx_{l}\xi_{l}x}\psi_{\by_{l}\phi_{l}x}\psi_{\by_{l-1}\phi_{l-1}x}\cdots\psi_{\by_{1}\phi_{1}x}.
\end{split}
\end{equation*}
\end{lemma}
\begin{proof}
First note that $\det C_h\neq 0$ for any $h\in 2\N/\beta$ by
 \cite[\mbox{Proposition C.7}]{K}. 
By applying Lemma \ref{lem_grassmann_gaussian_equality} to $\det(C_h((\widetilde{\bx\xi s})_p,(\widetilde{\by\phi s})_q))_{1\le p,q\le \sum_{k=1}^ml_k}$ in \eqref{eq_discrete_perturbation_series} and using the fact that $\int1d\mu_{C_h}(\opsi_X,\psi_X)=1$, we have
\begin{equation*}
\begin{split}
&\left(\frac{\Tr e^{-\beta H_{\lambda}}}{\Tr e^{-\beta
 H_0}}\right)_h=1+\sum_{m=1}^{2L^d\beta h}\frac{(-1)^m}{m!}\\
&\cdot\prod_{k=1}^m\left(\sum_{l_k=1}^{\max\{\hm,\tn\}}\sum_{(X_k^{l_k},Y_k^{l_k}, \Xi_k^{l_k},\Phi_k^{l_k})\in\G^{2l_k}\times\spin^{2l_k}}U_{\lambda,l_k}(X_k^{l_k},Y_k^{l_k}, \Xi_k^{l_k},\Phi_k^{l_k})\frac{1}{h}\sum_{s_k\in [0,\beta)_h}\right)\\
&\qquad \cdot\int\opsi_{\bx_{m,l_m}\xi_{m,l_m}s_m}\cdots\opsi_{\bx_{m,1}\xi_{m,1}s_m}\cdots\opsi_{\bx_{1,l_1}\xi_{1,l_1}s_1}\cdots\opsi_{\bx_{1,1}\xi_{1,1}s_1}\\
&\qquad\quad\cdot\psi_{\by_{1,1}\phi_{1,1}s_1}\cdots\psi_{\by_{1,l_1}\phi_{1,l_1}s_1}\cdots\psi_{\by_{m,1}\phi_{m,1}s_m}\cdots\psi_{\by_{m,l_m}\phi_{m,l_m}s_m}d\mu_{C_h}(\opsi_X,\psi_X)
\end{split}
\end{equation*}
\begin{equation*}
\begin{split}
&=\int\Bigg(1+\sum_{m=1}^{2L^d\beta h}\frac{1}{m!}\\
&\cdot \prod_{k=1}^m\Bigg(\sum_{l_k=1}^{\max\{\hm,\tn\}}\sum_{(X_k^{l_k},Y_k^{l_k}, \Xi_k^{l_k},\Phi_k^{l_k})\in\G^{2l_k}\times\spin^{2l_k}}U_{\lambda,l_k}(X_k^{l_k},Y_k^{l_k}, \Xi_k^{l_k},\Phi_k^{l_k})\frac{-1}{h}\sum_{s_k\in [0,\beta)_h}\\
&\qquad\quad\cdot\opsi_{\bx_{k,1}\xi_{k,1}s_k}\cdots\opsi_{\bx_{k,l_k}\xi_{k,l_k}s_k}\psi_{\by_{k,l_k}\phi_{k,l_k}s_k}\cdots\psi_{\by_{k,1}\phi_{k,1}s_k}\Bigg)\Bigg)d\mu_{C_h}(\opsi_X,\psi_X),
\end{split}
\end{equation*}
which is \eqref{eq_grassmann_gaussian_partition}.
\end{proof}

Here we complete our Grassmann integral formulation of the correlation function.\begin{proposition}\label{prop_correlation_grassmann_integral}
\begin{equation}\label{eq_correlation_grassmann_integral}
\begin{split}
&\<\psi_{\hbx_1\hxi_1}^*\cdots\psi_{\hbx_{\hm}\hxi_{\hm}}^*\psi_{\hby_{\hm}\hphi_{\hm}}\cdots\psi_{\hby_{1}\hphi_{1}}+\psi_{\hby_{1}\hphi_{1}}^*\cdots\psi_{\hby_{\hm}\hphi_{\hm}}^*\psi_{\hbx_{\hm}\hxi_{\hm}}\cdots\psi_{\hbx_1\hxi_1}\>_L\\
&=-\frac{1}{\beta}\lim_{h\to +\infty\atop h\in
 2\N/\beta}\int \left(V_{h,\hX^{\hm},\hY^{\hm},\hXi^{\hm},\hPhi^{\hm}}^{\hm}(\opsi_X,\psi_X)+V_{h,\hY^{\hm},\hX^{\hm},\hPhi^{\hm}, \hXi^{\hm}}^{\hm}(\opsi_X,\psi_X)\right)\\
&\quad\cdot  e^{\sum_{l=1}^{\tn}\sum_{(X^l,\Xi^l,\Phi^l)\in\G^{l}\times\spin^{2l}}U_{L,l}(X^l,\Xi^l,\Phi^l)V_{h,X^l,X^l,\Xi^l,\Phi^l}^l(\opsi_X,\psi_X)}d\mu_{C_h}(\opsi_X,\psi_X)\\
&\quad\cdot\Bigg/\int e^{\sum_{l=1}^{\tn}\sum_{(X^l,\Xi^l,\Phi^l)\in\G^{l}\times\spin^{2l}}U_{L,l}(X^l,\Xi^l,\Phi^l)V_{h,X^l,X^l,\Xi^l,\Phi^l}^l(\opsi_X,\psi_X)}d\mu_{C_h}(\opsi_X,\psi_X).
\end{split}
\end{equation}
\end{proposition}

\begin{proof}
The equality \eqref{eq_correlation_grassmann_integral} is obtained by substituting
 \eqref{eq_grassmann_gaussian_partition} into the right hand side of
 \eqref{eq_correlation_h_limit} and differentiating the Grassmann polynomial
$$e^{\sum_{l=1}^{\max\{\hm,\tn\}}\sum_{(X^l,Y^l,\Xi^l,\Phi^l)\in\G^{2l}\times\spin^{2l}}U_{\lambda,l}(X^l,Y^l,\Xi^l,\Phi^l)V_{h,X^l,Y^l,\Xi^l,\Phi^l}^l(\opsi_X,\psi_X)}$$
by the variable $\lambda(\hX^{\hm},\hY^{\hm},\hXi^{\hm},\hPhi^{\hm})$
 inside the Grassmann integral (see \\
\cite[\mbox{Problem I.3}]{FKT} for differentiation of Grassmann polynomials).
\end{proof}

\section{Perturbative bounds on the Grassmann integral formulation}\label{sec_perturbative_bound}
In this section we find upper bounds on each term of a perturbative expansion of the Grassmann integrals of the form
\eqref{eq_correlation_grassmann_integral} for fixed $h\in
2\N/\beta$. Here we generalize the problem. The covariance is
assumed to be any matrix $G_h:(\G\times\spin\times[0,\beta)_h)^2\to \C$ with $\det G_h\neq 0$. 

Set an integral of $G_h$ by 
\begin{equation*}
\begin{split}
\cD:=\max\Bigg\{&\max_{(\by,\phi,y)\in\G\times\spin\times[0,\beta)_h}\frac{1}{h}\sum_{(\bx,\xi,x)\in\G\times\spin\times[0,\beta)_h}|G_h(\bx\xi x,\by\phi y)|,\\
&\max_{(\by,\phi,y)\in\G\times\spin\times[0,\beta)_h}\frac{1}{h}\sum_{(\bx,\xi,x)\in\G\times\spin\times[0,\beta)_h}|G_h(\by\phi y,\bx\xi x)|\Bigg\}.
\end{split}
\end{equation*}

We define the Schwinger function by  
\begin{equation}\label{eq_schwinger_functional}
\begin{split}
&S_{\hX^{\hm},\hY^{\hm},\hXi^{\hm},\hPhi^{\hm}}(G_h,\eta)\\
&:=-\frac{1}{\beta}\int V_{h,\hX^{\hm},\hY^{\hm},\hXi^{\hm},\hPhi^{\hm}}^{\hm}(\opsi_X,\psi_X)\\
&\quad\cdot  e^{\eta\sum_{l=1}^{\tn}\sum_{(X^l,\Xi^l,\Phi^l)\in\G^{l}\times\spin^{2l}} U_{L,l}(X^l,\Xi^l,\Phi^l)V_{h,X^l,X^l,\Xi^l,\Phi^l}^l(\opsi_X,\psi_X)}d\mu_{G_h}(\opsi_X,\psi_X)\\
&\quad\cdot\Bigg/\int e^{\eta\sum_{l=1}^{\tn}\sum_{(X^l,\Xi^l,\Phi^l)\in\G^{l}\times\spin^{2l}}U_{L,l}(X^l,\Xi^l,\Phi^l)V_{h,X^l,X^l,\Xi^l,\Phi^l}^l(\opsi_X,\psi_X)}d\mu_{G_h}(\opsi_X,\psi_X)
\end{split}
\end{equation}
for any $\hm\in\N$ and $(\hX^{\hm},\hY^{\hm},\hXi^{\hm},\hPhi^{\hm})\in
\G^{2\hm}\times\spin^{2\hm}$ and any $\eta\in\C$ for which the denominator of \eqref{eq_schwinger_functional} does not vanish. In fact, since 
\begin{equation*}
\begin{split}
\lim_{\eta\to 0}&\int e^{\eta\sum_{l=1}^{\tn}\sum_{(X^l,\Xi^l,\Phi^l)\in\G^{l}\times\spin^{2l}}U_{L,l}(X^l,\Xi^l,\Phi^l)V_{h,X^l,X^l,\Xi^l,\Phi^l}^l(\opsi_X,\psi_X)}\\
&\cdot d\mu_{G_h}(\opsi_X,\psi_X)=1,
\end{split}
\end{equation*}
$S_{\hX^{\hm},\hY^{\hm},\hXi^{\hm},\hPhi^{\hm}}(G_h,\eta)$ is analytic with respect to $\eta$ in a neighborhood of origin.

The purpose of this section is to prove the following proposition.
\begin{proposition}\label{prop_schwinger_functional_bound}
Assume that there exists a positive constant $\cB>0$ such that for any $p\in\N$, $(\bx_{j_l},\xi_{j_l},x_{j_l}),(\by_{k_l},\phi_{k_l},y_{k_l})\in \G\times\spin\times[0,\beta)_h$ $(\forall l\in\{1,\cdots,p\})$ and $n\in\N$, 
\begin{equation}\label{eq_assumption_determinant_bound}
\sup_{\bu_l,\bv_l\in\C^n\atop \|\bu_l\|_{\C^n},\|\bv_l\|_{\C^n}\le 1\ \forall l\in\{1,\cdots,p\}}|\det\left(\<\bu_l,\bv_m\>_{\C^n}G_h(\bx_{j_l}\xi_{j_l}x_{j_l},\by_{k_m}\phi_{k_m}y_{k_m})\right)_{1\le l,m\le p}|\le \cB^p.
\end{equation}

Let $b_m$ denote the coefficient of $\eta^m$ in the Taylor series of the function $\eta\mapsto $\\
$S_{\hX^{\hm},\hY^{\hm},\hXi^{\hm},\hPhi^{\hm}}(G_h,\eta)$ around $0\in\C$. The following bounds hold.
\begin{equation*}
\begin{split}
&|b_0|\le \cB^{\hm},\\
&|b_m|\le \frac{\hm 4^{\hm}\cB^{\hm}}{m}\left(\sum_{l=1}^{\tn}l4^l\cB^{l-1}\|U_{L,l}\|_{L,l}\cD\right)^m\ (\forall m\in\N).
\end{split}
\end{equation*}
\end{proposition}

\begin{proof} Let us consider $\{U_{\lambda,l}(X^l,Y^l,\Xi^l,\Phi^l)\}_{l\in\{1,\cdots,\max\{\hm,\tn\}\}, (X^l,Y^l,\Xi^l,\Phi^l)\in\G^{2l}\times\spin^{2l}}$ as mutually independent, complex multi-variables. In a neighborhood of origin we can define a multi-variable analytic function $W(\{U_{\lambda,l}(X^l,Y^l,\Xi^l,\Phi^l)\})$ by
\begin{equation*}
\begin{split}
&W(\{U_{\lambda,l}(X^l,Y^l,\Xi^l,\Phi^l)\})\\
&:=\log\Bigg(\int
 e^{\sum_{l=1}^{\max\{\hm,\tn\}}\sum_{(X^l,Y^l,\Xi^l,\Phi^l)\in\G^{2l}\times\spin^{2l}}U_{\lambda,l}(X^l,Y^l,\Xi^l,\Phi^l)V_{h,X^l,Y^l,\Xi^l,\Phi^l}^l(\opsi_X,\psi_X)}\\
&\qquad\qquad \cdot d\mu_{G_h}(\opsi_X,\psi_X)\Bigg).
\end{split}
\end{equation*}
Its Taylor expansion becomes
\begin{equation*}
\begin{split}
&W(\{U_{\lambda,l}(X^l,Y^l,\Xi^l,\Phi^l)\})=\sum_{m=0}^{\infty}\frac{1}{m!}\prod_{k=1}^m\left(\sum_{l_k=1}^{\max\{\hm,\tn\}}\sum_{(X^{l_k}_k,Y^{l_k}_k,\Xi^{l_k}_k,\Phi^{l_k}_k)\in\G^{2l_k}\times\spin^{2l_k}}\right)\\
&\cdot \prod_{k=1}^m\frac{\partial}{\partial U_{\lambda,l_k}(X^{l_k}_k,Y^{l_k}_k,\Xi^{l_k}_k,\Phi^{l_k}_k)}W(\{U_{\lambda,l}(X^l,Y^l,\Xi^l,\Phi^l)\})\Bigg|_{\substack{U_{\lambda,l}(X^{l},Y^{l},\Xi^{l},\Phi^{l})=0\\ \forall l\in\{1,\cdots,\max\{\hm,\tn\}\}\\ \forall (X^{l},Y^{l},\Xi^{l},\Phi^{l})\in\G^{2l}\times\spin^{2l}}}\\
&\cdot \prod_{k=1}^mU_{\lambda,l_k}(X^{l_k}_k,Y^{l_k}_k,\Xi^{l_k}_k,\Phi^{l_k}_k).
\end{split}
\end{equation*}
We can take a small $r>0$ so that for all $\eta\in\C\backslash \{0\}$ with $|\eta|<r$
\begin{equation}\label{eq_schwinger_functional_expansion}
\begin{split}
&S_{\hX^{\hm},\hY^{\hm},\hXi^{\hm},\hPhi^{\hm}}(G_h,\eta)\\
&\quad =-\frac{1}{\beta\eta}\frac{\partial}{\partial
 U_{\lambda,\hm}(\hX^{\hm},\hY^{\hm},\hXi^{\hm},\hPhi^{\hm})}\\
&\qquad \cdot W(\{\eta U_{\lambda,l}(X^l,Y^l,\Xi^l,\Phi^l)\})\Bigg|_{\lambda(X^{\hm},Y^{\hm},\Xi^{\hm},\Phi^{\hm})=0\atop\forall (X^{\hm},Y^{\hm},\Xi^{\hm},\Phi^{\hm})\in\G^{2\hm}\times\spin^{2\hm}}\\  
&\quad =
 -\frac{1}{\beta}\sum_{m=1}^{\infty}\frac{1}{(m-1)!}\prod_{k=1}^{m-1}\left(\sum_{l_k=1}^{\tn}\sum_{(X^{l_k}_k,\Xi^{l_k}_k,\Phi^{l_k}_k)\in\G^{l_k}\times\spin^{2l_k}}\right)\\
&\qquad \cdot\frac{\partial}{\partial
 U_{\lambda,\hm}(\hX^{\hm},\hY^{\hm},\hXi^{\hm},\hPhi^{\hm})}\prod_{k=1}^{m-1}\frac{\partial}{\partial U_{\lambda,l_k}(X^{l_k}_k,X^{l_k}_k,\Xi^{l_k}_k,\Phi^{l_k}_k)}\\
&\qquad \cdot W(\{ U_{\lambda,l}(X^l,Y^l,\Xi^l,\Phi^l)\})\Bigg|_{\substack{U_{\lambda,l}(X^{l},Y^{l},\Xi^{l},\Phi^{l})=0\\ \forall l\in\{1,\cdots,\max\{\hm,\tn\}\}\\ \forall (X^{l},Y^{l},\Xi^{l},\Phi^{l})\in\G^{2l}\times\spin^{2l}}}\\
&\qquad\cdot\prod_{k=1}^{m-1}U_{L,l_k}(X_k^{l_k},\Xi_k^{l_k},\Phi_k^{l_k})\cdot\eta^{m-1}.
\end{split}
\end{equation}
By comparing with the definition \eqref{eq_schwinger_functional} we see that the equality \eqref{eq_schwinger_functional_expansion} also
 holds true for $\eta =0$.

We will bound the coefficient of $\eta^{m}$ in the right hand side of
 \eqref{eq_schwinger_functional_expansion}.  To organize the argument we
 divide the rest of the proof by 4 steps.

(Step 1) First remark that by Lemma \ref{lem_grassmann_gaussian_equality}. 
\begin{equation}\label{eq_bound_for_m_1}
\begin{split}
b_0&=\\
&\ -\frac{1}{\beta}\frac{\partial}{\partial U_{\lambda,\hm}(\hX^{\hm},\hY^{\hm},\hXi^{\hm},\hPhi^{\hm})}W(\{ U_{\lambda,l}(X^l,Y^l,\Xi^l,\Phi^l)\})\Bigg|_{\substack{U_{\lambda,l}(X^{l},Y^{l},\Xi^{l},\Phi^{l})=0\\ \forall l\in\{1,\cdots,\max\{\hm,\tn\}\}\\ \forall (X^{l},Y^{l},\Xi^{l},\Phi^{l})\in\G^{2l}\times\spin^{2l}}}\\
&=\frac{1}{\beta h}\sum_{s\in[0,\beta)_h}\int\opsi_{\hbx_1\hxi_1s}\cdots\opsi_{\hbx_{\hm}\hxi_{\hm}s}\psi_{\hby_{\hm}\hphi_{\hm}s}\cdots \psi_{\hby_{1}\hphi_{1}s}d\mu_{G_h}(\opsi_X,\psi_X)\\
&=\det(G_h(\hbx_j\hxi_{j}0,\hby_k\hphi_k0))_{1\le j,k\le \hm}.
\end{split}
\end{equation}
Thus, by the assumption \eqref{eq_assumption_determinant_bound} $|b_0|\le \cB^{\hm}$.

(Step 2) Consider the coefficient of $\eta^{m}$ in
 \eqref{eq_schwinger_functional_expansion} for $m\ge 1$. Here we apply
 the tree formula \cite[\mbox{Theorem 3}]{SW} to characterize the partial derivatives of $W(\{ U_{\lambda,l}(X^l,Y^l,\Xi^l,\Phi^l)\})$ evaluated at origin. We then obtain
\begin{equation}\label{eq_tree_expansion}
\begin{split}
&b_{m}\\
&=-\frac{1}{\beta}\frac{1}{m!}\prod_{k=1}^{m}\left(\sum_{l_k=1}^{\tn}\sum_{(X^{l_k}_k,\Xi^{l_k}_k,\Phi^{l_k}_k)\in\G^{l_k}\times\spin^{2l_k}}U_{L,l_k}(X_k^{l_k},\Xi_k^{l_k},\Phi_k^{l_k})\right)\\
&\quad\cdot\frac{\partial}{\partial
 U_{\lambda,\hm}(\hX^{\hm},\hY^{\hm},\hXi^{\hm},\hPhi^{\hm})}\prod_{k=1}^{m}\frac{\partial}{\partial U_{\lambda,l_k}(X^{l_k}_k,X^{l_k}_k,\Xi^{l_k}_k,\Phi^{l_k}_k)}\\
&\quad\cdot W(\{ U_{\lambda,l}(X^l,Y^l,\Xi^l,\Phi^l)\})\Bigg|_{\substack{U_{\lambda,l}(X^{l},Y^{l},\Xi^{l},\Phi^{l})=0\\ \forall l\in\{1,\cdots,\max\{\hm,\tn\}\}\\ \forall (X^{l},Y^{l},\Xi^{l},\Phi^{l})\in\G^{2l}\times\spin^{2l}}}\\
&= -\frac{1}{\beta}\frac{1}{m!}\prod_{k=1}^{m}\left(\sum_{l_k=1}^{\tn}\sum_{(X^{l_k}_k,\Xi^{l_k}_k,\Phi^{l_k}_k)\in\G^{l_k}\times\spin^{2l_k}}U_{L,l_k}(X_k^{l_k},\Xi_k^{l_k},\Phi_k^{l_k})\right)\\
&\cdot \sum_{T\in\T(\{0,1,\cdots,m\})}\prod_{\{p,q\}\in
 T}(\D_{p,q}+\D_{q,p})\int_{[0,1]^{m}}d\bs \sum_{\pi\in
 \S_{m+1}(T)}\chi(T,\pi,\bs)\cdot e^{\D(M(T,\pi,\bs))}\\
&\cdot V_{h,\hX^{\hm},\hY^{\hm},\hXi^{\hm},\hPhi^{\hm}}^{\hm}(\opsi_X^0,\psi_X^0)\prod_{k=1}^{m}V_{h,X_k^{l_k},X_k^{l_k},\Xi_k^{l_k},\Phi_k^{l_k}}^{l_k}(\opsi_X^k,\psi_X^k)\Bigg|_{\opsi_X^q=\psi_X^q=\b0\atop \forall q\in\{0,1,\cdots,m\}}.
\end{split}
\end{equation}
Definitions of the newly introduced notations are in order; $\T(\{0,1,\cdots,m\})$ is the set of all trees on the vertices $\{0,1,\cdots,m\}$, $\S_{m+1}(T)$ is a $T$-dependent set of permutations over $m+1$ elements, $\chi(T,\pi,\bs)$ is a $(T,\pi,\bs)$-dependent non-negative real function obeying
\begin{equation}\label{eq_phi_equality}
\int_{[0,1]^{m}}d\bs\sum_{\pi\in \S_{m+1}(T)}\chi(T,\pi,\bs)=1,
\end{equation}
$M(T,\pi,\bs)$ is a $(T,\pi,\bs)$-dependent non-negative real symmetric
 $(m+1)\times (m+1)$ matrix satisfying $M(T,\pi,\bs)_{p,p}=1$ for all $p \in \{0,1,\cdots,m\}$ and the Laplacian operators $\D_{p,q}$ and $\D(M(T,\pi,\bs))$ are defined by
$$\D_{p,q}:=-\<\left(\frac{\partial}{\partial \opsi_X^p}\right)^t,G_h\left(\frac{\partial}{\partial \psi_X^{q}}\right)^t\>,\ \D(M(T,\pi,\bs)):=\sum_{p,q=0}^{m}M(T,\pi,\bs)_{p,q}\D_{p,q},$$
where 
$$\frac{\partial}{\partial \opsi_X^q}:=\left(\frac{\partial}{\partial \opsi_{(\bx\xi x)_1}^q},\cdots,\frac{\partial}{\partial \opsi_{(\bx\xi x)_N}^q}\right),\ \frac{\partial}{\partial \psi_X^q}:=\left(\frac{\partial}{\partial \psi_{(\bx\xi x)_1}^q},\cdots,\frac{\partial}{\partial \psi_{(\bx\xi x)_N}^q}\right)$$
are vectors of the Grassmann left derivatives corresponding to the labeled Grassmann variables 
$\{\opsi_{(\bx\xi x)_j}^q,\psi_{(\bx\xi x)_j}^q\ |\
 j\in\{1,\cdots,N\}\}$ $(\forall q\in\{0,1,\cdots,m\})$. 

(Step 3) We will bound the right hand side of \eqref{eq_tree_expansion} by replacing the sum over trees by the sum over incidence numbers in the next step. As the preparation let us summarize some facts. Let
 $d_0,d_1,\cdots,d_{m}$ denote the incidence numbers of a tree $T$
 corresponding to the vertices $0,1,\cdots,m$, respectively. Every term of the expansion of $\prod_{\{p,q\}\in T}(\D_{p,q}+\D_{q,p})$ is a product of $m$ Laplacians as the tree $T$ has $m$ lines. Each product
 of the $m$ Laplacians contains $d_q$ derivatives with respect to the
 Grassmann variables with label $q$ for all
 $q\in\{0,1,\cdots,m\}$. Moreover, the $d_q$ derivatives consist of
 $\alpha_q$ derivatives with respect to the variables $\{\psi_{\bx\xi
 x}^q\}_{(\bx,\xi, x)\in \G\times \spin\times [0,\beta)_h}$ and
 $\talpha_q$ derivatives with respect to the variables $\{\opsi_{\bx\xi x}^q\}_{(\bx,\xi, x)\in \G\times \spin\times [0,\beta)_h}$ for some non-negative integers $\alpha_q,\talpha_q$ satisfying $d_q=\alpha_q+\talpha_q$.

This product of the $m$ Laplacians produces
$$1_{\alpha_0\le \hm}1_{\talpha_0\le \hm}\left(\begin{array}{c} \hm \\ \alpha_0\end{array}\right)\alpha_0!\left(\begin{array}{c} \hm \\ \talpha_0\end{array}\right)\talpha_0!\prod_{k=1}^{m}1_{\alpha_k\le l_k}1_{\talpha_k\le l_k}\left(\begin{array}{c} l_k \\ \alpha_k\end{array}\right)\alpha_k!\left(\begin{array}{c} \l_k \\ \talpha_k\end{array}\right)\talpha_k!$$
monomials when the term acts on the monomial
\begin{equation}\label{eq_grassmann_monomial}
\begin{split}
&\opsi_{\hbx_1\hxi_1 s_0}\opsi_{\hbx_2\hxi_2 s_0}\cdots
 \opsi_{\hbx_{\hm}\hxi_{\hm} s_0}\psi_{\hby_{\hm}\hphi_{\hm}
 s_0}\psi_{\hby_{\hm-1}\hphi_{\hm-1} s_0}\cdots\psi_{\hby_{1}\hphi_{1}
 s_0}\\
&\cdot\prod_{k=1}^{m}\Big(\opsi_{\bx_{k,1}\xi_{k,1}s_k}^k\opsi_{\bx_{k,2}\xi_{k,2}s_k}^k\cdots\opsi_{\bx_{k,l_k}\xi_{k,l_k}s_k}^k\\
&\qquad\quad\cdot\psi_{\by_{k,l_k}\phi_{k,l_k}s_k}^k\psi_{\by_{k,l_k-1}\phi_{k,l_k-1}s_k}^k\cdots\psi_{\by_{k,1}\phi_{k,1}s_k}^k\Big).
\end{split}
\end{equation}

Moreover, every remaining monomial left after the product of the $m$ Laplacians acting on \eqref{eq_grassmann_monomial} consists of $\hm-\alpha_0+\sum_{k=1}^{m}(l_k-\alpha_k)$ products of the Grassmann variables of the form $\psi_{\bx\xi s}^q$ and $\hm-\talpha_0+\sum_{k=1}^{m}(l_k-\talpha_k)$ products of the Grassmann variables of the form $\opsi_{\bx\xi s}^q$ and is acted by the operator $e^{\D(M(T,\pi,\bs))}$. The determinant bound \eqref{eq_assumption_determinant_bound} and the non-negative symmetric property of $M(T,\pi,\bs)$ together with the fact that all the diagonal elements of $M(T,\pi,\bs)$ are $1$ validate the inequality 
\begin{equation}\label{eq_determinant_bound_application}
\begin{split}
&\Bigg|e^{\D(M(T,\pi,\bs))}\opsi_{\bx_{0,1}\xi_{0,1}s_0}^0\cdots \opsi_{\bx_{0,\hm-\talpha_0}\xi_{0,\hm-\talpha_0}s_0}^0\psi_{\by_{0,1}\phi_{0,1}s_0}^0\cdots \psi_{\by_{0,\hm-\alpha_0}\phi_{0,\hm-\alpha_0}s_0}^0\\
&\quad\cdot\prod_{k=1}^{m}\opsi_{\bx_{k,1}\xi_{k,1}s_k}^k\cdots
 \opsi_{\bx_{k,l_k-\talpha_k}\xi_{k,l_k-\talpha_k}s_k}^k\\
&\qquad\quad\cdot\psi_{\by_{k,1}\phi_{k,1}s_k}^k\cdots \psi_{\by_{k,l_k-\alpha_k}\phi_{k,l_k-\alpha_k}s_k}^k\Big|_{\opsi_X^q=\psi_X^q=\b0\atop \forall q\in\{0,1,\cdots,m\}}\Bigg|\\
&\quad\le 1_{\sum_{q=0}^{m}\alpha_q=\sum_{q=0}^{m}\talpha_q}\cB^{\hm-\alpha_0+\sum_{k=1}^{m}(l_k-\alpha_k)}
\end{split}
\end{equation}
(see \cite[\mbox{Lemma 4.5}]{K} for a detailed proof of the same bound in the case $\cB=4$).

The coefficient of each remaining monomial left after the product of the $m$ Laplacians acting on the monomial \eqref{eq_grassmann_monomial} is a product of $m$ elements of the covariance 
$(G_h(\bx\xi x,\by\phi y))_{(\bx,\xi, x),(\by,\phi, y)\in\G\times\spin\times[0,\beta)_h}$ and going to be multiplied by \\
$\prod_{k=1}^{m}\frac{1}{h}U_{L,l_k}(X_k^{l_k},\Xi_{k}^{l_k},\Phi_k^{l_k})$ and summed over the variables $(X_k^{l_k},\Xi_{k}^{l_k},\Phi_k^{l_k},s_k)\in$\\
$\G^{l_k}\times\spin^{2l_k}\times[0,\beta)_h$ for all $k\in\{1,\cdots,m\}$. By considering that the tree $T$ starts from the vertex $0$ and using the properties \eqref{eq_condition_U}-\eqref{eq_translation_invariance}, a recursive argument through all the lines of $T$ running from the later generation to the earlier generation shows that this sum is bounded by
\begin{equation}\label{eq_recurcive_bound}
\prod_{k=1}^{m}\|U_{L,l_k}\|_{L,l_k}\cD^{m}.
\end{equation}
Note that to derive the upper bound \eqref{eq_recurcive_bound} we repeatedly use the inequality of the form
$$\frac{1}{h}\sum_{s\in[0,\beta)_h}\sum_{(X^l,\Xi^l,\Phi^l)\in\G^l\times \spin^{2l}}|U_{L,l}(X^l,\Xi^l,\Phi^l)||G_h(\by\s y,\bx_j\xi_j s)|\le\|U_{L,l}\|_{L,l}\cD,$$
where $(\by,\s,y)\in \G\times\spin\times [0,\beta)_h$ is any fixed element and $\bx_j$ and $\xi_j$ are $j$-th component of $X^l$ and $\Xi^l$ respectively for some $j\in\{1,\cdots,l\}$.

Lastly remark that the number of terms containing $\alpha_q$ derivatives with respect to the variables $\{\psi_{\bx\xi x}^q\}_{(\bx,\xi,x)\in\G\times \spin\times [0,\beta)_h
 }$ $(\forall q\in\{0,1,\cdots,m\})$ in the expansion $\prod_{\{p,q\}\in
 T}(\Delta_{p,q}+\Delta_{q,p})$ is at most 
$$\prod_{q=0}^{m}\left(\begin{array}{c} d_q \\ \alpha_q \end{array}\right).$$

(Step 4) By summing up all these considerations, replacing the sum over trees by the sum over possible incidence numbers $d_0,\cdots,d_{m}$ satisfying the condition $\sum_{q=0}^{m}d_q=2m$, and using Cayley's theorem on the number of trees with fixed incidence numbers (see, e.g, \cite[\mbox{Corollary 2.2.4}]{W}) and the equality \eqref{eq_phi_equality}, we have
\begin{equation*}
\begin{split}
&|b_{m}|\\
&\le \frac{1}{m!}\frac{1}{\beta
 h}\sum_{s_0\in[0,\beta)_h}\prod_{k=1}^{m}\\
&\quad\cdot\left(\sum_{l_k=1}^{\tn}\sum_{(X^{l_k}_k,\Xi^{l_k}_k,\Phi^{l_k}_k)\in\G^{l_k}\times\spin^{2l_k}}|U_{L,l_k}(X_k^{l_k},\Xi_k^{l_k},\Phi_k^{l_k})|\frac{1}{h}\sum_{s_k\in[0,\beta)_h}\right)\\
&\quad\cdot \sum_{T\in\T(\{0,1,\cdots,m\})}\sup_{\bs\in[0,1]^{m}\atop \pi
 \in \S_{m+1}(T)}\\
&\quad\cdot\Bigg|e^{\D(M(T,\pi,\bs))}\prod_{\{p,q\}\in
 T}(\D_{p,q}+\D_{q,p})\cdot(\text{the monomial
 }\eqref{eq_grassmann_monomial})\Big|_{\opsi_X^q=\psi_X^q=\b0\atop
 \forall q\in\{0,1,\cdots,m\}}\Bigg|\\
&\le \frac{1}{m!}\frac{1}{\beta
 h}\sum_{s_0\in[0,\beta)_h}\prod_{k=1}^{m}\left(\sum_{l_k=1}^{\tn}\|U_{L,l_k}\|_{L,l_k}\cD\right)\sum_{\substack{d_0\in\{1,\cdots,2\hm\}\\ d_k\in\{1,\cdots,2l_k\}\\ \forall k\in\{1,\cdots,m\}}}\frac{(m-1)!}{\prod_{q=0}^{m}(d_q-1)!}1_{\sum_{q=0}^{m}d_q=2m}\\
&\quad\cdot\sum_{\alpha_q,\talpha_q\in\{0,\cdots,d_q\}\atop \forall q\in\{0,1,\cdots,m\}}\prod_{q=0}^{m}1_{\alpha_q+\talpha_q=d_q}1_{\sum_{q=0}^{m}\alpha_q=\sum_{q=0}^{m}\talpha_q}\prod_{q=0}^{m}\left(\begin{array}{c} d_q \\ \alpha_q \end{array}\right)\cB^{\hm-\alpha_0+\sum_{k=1}^{m}(l_k-\alpha_k)}\\
&\quad\cdot1_{\alpha_0\le\hm}1_{\talpha_0\le\hm}\left(\begin{array}{c} \hm \\
		 \alpha_0\end{array}\right)\alpha_0!\left(\begin{array}{c} \hm \\ \talpha_0\end{array}\right)\talpha_0!\prod_{k=1}^{m}1_{\alpha_k\le l_k}1_{\talpha_k\le l_k}\left(\begin{array}{c} l_k \\ \alpha_k\end{array}\right)\alpha_k!\left(\begin{array}{c} \l_k \\ \talpha_k\end{array}\right)\talpha_k!\\
&=\frac{1}{m}\prod_{k=1}^{m}\left(\sum_{l_k=1}^{\tn}\|U_{L,l_k}\|_{L,l_k}\cD\right)\cB^{\hm+\sum_{k=1}^{m}l_k-m}\sum_{\substack{\alpha_0,\talpha_0\in\{0,\cdots,\hm\}\\ \alpha_k,\talpha_k\in\{0,\cdots,l_k\}\\ \forall k\in\{1,\cdots,m\}}}1_{\sum_{q=0}^m\alpha_q=\sum_{q=0}^m\talpha_q=m}\\
&\quad\cdot\prod_{q=0}^{m}(\alpha_q+\talpha_q)\left(\begin{array}{c} \hm \\ \alpha_0\end{array}\right)\left(\begin{array}{c} \hm \\ \talpha_0\end{array}\right)\prod_{k=1}^{m}\left(\begin{array}{c} l_k \\ \alpha_k\end{array}\right)\left(\begin{array}{c} \l_k \\ \talpha_k\end{array}\right).
\end{split}
\end{equation*}
Then, by dropping the constraint $1_{\sum_{q=0}^m\alpha_q=\sum_{q=0}^m\talpha_q=m}$ and using the equality 
$$\sum_{\alpha=0}^l\sum_{\talpha=0}^l(\alpha+\talpha)\left(\begin{array}{c} l \\			\alpha\end{array}\right)\left(\begin{array}{c} l   \\
						      \talpha\end{array}\right)=l4^l,
$$
we obtain 
\begin{equation*}
\begin{split}
|b_{m}|&\le \frac{1}{m}\prod_{k=1}^{m}\left(\sum_{l_k=1}^{\tn}\|U_{L,l_k}\|_{L,l_k}\cD\right)\cB^{\hm+\sum_{k=1}^{m}l_k-m}\hm4^{\hm}\prod_{k=1}^{m}l_k4^{l_k}\\
&= \frac{\hm 4^{\hm}\cB^{\hm}}{m}\left(\sum_{l=1}^{\tn}l4^l\cB^{l-1}\|U_{L,l}\|_{L,l}\cD\right)^{m}.
\end{split}
\end{equation*}
\end{proof}

The evaluation of the tree expansion presented above overcounts the
combinatorial factor, which is defined as the number of monomials appearing in the expansion $\prod_{\{p,q\}\in T}(\D_{p,q}+\D_{q,p})\cdot(\text{the monomial
 }\eqref{eq_grassmann_monomial})$.
If we consider the $4$ point correlation functions for the on-site interaction $V=U\sum_{\bx\in\G}\psi_{\bx\ua}^*\psi_{\bx\da}^*\psi_{\bx\da}\psi_{\bx\ua}$ $(U\in\R)$, the
exact calculation of the combinatorial factor without overcounting is possible as proved in \cite[\mbox{Lemma 4.7}]{K}. Consequently, the perturbative bound is improved in this case.

\begin{proposition}\label{prop_schwinger_functional_hubbard_model_bound}
Assume that a non-singular matrix $G_h$ satisfies \eqref{eq_assumption_determinant_bound}. For any $m\in\N\cup\{0\}$ let $c_m$ denote the coefficient of $\eta^m$ in the Taylor series of the function
\begin{equation*}
\begin{split}
&\eta\longmapsto \\
&\frac{1}{\beta h}\sum_{x\in[0,\beta)_h}\int\opsi_{\hbx_1\ua x}\opsi_{\hbx_2\da x}\psi_{\hby_2\da x}\psi_{\hby_1\ua x}e^{-\eta\frac{U}{h}\sum_{s\in[0,\beta)_h}\opsi_{\bx\ua s}\opsi_{\bx\da s}\psi_{\bx\da s}\psi_{\bx\ua s}}d\mu_{G_h}(\opsi_X,\psi_X)\\
&\cdot \Big/\int e^{-\eta\frac{U}{h}\sum_{s\in[0,\beta)_h}\opsi_{\bx\ua s}\opsi_{\bx\da s}\psi_{\bx\da s}\psi_{\bx\ua s}}d\mu_{G_h}(\opsi_X,\psi_X)
\end{split}
\end{equation*}
around $0\in\C$. The following bound holds. For any $m\in\N\cup\{0\}$
$$|c_m|\le\frac{4\cB^2}{3m+4}\left(\begin{array}{c} 3m+4 \\ m \end{array}\right)(\cD\cB|U|)^m.$$
\end{proposition}
\begin{proof} The argument in \cite[\mbox{Section 4}]{K} leading to the bound \cite[\mbox{Lemma 4.8}]{K} straightforwardly applies to deduce the claimed bound.
\end{proof}

\section{Exponential decay of the correlation functions}\label{sec_exponential_decay_correlation}

In this section we prove the exponential decay property of the
correlation functions. The Grassmann integral formulation
\eqref{eq_correlation_grassmann_integral} serves the base of our analysis. 
In the subsection
\ref{subsec_exponential_decay_correlation} we will see that the $U(1)$-invariance property of the Grassmann integral leads to a
replacement of the covariance $C_h$ by a modified covariance in the
Grassmann integral formulation. In the following subsections
\ref{subsec_determinant_bound}-\ref{subsec_exponential_decay_covariance}
we study properties of such perturbed covariance matrices as
preliminaries for proving our main theorems.

\subsection{Determinant bound on modified covariance matrices}\label{subsec_determinant_bound}
Here we prove the determinant bound on a generalized covariance matrix of the form 
\begin{equation}\label{eq_extended_covariance}
C_{\cE}(\bx\xi x,\by\phi y):=\frac{\delta_{\xi,\phi}}{L^d}\sum_{\bk\in\G^*}e^{i\<\bk,\by-\bx\>}e^{-(y-x){\cE}_{\bk}}\left(\frac{1_{y-x\le
0}}{1+e^{\beta {\cE}_{\bk}}}-\frac{1_{y-x>
0}}{1+e^{-\beta {\cE}_{\bk}}}\right)
\end{equation}
$(\forall (\bx,\xi,x),(\by,\phi,y)\in\G\times\spin\times[0,\beta))$. In
\eqref{eq_extended_covariance} $\cE_{\bk}:\G^*\to\C$ is any function satisfying 
\begin{equation}\label{eq_regular_condition}
|\Im \cE_{\bk}|<\frac{\pi}{\beta}\ (\forall\bk\in\G^*).
\end{equation}
Note that the inequality \eqref{eq_regular_condition} implies $|1+e^{\beta\cE_{\bk}}|,\ |1+e^{-\beta\cE_{\bk}}|>0$ $(\forall\bk\in\G^*)$, and thus $C_{\cE}$ is well-defined. The covariance \eqref{eq_extended_covariance} is a generalization of the covariance \eqref{eq_covariance} in which the dispersion relation $E_{\bk}$ is real-valued.

The volume- and temperature-independent determinant bound on the covariance \eqref{eq_covariance} was established by Pedra and Salmhofer in \cite{PS} by extending the notion of the Gram bound on determinant. However, the Gram representation of the covariance presented in \cite[\mbox{Section 4.1}]{PS} does not apply to  the case that the dispersion relation is complex-valued. We need to construct a Gram representation of the covariance matrix \eqref{eq_extended_covariance} to bound its determinant.

We use the following abstract framework proved by Pedra and Salmhofer as a generalization of Gram's inequality.

\begin{lemma}{\cite[\mbox{Theorem 1.3}]{PS}}\label{lem_abstract_determinant_bound}
Assume that $K,k\in\N\cup\{0\}$ satisfy $k+K\ge 1$. Let $\X$ be any
 set. Assume that $J$ is a set totally ordered under `$\preceq$' and
 `$\prec$' is a strict total order in $J$. Let $\zeta_l$, $\zeta_l'$ be
 maps from $\X$ to $J$ $(\forall l\in\{1,\cdots,k+K\})$. Let $\cH$ be a
 Hilbert space with the inner product $\<\cdot,\cdot\>_{\cH}$ and
 $\|\cdot\|_{\cH}$ denote the norm induced by $\<\cdot,\cdot\>_{\cH}$. Define a matrix $M_{x,y}$ $(x,y\in\X)$ by
\begin{equation}\label{eq_gram_representation}
M_{x,y}:=\<f_x^0,g_y^0\>_{\cH}+\sum_{l=1}^k\<f_x^l,g_y^l\>_{\cH}1_{\zeta_l'(x)\succ\zeta_l(y)}+\sum_{l=k+1}^{K+k}\<f_x^l,g_y^l\>_{\cH}1_{\zeta_l'(x)\succeq\zeta_l(y)},
\end{equation}
where $f_x^l$, $g_y^l\in\cH$. If 
$$\sup_{x\in\X}\max\{\|f_x^l\|_{\cH},\|g_x^l\|_{\cH}\}\le\g_l\ (\forall l\in\{0,\cdots,K+k\}),$$
the following bound holds. For any $m, n\in\N$ and $x_1,\cdots,x_n,y_1,\cdots,y_n\in\X$
\begin{equation}\label{eq_PS_bound}
\sup_{\substack{\bu_j,\bv_j\in\C^m\\
 \|\bu_j\|_{\C^m},\|\bv_j\|_{\C^m}\le 1\\ \forall j\in
  \{1,\cdots,n\}}}|\det(\<\bu_j,\bv_k\>_{\C^m}M_{x_j,y_k})_{1\le j,k\le
  n}|\le \left(\sum_{l=0}^{k+K}\g_l\right)^{2n}.
\end{equation}
\end{lemma}  

\begin{remark}
The determinant bound \cite[\mbox{Theorem 1.3}]{PS} was originally claimed
 for the case that $m=n$ in \eqref{eq_PS_bound}. The bound for general $m$ immediately follows from \eqref{eq_PS_bound} without 
the coefficient $\<\bu_j,\bv_k\>_{\C^m}$. In fact it is sufficient to redefine $\cH$ by $\C^m \otimes \cH$ and $M_{x,y}$ by 
\begin{equation*}
\begin{split}
M_{x,y}:=&\<\bu_x\otimes f_x^0,\bv_y\otimes g_y^0\>_{\C^m \otimes\cH}+\sum_{l=1}^k\<\bu_x\otimes f_x^l,\bv_y\otimes g_y^l\>_{\C^m \otimes\cH}1_{\zeta_l'(x)\succ\zeta_l(y)}\\
&+\sum_{l=k+1}^{K+k}\<\bu_x\otimes f_x^l, \bv_y\otimes g_y^l\>_{\C^m \otimes\cH}1_{\zeta_l'(x)\succeq\zeta_l(y)}
\end{split}
\end{equation*}
for any vectors $\bu_x,\bv_x\in\C^m$ with $\|\bu_x\|_{\C^m},\|\bv_x\|_{\C^m}\le 1$.
\end{remark}

As an application of Lemma \ref{lem_abstract_determinant_bound} we prove the following statement.
\begin{proposition}\label{prop_extended_determinant_bound}
If the inequality
\begin{equation}\label{eq_imaginary_condition}
|\Im \cE_{\bk}|\le\frac{\pi}{2\beta}\ (\forall \bk\in\G^*)
\end{equation}
holds, then for any $m,n\in\N$ and $(\bx_j,\xi_j,x_j),(\by_j,\phi_j,y_j)\in\G\times\spin\times [0,\beta)$ $(\forall j\in\{1,\cdots,n\})$
\begin{equation}\label{eq_extended_determinant_bound}
\sup_{\substack{\bu_j,\bv_j\in\C^m\\ \|\bu_j\|_{\C^m},\|\bv_j\|_{\C^m}\le 1\\ \forall j\in \{1,\cdots,n\}}}|\det(\<\bu_j,\bv_k\>_{\C^m}C_{\cE}(\bx_j\xi_jx_j,\by_k\phi_ky_k))_{1\le j,k\le n}|\le 4^n.
\end{equation}
\end{proposition}
\begin{remark}The upper bound $4^n$ claimed above is same as that obtained in \cite[\mbox{Theorem 2.4}]{PS} for the covariance matrices with real-valued dispersion relations.
\end{remark}
\begin{proof}[Proof of Proposition \ref{prop_extended_determinant_bound}]
Our argument here closely follows \cite[\mbox{Section 4.1}]{PS}. 
Since $\G^*$ is a finite set, for a sufficiently small $\eps>0$ we may assume that $\Re(\cE_{\bk}+\eps)\neq 0$ for all $\bk\in\G^*$. Set $\cE_{\bk,\eps}:=\cE_{\bk}+\eps$ $(\forall \bk\in\G^*)$ and define the parameterized covariance $C_{\cE}^{\eps}$ by 
$$C_{\cE}^{\eps}(\bx\xi x,\by \phi y):=\frac{\delta_{\xi,\phi}}{L^d}\sum_{\bk\in\G^*}e^{i\<\bk,\by-\bx\>}e^{-(y-x){\cE}_{\bk,\eps}}\left(\frac{1_{y-x\le
0}}{1+e^{\beta {\cE}_{\bk,\eps}}}-\frac{1_{y-x>
0}}{1+e^{-\beta {\cE}_{\bk,\eps}}}\right).
$$
We construct a Gram representation of $C_{\cE}^{\eps}$ fitting in the form \eqref{eq_gram_representation} to apply Lemma \ref{lem_abstract_determinant_bound}. Observe the following decomposition. 
\begin{equation}\label{eq_covariance_decomposition}
\begin{split}
& C_{\cE}^{\eps}(\bx\xi x,\by \phi y)=\\
&1_{y-x\le 0}\frac{\delta_{\xi,\phi}}{L^d}\sum_{\bk\in\G^*}e^{i\<\bk,\by-\bx\>}\left(1_{\Re \cE_{\bk,\eps}>0}\frac{e^{-(\beta+y-x){\cE_{\bk,\eps}}}}{1+e^{-\beta{\cE_{\bk,\eps}}}}+1_{\Re \cE_{\bk,\eps}<0}\frac{e^{-(x-y){(-\cE_{\bk,\eps})}}}{1+e^{\beta{\cE_{\bk,\eps}}}}\right)\\
&-1_{y-x> 0}\frac{\delta_{\xi,\phi}}{L^d}\sum_{\bk\in\G^*}e^{i\<\bk,\by-\bx\>}\left(1_{\Re \cE_{\bk,\eps}>0}\frac{e^{-(y-x){\cE_{\bk,\eps}}}}{1+e^{-\beta{\cE_{\bk,\eps}}}}+1_{\Re \cE_{\bk,\eps}<0}\frac{e^{-(\beta+x-y){(-\cE_{\bk,\eps})}}}{1+e^{\beta{\cE_{\bk,\eps}}}}\right).\end{split}
\end{equation}
Note that for $t\ge 0$ and $\cE\in\C$ with $|\Im \cE|<\pi/\beta$ and $\Re \cE>0$
\begin{equation}\label{eq_fermi_formula}
\begin{split}
\frac{e^{-t\cE}}{1+e^{-\beta{\cE}}}&=\frac{e^{-it\Im \cE}(1+e^{-\beta{\overline{\cE}}})}{|1+e^{-\beta{\cE}}|^2}e^{-t\Re \cE}\\
&= \frac{e^{-it\Im \cE}(1+e^{-\beta{\overline{\cE}}})}{|1+e^{-\beta{\cE}}|^2}\frac{\Re \cE}{\pi}\int_{\R}dv\frac{e^{itv}}{v^2+(\Re \cE)^2},
\end{split}
\end{equation}
where $\overline{\cE}$ denotes the complex conjugate of $\cE$. 

By substituting the formula \eqref{eq_fermi_formula} into \eqref{eq_covariance_decomposition} we obtain
\begin{equation}\label{eq_covariance_reformation}
\begin{split}
& C_{\cE}^{\eps}(\bx\xi x,\by \phi y)=\frac{\delta_{\xi,\phi}}{L^d}\sum_{\bk\in\G^*}e^{i\<\bk,\by-\bx\>}\\
&\cdot\Bigg(1_{y-x\le 0}\Bigg(1_{\Re \cE_{\bk,\eps}>0}\frac{e^{-i(\beta +y -x)\Im \cE_{\bk,\eps}}(1+e^{-\beta{\overline{\cE_{\bk,\eps}}}})}{|1+e^{-\beta{\cE_{\bk,\eps}}}|^2}\frac{\Re \cE_{\bk,\eps}}{\pi}\int_{\R}dv\frac{e^{i(\beta+y-x)v}}{v^2+(\Re \cE_{\bk,\eps})^2}\\
&+1_{\Re \cE_{\bk,\eps}<0}\frac{e^{-i(x -y)\Im(- \cE_{\bk,\eps})}(1+e^{\beta\overline{\cE_{\bk,\eps}}})}{|1+e^{\beta{\cE_{\bk,\eps}}}|^2}\frac{\Re(- \cE_{\bk,\eps})}{\pi}\int_{\R}dv\frac{e^{i(x-y)v}}{v^2+(\Re (-\cE_{\bk,\eps}))^2}\Bigg)\\
&-1_{y-x> 0}\Bigg(1_{\Re \cE_{\bk,\eps}>0}\frac{e^{-i(y -x)\Im \cE_{\bk,\eps}}(1+e^{-\beta{\overline{\cE_{\bk,\eps}}}})}{|1+e^{-\beta{\cE_{\bk,\eps}}}|^2}\frac{\Re \cE_{\bk,\eps}}{\pi}\int_{\R}dv\frac{e^{i(y-x)v}}{v^2+(\Re \cE_{\bk,\eps})^2}\\
&+1_{\Re \cE_{\bk,\eps}<0}\frac{e^{-i(\beta+ x -y)\Im(- \cE_{\bk,\eps})}(1+e^{\beta\overline{\cE_{\bk,\eps}}})}{|1+e^{\beta{\cE_{\bk,\eps}}}|^2}\frac{\Re(- \cE_{\bk,\eps})}{\pi}\int_{\R}dv\frac{e^{i(\beta+x-y)v}}{v^2+(\Re (-\cE_{\bk,\eps}))^2}\Bigg)\Bigg).
\end{split}
\end{equation}

Let us define the Hilbert space $L^2(\G^*\times \spin\times\R)$ equipped with the inner product $\<\cdot,\cdot\>_{L^2(\G^*\times \spin\times\R)}$ given by
$$\<f,g\>_{L^2(\G^*\times \spin\times\R)}:=\frac{1}{L^d}\sum_{\bk\in\G^*}\sum_{\xi\in\spin}\int_{\R}dv \overline{f(\bk,\xi,v)}g(\bk,\xi,v).
$$
For $(\bx,\xi,x)\in\G\times\spin\times[-\beta,\beta]$ and $a\in\{-1,1\}$ define vectors $f_{\bx,\xi, x}^{a}$, $g_{\bx,\xi, x}^{a}\in L^2(\G^*\times \spin\times\R)$ by
\begin{equation*}
\begin{split}
&f_{\bx,\xi, x}^{a}(\bk,\phi,v)\\
&\ :=1_{\Re(a\cE_{\bk,\eps})>0}\delta_{\xi,\phi}e^{i\<\bk,\bx\>}e^{-ix(\Im(a\cE_{\bk,\eps})-v)}\frac{1+e^{-\beta a \cE_{\bk,\eps}} }{|1+e^{-\beta a\cE_{\bk,\eps}}|}\sqrt{\frac{|\Re \cE_{\bk,\eps}|}{\pi}}\frac{1}{iv +\Re(a\cE_{\bk,\eps})},\\
&g_{\bx,\xi, x}^{a}(\bk,\phi,v)\\
&\ :=1_{\Re(a\cE_{\bk,\eps})>0}\delta_{\xi,\phi}e^{i\<\bk,\bx\>}e^{-ix(\Im(a\cE_{\bk,\eps})-v)}\frac{1}{|1+e^{-\beta a\cE_{\bk,\eps}}|}\sqrt{\frac{|\Re \cE_{\bk,\eps}|}{\pi}}\frac{1}{iv +\Re(a\cE_{\bk,\eps})}.
\end{split}
\end{equation*}  
Note that
\begin{equation}\label{eq_vector_orthogonal}
\begin{split}
\<f_{\bx,\xi ,x}^{a},f_{\by, \phi, y}^{-a}\>_{L^2(\G^*\times\spin\times\R)}&=\<f_{\bx,\xi, x}^{a},g_{\by, \phi, y}^{-a}\>_{L^2(\G^*\times\spin\times\R)}\\
&=\<g_{\bx,\xi, x}^{a},g_{\by, \phi, y}^{-a}\>_{L^2(\G^*\times\spin\times\R)}=0,
\end{split}
\end{equation}
and
\begin{equation}\label{eq_vector_norm}
\begin{split}
&\|f_{\bx,\xi, x}^{a}\|^2_{L^2(\G^*\times\spin\times\R)}=\frac{1}{L^d}\sum_{\bk\in\G^*}1_{\Re(a\cE_{\bk,\eps})>0},\\
& \|g_{\bx,\xi, x}^{a}\|^2_{L^2(\G^*\times\spin\times\R)}=\frac{1}{L^d}\sum_{\bk\in\G^*}1_{\Re(a\cE_{\bk,,\eps})>0}\frac{1}{|1+e^{-\beta a\cE_{\bk,\eps}}|^2}.
\end{split}
\end{equation}
The equalities \eqref{eq_vector_orthogonal}-\eqref{eq_vector_norm} imply that
\begin{equation}\label{eq_gram_norm_estimate_f}
\|f_{\bx,\xi, x}^{a}+f_{\by,\phi, y}^{-a}\|^2_{L^2(\G^*\times\spin\times\R)}=\frac{1}{L^d}\sum_{\bk\in\G^*}\left(1_{\Re(a\cE_{\bk,\eps})>0}+1_{\Re(-a\cE_{\bk,\eps})>0}\right)= 1
\end{equation}
and under the assumption that $|\Im \cE_{\bk}|\le \pi/(2\beta)$ 
\begin{equation}\label{eq_gram_norm_estimate_g}                    
\begin{split}
&\|g_{\bx,\xi, x}^{a}+g_{\by,\phi, y}^{-a}\|^2_{L^2(\G^*\times\spin\times\R)}\\
&\le\frac{1}{L^d}\sum_{\bk\in\G^*}\Bigg(1_{\Re(a\cE_{\bk,\eps})>0}\frac{1}{e^{-2\beta\Re(a\cE_{\bk,\eps})}+2\cos(\beta \Im \cE_{\bk,\eps})e^{-\beta\Re(a\cE_{\bk,\eps})}+1}\\
&\qquad\qquad\qquad +1_{\Re(-a\cE_{\bk,\eps})>0}\frac{1}{e^{2\beta\Re(a\cE_{\bk,\eps})}+2\cos(\beta \Im \cE_{\bk,\eps})e^{\beta\Re(a\cE_{\bk,\eps})}+1}\Bigg)\\
&\le\frac{1}{L^d}\sum_{\bk\in\G^*}\left(1_{\Re(a\cE_{\bk,\eps})>0}\frac{1}{e^{-2\beta\Re(a\cE_{\bk,\eps})}+1}+1_{\Re(-a\cE_{\bk,\eps})>0}\frac{1}{e^{2\beta\Re(a\cE_{\bk,\eps})}+1}\right)\le 1
\end{split}
\end{equation}
for all $(\bx,\xi,x),(\by,\phi,y)\in\G\times\spin\times[-\beta,\beta]$ and $a\in\{-1,1\}$.

The equality \eqref{eq_covariance_reformation} can be written as
\begin{equation}\label{eq_covariance_gram}
\begin{split}
C_{\cE}^{\eps}(\bx\xi x,\by \phi y)
&=1_{y-x\le 0}\<f_{\bx,\xi, x-\beta}^{1}+f_{\bx,\xi, -x}^{-1},g_{\by,\phi, y}^{1}+g_{\by,\phi, -y}^{-1}\>_{L^2(\G^*\times\spin\times\R)}\\
&\quad +1_{y-x> 0}\<-f_{\bx,\xi, x}^{1}-f_{\bx,\xi, -x}^{-1},g_{\by,\phi, y}^{1}+g_{\by,\phi, \beta -y}^{-1}\>_{L^2(\G^*\times\spin\times\R)}.
\end{split}
\end{equation}
The representation \eqref{eq_covariance_gram} and the bounds \eqref{eq_gram_norm_estimate_f}-\eqref{eq_gram_norm_estimate_g} enable us to apply Lemma \ref{lem_abstract_determinant_bound} to obtain the bound \eqref{eq_extended_determinant_bound} for $C_{\cE}^{\eps}$. Then, by sending $\eps\searrow 0$ and the continuity of $C_{\cE}^{\eps}$ with respect to $\eps$ we complete the proof.
\end{proof}

\subsection{Exponential decay of the covariance}\label{subsec_exponential_decay_covariance}
In this subsection we show the exponential decay property of the covariance and find a volume-independent upper bound on its integral as the corollary.

Later in our analysis we will need to deal with a modified covariance whose dispersion relation is given by $E_{\bk+z\be_p}$ with $z\in\C$, $p\in\{1,\cdots,d\}$. It is convenient to prove the $L^1$-bound on the modified covariances which include the normal covariance \eqref{eq_covariance} with the real-valued dispersion relation \eqref{eq_dispersion_relation} as a special case at this stage.  

Furthermore, to bound the covariance with the dispersion relation $E_{\bk+z\be_p}$ needs to control the covariance with the perturbed dispersion relation $E_{\bk+z\be_p+w\be_q}$ defined by
\begin{equation*}
\begin{split}
&E_{\bk+z\be_p+w\be_q}:=-2t\sum_{j=1}^d\cos(\<\bk+z\be_p+w\be_q,\be_j\>)\\
&\qquad - 4t'\cdot 1_{d\ge 2}\sum_{j,k = 1\atop j<k}^d\cos(\<\bk+z\be_p+w\be_q,\be_j\>)\cos(\<\bk+z\be_p+w\be_q,\be_k\>)-\mu
\end{split}
\end{equation*}
for $z,w\in\C$ and $p,q\in\{1,\cdots,d\}$. Let us study properties of $E_{\bk+z\be_p+w\be_q}$.

First we clarify a sufficient condition for the inequalities
$$|\Im (E_{\bk+z\be_p+w\be_q})|< \frac{\pi}{\beta},\ |\Im (E_{\bk+z\be_p+w\be_q})|\le \frac{\pi}{2\beta},$$
which will ensure the analyticity of the modified covariance with respect to the variables $z,w$ and make the determinant bound Proposition \ref{prop_extended_determinant_bound} applicable.
 
\begin{lemma}\label{lem_dispersion_relation_condition}
Let $z,w\in\C$ and $r>0$. If $|\Im z|,|\Im w|< \frac{1}{2}\log F_{t,t',d}(r)$ with the function $F_{t,t',d}(\cdot)$ defined in \eqref{eq_definition_keyfunction}, then $|\Im E_{\bk+z\be_p+w\be_q}|<r$ holds for all $\bk\in\G^*$ and $p,q\in\{1,\cdots,d\}$. The parallel statement with  `$\le$'s in place of `$<$'s holds as well.
\end{lemma}
\begin{proof}
Let us write $z=a+b i$, $w=u+v i$ $(a,b,u,v\in\R)$. By using the formula 
$$\cos(x+iy)=\cos(x)\cosh(y)-i\sin(x)\sinh(y)\ (x,y\in\R)$$
and writing $\bk=(k_1,\cdots,k_d)$,
we find that
\begin{equation*}
\begin{split}
\Im E_{\bk+z\be_p+w\be_q} 
&= 2t\sum_{j=1}^d\sin(k_j+a\delta_{p,j}+u\delta_{q,j})\sinh(b\delta_{p,j}+v\delta_{q,j})\\
&\quad+4t'1_{d\ge 2}\sum_{j,l=1\atop j\neq l}^d\cos(k_j+a\delta_{p,j}+u\delta_{q,j})\cosh(b\delta_{p,j}+v\delta_{q,j})\\
&\qquad\qquad\qquad\quad\cdot\sin(k_l+a\delta_{p,l}+u\delta_{q,l})\sinh(b\delta_{p,l}+v\delta_{q,l}).
\end{split}
\end{equation*}

Assume that $|b|,|v|\le s$ for $s>0$. 
\begin{equation*}
\begin{split}
&|\Im E_{\bk+z\be_p+w\be_q}|\\
&\le 2|t|\sum_{j=1}^d\sinh(s(\delta_{p,j}+\delta_{q,j}))+4|t'|1_{d\ge 2}\sum_{j,l=1\atop j\neq l}^d\cosh(s(\delta_{p,j}+\delta_{q,j}))\sinh(s(\delta_{p,l}+\delta_{q,l}))\\
&\le 1_{p=q}(2|t|\sinh(2s)+4|t'|(d-1)\sinh(2s))\\
&\quad + 1_{p\neq q}(4|t|\sinh(s)+8|t'|(d-1)\cosh(s)\sinh(s))\\
&\le 2\sinh(2s)(|t|+2(d-1)|t'|),
\end{split}
\end{equation*}
where we used the facts that for $s\ge0$
$$1\le \cosh(s),\ 2\sinh(s)\cosh(s)=\sinh(2s),\ 2\sinh(s)\le \sinh(2s).$$

Thus, the inequality
\begin{equation}\label{eq_sinh_inequality}
2\sinh(2s)(|t|+2(d-1)|t'|)<r
\end{equation}
implies that $|\Im E_{\bk+z\be_p+w\be_q}|<r$. Since 
$$\sinh^{-1}(X)=\log\left(X+\sqrt{X^2+1}\right) (\forall X\in\R),$$
the inequality \eqref{eq_sinh_inequality} is equivalent to the inequality $s<\frac{1}{2}\log F_{t,t',d}(r)$. 
\end{proof}

We define the perturbed covariance $C(\bx\xi x,\by \phi y)(z\be_p+w\be_q)$ as follows.
\begin{equation*}
\begin{split}
C(\bx\xi x,\by\phi y)(z\be_p+w\be_q):=&\frac{\delta_{\xi,\phi}}{L^d}\sum_{\bk\in\G^*}e^{i\<\bk,\by-\bx\>}e^{-(y-x){E}_{\bk+z\be_p+w\be_q}}\\
&\cdot\left(\frac{1_{y-x\le
0}}{1+e^{\beta {E}_{\bk+z\be_p+w\be_q}}}-\frac{1_{y-x>
0}}{1+e^{-\beta {E}_{\bk+z\be_p+w\be_q}}}\right)
\end{split}
\end{equation*}
for any $(\bx,\xi,x),(\by,\phi,y)\in\G\times\spin\times[0,\beta)$, $p,q\in\{1,\cdots,d\}$ and $z,w\in \C$. Moreover, if $x,y\in[0,\beta)_h$ we write 
$$C_h(\bx\xi x,\by\phi y)(z\be_p+w\be_q)=C(\bx\xi x,\by\phi y)(z\be_p+w\be_q).$$The properties of $C(z\be_p+w\be_q)$ are summarized as follows.
\begin{lemma}\label{lem_properties_covariance}
\begin{enumerate}[(i)]
\item\label{item_analyticity_covariance} For any $(\bx,\xi,x),(\by,\phi,y)\in\G\times\spin\times[0,\beta)$, $p,q\in\{1,\cdots,d\}$, the function $(z,w)\mapsto C(\bx\xi x,\by \phi y)(z\be_p+w\be_q)$ is analytic in the domain 
$$\left\{(z,w)\in\C^2\ \Big|\ |\Im z|,|\Im w|<\frac{1}{2}\log F_{t,t',d}\left(\frac{\pi}{\beta}\right)\right\}.$$
\item\label{item_bound_covariance} The inequality $|C(\bx\xi x,\by \phi y)(z\be_p+w\be_q)|\le 1$ holds for any $(\bx,\xi,x),(\by,\phi,y)\in\G\times\spin\times[0,\beta)$, $p,q\in\{1,\cdots,d\}$ and $z,w\in\C$ with $|\Im z|,|\Im w|\le \frac{1}{2}\log F_{t,t',d}(\pi/(2\beta))$.
\item\label{item_determinantbound_covariance}  For any $z,w\in\C$ with $|\Im z|,|\Im w|\le\frac{1}{2}\log F_{t,t',d}(\pi/(2\beta))$ and $p,q\in$\\
$\{1,\cdots,d\}$, the covariance $C(z\be_p+w\be_q)$ satisfies the determinant bound \eqref{eq_extended_determinant_bound}.
\item\label{item_determinant_covariance} For any  $p,q\in\{1,\cdots,d\}$, $z,w\in\C$ with $|\Im z|,|\Im w|< \frac{1}{2}\log F_{t,t',d}(\pi/\beta)$ and $h\in 2\N/\beta$,
$$\det(C_h(\bx\xi x,\by \phi y))_{(\bx,\xi, x),(\by,\phi, y)\in\G\times\spin\times[0,\beta)_h}=\prod_{\bk\in\G^*}\frac{1}{(1+e^{\beta {E}_{\bk+z\be_p+w\be_q}})^2}\neq 0.$$
\end{enumerate}
\end{lemma}

\begin{remark}
In \eqref{item_determinant_covariance} and in its proof below let us think that each element of $\G\times\spin\times[0,\beta)_h$ has been given a number from $1$ to $N(=2\beta h L^d)$ and according to this numbering the $N\times N$ matrix $C_h(z\be_p+w\be_q)$ is defined. However, the claim \eqref{item_determinant_covariance} implies that the value of the determinant is independent of how to number the elements of $\G\times\spin\times[0,\beta)_h$.
\end{remark}
\begin{proof}[Proof of Lemma \ref{lem_properties_covariance}]
If $|\Im z|,|\Im w|<\frac{1}{2}\log F_{t,t',d}(\pi/\beta)$, Lemma \ref{lem_dispersion_relation_condition} ensures that $|\Im {E}_{\bk+z\be_p+w\be_q}|<\pi/\beta$. Then,  for all $\bk\in\G^*$ and $a\in\{1,-1\}$
$$|1+e^{\beta a E_{\bk+z\be_p+w\be_q}}|^2>\left(e^{\beta a\Re E_{\bk+z\be_p+w\be_q}}-1\right)^2\ge 0.$$
Thus, the denominators in $C(\bx\xi x, \by\phi y)(z\be_p+w\be_q)$ do not vanish, which proves \eqref{item_analyticity_covariance}.

If $|\Im z|,|\Im w|\le\frac{1}{2}\log F_{t,t',d}(\pi/(2\beta))$, by Lemma \ref{lem_dispersion_relation_condition} $|\Im {E}_{\bk+z\be_p+w\be_q}|\le$\\
$\pi/(2\beta)$. Thus, the claim \eqref{item_determinantbound_covariance} holds true by Proposition \ref{prop_extended_determinant_bound}. Moreover, for any $x\in [0,\beta)$ and $\bk\in\G^*$
$$\left|\frac{e^{xE_{\bk+z\be_p+w\be_q}}}{1+e^{\beta E_{\bk+z\be_p+w\be_q}}}\right|\le \frac{e^{x\Re E_{\bk+z\be_p+w\be_q}}}{(1+e^{2\beta \Re E_{\bk+z\be_p+w\be_q}})^{1/2}}\le 1,$$
from which the claim \eqref{item_bound_covariance} follows.

To show \eqref{item_determinant_covariance}, define an $h$-dependent finite set $M_h$ of Matsubara frequencies by
$$M_h:=\left\{\o \in \frac{\pi (2\Z+1)}{\beta}\ \big|\ -\pi h< \o < \pi h\right\}.$$
Define the unitary matrix $Y=$\\
$(Y(\bk\phi\o,\bx\xi x))_{(\bk,\phi,\o)\in\G^*\times\spin\times M_h, (\bx,\xi, x)\in\G\times\spin\times[0,\beta)_h}$ by
$$Y(\bk\phi\o,\bx\xi x):=\frac{\delta_{\phi,\xi}}{\sqrt{\beta hL^d}}e^{i\<\bk,\bx\>}e^{-i\o x}.$$
Using the assumption that $h\in 2\N/\beta$, the argument parallel to \cite[\mbox{Appendix C}]{K} demonstrates that
\begin{equation}\label{eq_diagonalization_covariance}
(YC_hY^*)(\bk\phi \o,\hat{\bk}\hat{\phi}\hat{\o})=\delta_{\bk,\hat{\bk}}\delta_{\phi,\hat{\phi}}\delta_{\o,\hat{\o}}\frac{1}{1-e^{-i\o/h+E_{\bk+z\be_p + w\be_q}/h}}.
\end{equation}
As in \cite[\mbox{Proposition C.7}]{K}, the diagonalization \eqref{eq_diagonalization_covariance} deduces \eqref{item_determinant_covariance}.
\end{proof} 

Here we present a calculation, which shows the essence of our analysis to bound the correlation functions. Take any $n\in\N$ and set $r_n:=\frac{1}{2n}\log F_{t,t',d}(\beta/(2\pi))>0$. 
 By periodicity we observe that for any $p,q\in\{1,\cdots,d\}$ and $z\in\C$ with $|\Im z|\le \frac{1}{2}\log F_{t,t',d}(\beta/(2\pi))$
\begin{equation}\label{eq_reform_covariance}
e^{i\frac{2\pi}{L}\<\bx-\by,\be_q\>}C(\bx\xi x,\by \phi y)(z\be_p)=C(\bx\xi x,\by \phi y)\left(z\be_p+\frac{2\pi}{L}\be_q\right).
\end{equation}
Moreover, by Cauchy's integral formula we have
\begin{equation*}
\begin{split}
\frac{e^{i\frac{2\pi}{L}\<\bx-\by,\be_q\>}-1}{2\pi/L}&C(\bx\xi x,\by \phi y)(z\be_p)=\frac{L}{2\pi}\int_{0}^{2\pi/L}d\theta\frac{d}{d\theta}C(\bx\xi x,\by \phi y)(z\be_p+\theta \be_q)\\
&=\frac{L}{2\pi}\int_{0}^{2\pi/L}d\theta\frac{1}{2\pi i}\oint_{|w-\theta|=r_n}dw\frac{C(\bx\xi x,\by \phi y)(z\be_p+w \be_q)}{(w-\theta)^2},
\end{split}
\end{equation*}
where $\oint_{|w-\theta|=r_n}dw$ stands for the contour integral along the contour $\{w\in\C\ |\ |w-\theta|=r_n\}$ oriented counter clock-wise. With this choice of $r_n$, Lemma \ref{lem_properties_covariance} \eqref{item_analyticity_covariance} allows us to repeat this operation $n$ times to obtain

\begin{lemma}\label{lem_covariance_integral_formula}
For any $z\in\C$ with $|\Im z|\le  \frac{1}{2}\log F_{t,t',d}(\beta/(2\pi))$ and $p,q\in \{1,\cdots,d\}$
\begin{equation}\label{eq_covariance_integral_formula}
\begin{split}
&\left(\frac{e^{i\frac{2\pi}{L}\<\bx-\by,\be_q\>}-1}{2\pi/L}\right)^nC(\bx\xi x,\by \phi y)(z\be_p)\\
&\quad=\prod_{j=1}^n\left(\frac{L}{2\pi}\int_{0}^{2\pi/L}d\theta_j\frac{1}{2\pi i}\oint_{|w_j-\theta_j|=r_n}dw_j\frac{1}{(w_j-\theta_j)^2}\right)\\
&\quad\qquad\cdot C(\bx\xi x,\by \phi y)\left(z\be_p+\sum_{j=1}^nw_j\be_q\right),
\end{split}
\end{equation}
where $r_n=\frac{1}{2n}\log F_{t,t',d}(\beta/(2\pi))$.
\end{lemma}

Evaluating both sides of \eqref{eq_covariance_integral_formula} leads to the following bounds.

\begin{proposition}\label{prop_covariance_exponential_decay}
For any $z\in\C$ with $|\Im z|\le  \frac{1}{2}\log F_{t,t',d}(\beta/(2\pi))$, $p\in\{1,\cdots,d\}$, the following inequalities hold.
\begin{equation}\label{eq_covariance_exponential_decay}
|C(\bx\xi x,\by \phi y)(z\be_p)|\le 2 F_{t,t',d}\left(\frac{\pi}{2\beta}\right)^{-\frac{1}{4ed}\sum_{q=1}^d\left|\frac{e^{i2\pi\<\bx-\by,\be_q\>/L}-1}{2\pi/L}\right|}
\end{equation}
for all $(\bx,\xi,x),(\by,\phi,y)\in\G\times\spin\times[0,\beta)$. Especially,
\begin{equation}\label{eq_covariance_exponential_decay_limit}
|C(\bx\xi x,\by \phi y)(z\be_p)|\le2  F_{t,t',d}\left(\frac{\pi}{2\beta}\right)^{-\frac{1}{2e \pi d}\sum_{q=1}^d|\<\bx-\by,\be_q\>|}
\end{equation}
if $\bx-\by\in (\{-\lfloor L/2 \rfloor,-\lfloor L/2\rfloor + 1,\cdots,-\lfloor L/2\rfloor+L-1\})^d$.  
\end{proposition}
\begin{proof}
By using Lemma \ref{lem_properties_covariance} \eqref{item_bound_covariance} and the inequality 
\begin{equation}\label{eq_evaluation_multi_integral}
\left|\prod_{j=1}^n\left(\frac{L}{2\pi}\int_{0}^{2\pi/L}d\theta_j\frac{1}{2\pi i}\oint_{|w_j-\theta_j|=r_n}dw_j\frac{1}{(w_j-\theta_j)^2}\right)\right|\le\left(\frac{2n}{\log F_{t,t',d}(\pi/(2\beta))}\right)^n,\end{equation}
we have
\begin{equation}\label{eq_covariance_difference_evaluation}
|(\text{The right hand side of \eqref{eq_covariance_integral_formula}})|\le\left(\frac{2n}{\log F_{t,t',d}(\pi/(2\beta))}\right)^n.
\end{equation}
Stirling's formula states that for any $n\in \N$ there exists $\eps(n)\in (0,1)$ such that 
\begin{equation}\label{eq_stirling_formula}
n^n=n!(2\pi n)^{-1/2}e^{n-\eps(n)/(12n)}
\end{equation}
(see, e.g, \cite[\mbox{Eq.6.1.38}]{AS}). 
On the convention that $0^0=0!=1$, the formula \eqref{eq_stirling_formula} gives the inequality
\begin{equation}\label{eq_stirling_inequality}
n^n\le n!e^n \ (\forall n\in\N\cup\{0\}).
\end{equation}
By combining \eqref{eq_stirling_inequality} with \eqref{eq_covariance_difference_evaluation} we obtain
\begin{equation}\label{eq_covariance_weight_evaluation}
\frac{1}{n!}\left(\frac{\log F_{t,t',d}(\pi/(2\beta))}{2e}\right)^n\left|\frac{e^{i\frac{2\pi}{L}\<\bx-\by,\be_q\>}-1}{2\pi/L}\right|^n|C(\bx\xi x,\by \phi y)(z\be_p)|\le 1.
\end{equation}
By multiplying both sides of \eqref{eq_covariance_weight_evaluation} by $2^{-n}$ and summing over $\N\cup\{0\}$, we have 
$$
|C(\bx\xi x,\by \phi y)(z\be_p)|\le 2 F_{t,t',d}\left(\frac{\pi}{2\beta}\right)^{-\frac{1}{4e}\left|\frac{e^{i2\pi\<\bx-\by,\be_q\>/L}-1}{2\pi/L}\right|}
$$
for any $q\in\{1,\cdots,d\}$, which yields \eqref{eq_covariance_exponential_decay}. By applying the inequality that $|e^{i\theta}-1|\ge 2|\theta|/\pi$ $(\forall \theta\in [-\pi,\pi])$ to \eqref{eq_covariance_exponential_decay}, we can derive \eqref{eq_covariance_exponential_decay_limit}.
\end{proof}

\begin{corollary}\label{cor_covariance_L1_bound}
For any $z\in\C$ with $|\Im z|\le \frac{1}{2}\log F_{t,t',d}(\beta/(2\pi))$ and $p\in\{1,\cdots,d\}$, 
\begin{equation}\label{eq_covariance_L1_bound}
\frac{1}{h}\sum_{x\in[-\beta,\beta)_h}\sum_{\bx\in\G}|C_h(\bx\xi x,\b0\xi 0)(z\be_p)| \le 4\beta \left(\frac{F_{t,t',d}\left(\frac{\pi}{2\beta}\right)^{1/(2e\pi d)}+1}{F_{t,t',d}\left(\frac{\pi}{2\beta}\right)^{1/(2e\pi d)}-1}\right)^{d}.
\end{equation}
\end{corollary} 
\begin{proof}
By using \eqref{eq_covariance_exponential_decay_limit} and periodicity we have that
\begin{equation*}
\begin{split}
&\sum_{\bx\in\G}|C(\bx\xi x,\b0 \xi 0)(z\be_p)|=\sum_{x_j=-\lfloor L/2\rfloor\atop \forall j\in\{1,\cdots,d\}}^{-\lfloor L/2\rfloor+L-1}|C(\bx\xi x,\b0 \xi 0)(z\be_p)|\\
&\quad\le 2\left(1+ 2\sum_{l=1}^{\infty}F_{t,t',d}\left(\frac{\pi}{2\beta}\right)^{-\frac{l}{2e\pi d}}\right)^d= 2\left(\frac{F_{t,t',d}\left(\frac{\pi}{2\beta}\right)^{1/(2e\pi d)}+1}{F_{t,t',d}\left(\frac{\pi}{2\beta}\right)^{1/(2e\pi d)}-1}\right)^d,
\end{split}
\end{equation*}
which gives \eqref{eq_covariance_L1_bound}.
\end{proof}
\begin{remark}
Remark \ref{rem_beta_dependency} and the inequality \eqref{eq_covariance_L1_bound} imply that
$$\frac{1}{h}\sum_{x\in [-\beta,\beta)_h}\sum_{\bx\in\G}|C_h(\bx\xi x,\b0\xi 0)(z\be_p)| =O(\beta^{d+1})$$
as $\beta \to +\infty$.
\end{remark}

\begin{remark}
It is the same procedure as the proof of Proposition \ref{prop_covariance_exponential_decay} to derive the decay bound on the determinant of the covariance. By using the equality that
 \begin{equation*}
\begin{split}
&e^{i\frac{2\pi}{L}\<\sum_{j=1}^{n}\bx_j-\sum_{j=1}^{n}\by_j,\be_q\>}\det(C(\bx_j\xi_jx_j,\by_k\phi_ky_k))_{1\le j,k\le n}\\
&\qquad=
\det\left(C(\bx_j\xi_jx_j,\by_k\phi_ky_k)\left(\frac{2\pi}{L}\be_q\right)\right)_{1\le j,k\le n}
\end{split}
\end{equation*}
and the determinant bound Proposition \ref{prop_extended_determinant_bound} one can prove that for any $(\bx_j,\xi_j,x_j)$,\\
$(\by_j,\phi_j,y_j)\in\G\times\spin\times [0,\beta)$ $(\forall j\in\{1,\cdots,n\})$
\begin{equation*}
\begin{split}
&\left|\det(C(\bx_j\xi_jx_j,\by_k\phi_ky_k))_{1\le j,k\le n}\right|\\
&\quad\le 2\cdot 4^{n}F_{t,t',d}\left(\frac{\pi}{2\beta}\right)^{-\frac{1}{4ed}\sum_{q=1}^d\left|\frac{e^{i2\pi\<\sum_{j=1}^n\bx_j-\sum_{j=1}^n\by_j,\be_q\>/L}-1}{2\pi/L}\right|}.
\end{split}
\end{equation*}
\end{remark}

\subsection{Proof of the main theorems}\label{subsec_exponential_decay_correlation}
We prepare some lemmas and give the proofs of Theorem \ref{thm_exponential_decay} and Theorem \ref{thm_exponential_decay_hubbard_model}.
The following lemma fundamentally supports the validity of our argument in this subsection.
\begin{lemma}\label{lem_schwinger_functional_analytic}
For any $R_r>0$, $R_i\in (0,\frac{1}{2}\log F_{t,t',d}(\pi/\beta))$ and $p\in\{1,\cdots,d\}$ there exists $Q>0$ such that the function
$$(z,\eta)\longmapsto S_{\hX^{\hm},\hY^{\hm},\hXi^{\hm},\hPhi^{\hm}}(C_h(z\be_p),\eta)$$
is analytic in the domain
$$\{ (z,\eta)\in\C^2\ |\ |\Re z|<R_r, \ |\Im z|<R_i,\ |\eta|<Q\}.$$
\end{lemma}
\begin{proof}
Note that by Lemma \ref{lem_properties_covariance} \eqref{item_analyticity_covariance} for any $R_r>0$ and $R_i\in (0,\frac{1}{2}\log F_{t,t',d}(\pi/\beta))$ the function $z\mapsto C_h(\bx\xi x,\by\phi y)(z\be_p)$ is bounded in the compact set $\{z\in\C\ |\ |\Re z|\le R_r,\ |\Im z|\le R_i\}$.
Thus, by Lemma \ref{lem_grassmann_gaussian_equality}
\begin{equation*}
\begin{split}
&\lim_{\eta\to 0}\sup_{z\in\C\atop |\Re z|\le R_r,\ |\Im z|\le R_i}\\
&\cdot\Bigg|\int e^{\eta\sum_{l=1}^{\tn}\sum_{(X^l,\Xi^l,\Phi^l)\in\G^{l}\times\spin^{2l}}U_{L,l}(X^l,\Xi^l,\Phi^l)V_{h,X^l,X^l,\Xi^l,\Phi^l}^l(\opsi_X,\psi_X)}d\mu_{C_h(z\be_p)}(\opsi_X,\psi_X)\\
&\qquad-1\Bigg|=0.
\end{split}
\end{equation*}
Hence, we can take a small $Q>0$ such that the denominator of 
$S_{\hX^{\hm},\hY^{\hm},\hXi^{\hm},\hPhi^{\hm}}$\\
$(C_h(z\be_p),\eta)$ does not vanish for any $\eta\in\C$ with $|\eta|<Q$ and $z\in\C$ with $|\Re z|< R_r$, $|\Im z|< R_i$. This proves the analyticity of $S_{\hX^{\hm},\hY^{\hm},\hXi^{\hm},\hPhi^{\hm}}(C_h(z\be_p),\eta)$ in the same domain.
\end{proof}

For $p\in\{1,\cdots,d\}$ we define the matrix $\cU_p=(\cU_p((\bx\xi x)_j,(\bx\xi x)_k))_{1\le j,k\le N}$ by
$$\cU_p(\bx\xi x,\by \phi y):=\delta_{\bx,\by}\delta_{\xi,\phi}\delta_{x,y}e^{i2\pi\<\bx,\be_p\>/L}.$$
The following equality is based on the $U(1)$-invariance property of the Grassmann integral $\int\cdot d\psi_Xd\opsi_X$.
\begin{lemma}\label{lem_U_1_symmetry}
For any $p\in\{1,\cdots,d\}$, $f(\opsi_X,\psi_X)\in\C[\opsi_{(\bx\xi x)_j},\psi_{(\bx\xi x)_j}\ |\ j\in\{1,\cdots,N\}]$ and $z\in\C$ with $|\Im z|\le\frac{1}{2}\log F_{t,t',d}(\pi/(2\beta))$, the following equality holds.
\begin{equation*}
\begin{split}
\int& f((\overline{\cU_p}\opsi_X^t)^t,(\cU_p\psi_X^t)^t)d \mu_{C_h((z+2\pi/L)\be_p)}(\opsi_X,\psi_X)\\
&=\int f(\opsi_X,\psi_X)d\mu_{C_h(z\be_p)}(\opsi_X,\psi_X).
\end{split}
\end{equation*}
\end{lemma}
\begin{proof}
Observe the invariance that
$$\int f((\overline{\cU_p}\opsi_X^t)^t,(\cU_p\psi_X^t)^t)d\psi_Xd\opsi_X=\int f(\opsi_X,\psi_X)d\psi_Xd\opsi_X,$$
which implies that 
\begin{equation}\label{eq_grassmann_reform}
\begin{split}
\int& f(\opsi_X,\psi_X)d\mu_{C_h(z\be_p)}(\opsi_X,\psi_X)\\
&=\int f((\overline{\cU_p}\opsi_X^t)^t,(\cU_p\psi_X^t)^t)e^{-\<{\cU_p}\psi_X^t,C_h(z\be_p)^{-1}\overline{\cU_p}\opsi_X^t\>}d\psi_Xd\opsi_X\\
&\quad\cdot\Bigg/\int e^{-\<{\cU_p}\psi_X^t,C_h(z\be_p)^{-1}\overline{\cU_p}\opsi_X^t\>}d\psi_Xd\opsi_X.
\end{split}
\end{equation}
Moreover, by using the equality \eqref{eq_reform_covariance} we have
\begin{equation}\label{eq_covariance_grassmann_reform}
\begin{split}
\<{\cU_p}\psi_X^t,C_h(z\be_p)^{-1}\overline{\cU_p}\opsi_X^t\>&=\<\psi_X^t,({\cU_p}C_h(z\be_p)\overline{\cU_p})^{-1}\opsi_X^t\>\\
&=\<\psi_X^t,C_h((z+2\pi/L)\be_p)^{-1}\opsi_X^t\>.
\end{split}
\end{equation}
By substituting \eqref{eq_covariance_grassmann_reform} into \eqref{eq_grassmann_reform}, we obtain the desired equality.
\end{proof}

Combining Lemma \ref{lem_schwinger_functional_analytic} with Lemma \ref{lem_U_1_symmetry} shows
\begin{lemma}\label{lem_schwinger_functional_formula}
For any $n\in\N$, set $r_n:=\frac{1}{2n}\log F_{t,t',d}(\pi/(2\beta))$. Let $p\in \{1,\cdots,d\}$. There exists $Q_n>0$ such that for any $\eta\in\C$ with $|\eta|<Q_n$, 
\begin{equation}\label{eq_schwinger_functional_formula}
\begin{split}
&\left(\frac{e^{i\frac{2\pi}{L}\<\sum_{j=1}^{\hm}\hbx_j-\sum_{j=1}^{\hm}\hby_j,\be_p\>}-1}{2\pi/L}\right)^nS_{\hX^{\hm},\hY^{\hm},\hXi^{\hm},\hPhi^{\hm}}(C_h,\eta)\\
&\quad = \prod_{j=1}^n\left(\frac{L}{2\pi}\int_{0}^{2\pi/L}d\theta_j\frac{1}{2\pi i}\oint_{|z_j-\theta_j|=r_n}dz_j\frac{1}{(z_j-\theta_j)^2}\right)\\ 
&\qquad\qquad\cdot S_{\hX^{\hm},\hY^{\hm},\hXi^{\hm},\hPhi^{\hm}}\left(C_h\left(\sum_{j=1}^nz_j\be_p\right),\eta\right),
\end{split}
\end{equation}
where $\oint_{|z_j-\theta_j|=r_n}dz_j$ stands for the contour integral along the contour $\{z_j\in\C\ |\ |z_j-\theta_j|=r_n\}$ oriented counter clock-wise.
\end{lemma}
\begin{proof}
Take any $R_i\in (\frac{1}{2}\log F_{t,t',d}(\pi/(2\beta)),\frac{1}{2}\log F_{t,t',d}(\pi/\beta))$ and $R_r>nr_n+\frac{2\pi n}{L}$. By Lemma \ref{lem_schwinger_functional_analytic} we can take a small $Q_n>0$ so that $(z,\eta)\mapsto S_{\hX^{\hm},\hY^{\hm},\hXi^{\hm},\hPhi^{\hm}}$\\$(C_h(z\be_p),\eta)$ is analytic in the domain $\{(z,\eta)\in\C^2\ |\ |\Re z|<R_r,\ |\Im z|<R_i,\ \  |\eta|<Q_n\}$. 

Fix any $\eta \in\C$ with $|\eta|<Q_n$. By noting this domain of analyticity, Lemma \ref{lem_U_1_symmetry}, the fundamental theorem of calculus and Cauchy's integral theorem verify the following equalities.
\begin{equation}\label{eq_schwinger_integral_1}
\begin{split}
&\frac{e^{i\frac{2\pi}{L}\<\sum_{j=1}^{\hm}\hbx_j-\sum_{j=1}^{\hm}\hby_j,\be_p\>}-1}{2\pi/L}S_{\hX^{\hm},\hY^{\hm},\hXi^{\hm},\hPhi^{\hm}}(C_h,\eta)\\
&\quad= \frac{L}{2\pi}\left(S_{\hX^{\hm},\hY^{\hm},\hXi^{\hm},\hPhi^{\hm}}\left(C_h\left(\frac{2\pi}{L}\be_p\right),\eta\right)-S_{\hX^{\hm},\hY^{\hm},\hXi^{\hm},\hPhi^{\hm}}\left(C_h,\eta\right)\right)\\
&\quad= \frac{L}{2\pi}\int^{2\pi/L}_0d\theta_1\frac{d}{d\theta_1}S_{\hX^{\hm},\hY^{\hm},\hXi^{\hm},\hPhi^{\hm}}\left(C_h\left(\theta_1\be_p\right),\eta\right)\\
&\quad=  \frac{L}{2\pi}\int^{2\pi/L}_0d\theta_1\frac{1}{2\pi i}\oint_{|z_1-\theta_1|=r_n}dz_1 \frac{1}{(z_1-\theta_1)^2}S_{\hX^{\hm},\hY^{\hm},\hXi^{\hm},\hPhi^{\hm}}\left(C_h\left(z_1\be_p\right),\eta\right).
\end{split}
\end{equation}

Then, we multiply both sides of \eqref{eq_schwinger_integral_1} by $(e^{i\frac{2\pi}{L}\<\sum_{j=1}^{\hm}\hbx_j-\sum_{j=1}^{\hm}\hby_j,\be_p\>}-1)/(2\pi/L)$ and do the same calculation for the integrand $S_{\hX^{\hm},\hY^{\hm},\hXi^{\hm},\hPhi^{\hm}}\left(C_h\left(z_1\be_p\right),\eta\right)$. Repeating this procedure $n$ times results in \eqref{eq_schwinger_functional_formula}.
\end{proof}

\begin{lemma}\label{lem_schwinger_functional_integral_bound}
For any $m,n\in\N\cup\{0\}$ let $b_{n,m}$ denote the coefficient of $\eta^m$ in the Taylor series of the function 
$$\eta\longmapsto\left(\frac{e^{i\frac{2\pi}{L}\<\sum_{j=1}^{\hm}\hbx_j-\sum_{j=1}^{\hm}\hby_j,\be_p\>}-1}{2\pi/L}\right)^nS_{\hX^{\hm},\hY^{\hm},\hXi^{\hm},\hPhi^{\hm}}(C_h,\eta)$$
around $0\in\C$. The following bounds hold. For any $n\in\N\cup\{0\}$,  $m\in\N$, 
\begin{equation*}
\begin{split}
&|b_{n,0}|\le 2^{n}n^n\left(\log F_{t,t',d}\left(\frac{\pi}{2\beta}\right)\right)^{-n}4^{\hm},\\
&|b_{n,m}|\le 2^{n}n^n\left(\log F_{t,t',d}\left(\frac{\pi}{2\beta}\right)\right)^{-n}\hm 16^{\hm}\\
&\qquad\qquad\cdot\frac{1}{m}\left(\beta\left(\frac{F_{t,t',d}\left(\frac{\pi}{2\beta}\right)^{1/(2e\pi d)}+1}{F_{t,t',d}\left(\frac{\pi}{2\beta}\right)^{1/(2e\pi d)}-1}\right)^d\sum_{l=1}^{\tn}l 16^{l}\|U_{L,l}\|_{L,l}\right)^m,
\end{split}
\end{equation*}
where $0^0=1$.
\end{lemma}

\begin{proof}
We apply Proposition \ref{prop_schwinger_functional_bound} to bound the integrand\\ $S_{\hX^{\hm},\hY^{\hm},\hXi^{\hm},\hPhi^{\hm}}(C_h(\sum_{j=1}^nz_j\be_p),\eta)$ in \eqref{eq_schwinger_functional_formula}. Since $|\Im \sum_{j=1}^nz_j|\le$\\
$\frac{1}{2}\log F_{t,t',d}(\pi/(2\beta))$, Lemma \ref{lem_properties_covariance} \eqref{item_determinantbound_covariance}, \eqref{item_determinant_covariance}  imply that the covariance $C_h(\sum_{j=1}^nz_j\be_p)$ satisfies the assumptions of Proposition \ref{prop_schwinger_functional_bound} with $\cB=4$. Moreover, Corollary \ref{cor_covariance_L1_bound}, the periodicity and the translation invariance of $C_h$ ensure that 
$$\cD\le 4\beta\left(\frac{F_{t,t',d}\left(\frac{\pi}{2\beta}\right)^{1/(2e\pi d)}+1}{F_{t,t',d}\left(\frac{\pi}{2\beta}\right)^{1/(2e\pi d)}-1}\right)^d.$$
Therefore, if $\tilde{b}_m$ denotes the coefficient of $\eta^m$ $(\forall m\in\N\cup \{0\})$ in the Taylor series of the function
$\eta\mapsto S_{\hX^{\hm},\hY^{\hm},\hXi^{\hm},\hPhi^{\hm}}(C_h(\sum_{j=1}^nz_j\be_p),\eta)$, we obtain the inequalities that $|\tilde{b}_0|\le 4^{\hm}$,  
\begin{equation}\label{eq_each_order_bound}
|\tilde{b}_m|\le \frac{\hm 16^{\hm}}{m}\left(\beta \left(\frac{F_{t,t',d}\left(\frac{\pi}{2\beta}\right)^{1/(2e\pi d)}+1}{F_{t,t',d}\left(\frac{\pi}{2\beta}\right)^{1/(2e\pi d)}-1}\right)^d\sum_{l=1}^{\tn}l16^{l}\|U_{L,l}\|_{L,l}\right)^m\ (\forall m\in\N).
\end{equation}

In the right hand side of \eqref{eq_schwinger_functional_formula} we expand $S_{\hX^{\hm},\hY^{\hm},\hXi^{\hm},\hPhi^{\hm}}(C_h(\sum_{j=1}^nz_j\be_p),\eta)$ into the power series of $\eta$. Here note that by Fubini-Tonelli's theorem we may exchange the integral $\prod_{j=1}^n\left(\frac{L}{2\pi}\int_{0}^{2\pi/L}d\theta_j\frac{1}{2\pi i}\oint_{|z_j-\theta_j|=r_n}dz_j\frac{1}{(z_j-\theta_j)^2}\right)\sum_{m=0}^{\infty}$ by\\ 
$\sum_{m=0}^{\infty}\prod_{j=1}^n\left(\frac{L}{2\pi}\int_{0}^{2\pi/L}d\theta_j\frac{1}{2\pi i}\oint_{|z_j-\theta_j|=r_n}dz_j\frac{1}{(z_j-\theta_j)^2}\right)$ if $|\eta|$ is sufficiently small, and thus the equality
$$b_{n,m}=\prod_{j=1}^n\left(\frac{L}{2\pi}\int_{0}^{2\pi/L}d\theta_j\frac{1}{2\pi i}\oint_{|z_j-\theta_j|=r_n}dz_j\frac{1}{(z_j-\theta_j)^2}\right)\tilde{b}_m\ $$
holds for all $m\in\N\cup \{0\}$. Combining the inequality \eqref{eq_evaluation_multi_integral} with the bounds \eqref{eq_each_order_bound} yields the desired bound on $|{b}_{n,m}|$.
\end{proof}

\begin{corollary}\label{cor_analytic_continuation}
Assume \eqref{eq_assumption_U_convergence} with $R\in (0,1)$ and use the same notation as in Lemma \ref{lem_schwinger_functional_integral_bound}. For any $n\in \N\cup \{0\}$ and $p\in \{1,\cdots,d\}$ there exists $N_0\in\N$ such that 
\begin{equation}\label{eq_analytic_continuation}
 \left(\frac{e^{i\frac{2\pi}{L}\<\sum_{j=1}^{\hm}\hbx_j-\sum_{j=1}^{\hm}\hby_j,\be_p\>}-1}{2\pi/L}\right)^nS_{\hX^{\hm},\hY^{\hm},\hXi^{\hm},\hPhi^{\hm}}(C_h,1)=\sum_{m=0}^{\infty}{b}_{n,m},
\end{equation}
and 
\begin{equation}\label{eq_analytic_continuation_bound}
\begin{split}
& \left|\frac{e^{i\frac{2\pi}{L}\<\sum_{j=1}^{\hm}\hbx_j-\sum_{j=1}^{\hm}\hby_j,\be_p\>}-1}{2\pi/L}\right|^n\left|S_{\hX^{\hm},\hY^{\hm},\hXi^{\hm},\hPhi^{\hm}}(C_h,1)\right|\\
&\quad\le 2^{n}n^n \left(\log F_{t,t',d}\left(\frac{\pi}{2\beta}\right)\right)^{-n}\left(4^{\hm}-\hm16^{\hm}\log(1-R)\right),
\end{split}
\end{equation}
for all $h\in 2\N/\beta$ with $h\ge 2N_0/\beta$.
\end{corollary}

\begin{proof}
Since $\Tr e^{-\beta(H_0+\eta V)}/\Tr e^{-\beta H_0}>0$ for all $\eta\in\R$, the uniform convergence property \eqref{eq_convergence_partition_function} ensures that there exists $N_0\in \N$ such that 
\begin{equation}\label{eq_regular_inequality}
\left|\int e^{\eta \sum_{l=1}^{\tn}\sum_{(X^l,\Xi^l,\Phi^l)\in\G^{l}\times\spin^{2l}}U_{L,l}(X^l,\Xi^l,\Phi^l)V_{h,X^l,X^l,\Xi^l,\Phi^l}^l(\opsi_X,\psi_X)}d\mu_{C_h}(\opsi_X,\psi_X)\right|>0
\end{equation}
for all $\eta\in\R$ with $|\eta|\le 1$ and $h\in 2\N/\beta$ with $h\ge 2N_0/\beta$. Let us fix such a large $h$. Since the Grassmann Gaussian integral in \eqref{eq_regular_inequality} is a polynomial of $\eta$, we can take a domain $O_h\subset \C$ such that $[-1,1]\subset O_h$ and the inequality \eqref{eq_regular_inequality} holds for all $\eta\in O_h$. This proves that the function
\begin{equation*}
\eta\longmapsto \left(\frac{e^{i\frac{2\pi}{L}\<\sum_{j=1}^{\hm}\hbx_j-\sum_{j=1}^{\hm}\hby_j,\be_p\>}-1}{2\pi/L}\right)^nS_{\hX^{\hm},\hY^{\hm},\hXi^{\hm},\hPhi^{\hm}}(C_h,\eta)
\end{equation*}
is analytic in the domain $O_h$. 

On the other hand, Lemma \ref{lem_schwinger_functional_integral_bound} and the inequality $\|U_{L,l}\|_{L,l}\le \|U_l\|_l$ imply that if $\eta\in\C$ satisfies
\begin{equation}\label{eq_condition_eta}
|\eta|<\left(\beta\left(\frac{F_{t,t',d}\left(\frac{\pi}{2\beta}\right)^{1/(2e\pi d)}+1}{F_{t,t',d}\left(\frac{\pi}{2\beta}\right)^{1/(2e\pi d)}-1}\right)^d\sum_{l=1}^{\tn}l16^{l}\|U_l\|_l\right)^{-1}R,
\end{equation}
\begin{equation}\label{eq_bound_sum_eta}
\begin{split}
\sum_{m=0}^{\infty}|{b}_{n,m}||\eta^m|&\le 2^{n}n^n
\left(\log F_{t,t',d}\left(\frac{\pi}{2\beta}\right)\right)^{-n}4^{\hm}\\
&\quad+2^{n}n^n \left(\log F_{t,t',d}\left(\frac{\pi}{2\beta}\right)\right)^{-n}\hm 16^{\hm}\sum_{m=1}^{\infty}\frac{1}{m}R^m\\
&= 2^{n}n^n \left(\log F_{t,t',d}\left(\frac{\pi}{2\beta}\right)\right)^{-n}\left(4^{\hm}-\hm16^{\hm}\log\left(1-R\right)\right)<\infty.
\end{split}
\end{equation}
Since by the assumption \eqref{eq_assumption_U_convergence} the right hand side of \eqref{eq_condition_eta} is larger than $1$, the identity theorem for analytic functions proves the equality
\begin{equation*}
 \left(\frac{e^{i\frac{2\pi}{L}\<\sum_{j=1}^{\hm}\hbx_j-\sum_{j=1}^{\hm}\hby_j,\be_p\>}-1}{2\pi/L}\right)^nS_{\hX^{\hm},\hY^{\hm},\hXi^{\hm},\hPhi^{\hm}}(C_h,\eta)=\sum_{m=0}^{\infty}{b}_{n,m}\eta^m
\end{equation*}
for all $\eta\in [-1,1]$. The bound \eqref{eq_bound_sum_eta} for $\eta=1$ gives \eqref{eq_analytic_continuation_bound}.
\end{proof}

As the last lemma let us claim
\begin{lemma}\label{lem_thermodynamic_limit}
Assume \eqref{eq_assumption_U_convergence}. The thermodynamic limit \eqref{eq_thermodynamic_limit} exists and takes a finite value.
\end{lemma}
For continuity of our argument in this subsection we present the proof of Lemma \ref{lem_thermodynamic_limit} in Appendix \ref{app_thermodynamic_limit}. We now proceed to
\begin{proof}[Proof of Theorem \ref{thm_exponential_decay}]
By substituting \eqref{eq_stirling_inequality} into \eqref{eq_analytic_continuation_bound} and dividing both sides by $4^ne^nn!\cdot\left(\log F_{t,t',d}\left(\pi/(2\beta)\right)\right)^{-n}$ we have
\begin{equation}\label{eq_identity_bound}
\begin{split}
& \frac{1}{n!}\left(\frac{\log F_{t,t',d}\left(\frac{\pi}{2\beta}\right)}{4e}\left|\frac{e^{i\frac{2\pi}{L}\<\sum_{j=1}^{\hm}\hbx_j-\sum_{j=1}^{\hm}\hby_j,\be_p\>}-1}{2\pi/L}\right|\right)^n\left|S_{\hX^{\hm},\hY^{\hm},\hXi^{\hm},\hPhi^{\hm}}(C_h,1)\right|\\
&\quad\le 2^{-n}\left(4^{\hm}-\hm16^{\hm}\log(1-R)\right)
\end{split}
\end{equation}
for all $n\in \N\cup \{0\}$ and $p\in\{1,\cdots,d\}$. Summing \eqref{eq_identity_bound} over $n\in\N\cup\{0\}$ yields 
\begin{equation}\label{eq_schwinger_pre_bound}
\begin{split}
\left|S_{\hX^{\hm},\hY^{\hm},\hXi^{\hm},\hPhi^{\hm}}(C_h,1)\right|\le& 2\left(4^{\hm}-\hm 16^{\hm}\log(1-R)\right)\\
&\cdot F_{t,t',d}\left(\frac{\pi}{2\beta}\right)^{-\frac{1}{4ed}\sum_{p=1}^d\left|\frac{e^{i2\pi\<\sum_{j=1}^{\hm}\hbx_j-\sum_{j=1}^{\hm}\hby_j,\be_p\>/L}-1}{2\pi/L}\right|}
\end{split}
\end{equation}
for all $h\in 2\N/\beta$ with $h\ge 2N_0/\beta$.
To derive the same bound on $S_{\hY^{\hm},\hX^{\hm},\hPhi^{\hm},\hXi^{\hm}}$\\
$(C_h,1)$ from \eqref{eq_schwinger_pre_bound} is immediate. Then, Proposition \ref{prop_correlation_grassmann_integral} ensures that
\begin{equation*}
\begin{split}
&\left|\<\psi_{\hbx_1\hxi_1}^*\cdots\psi_{\hbx_{\hm}\hxi_{\hm}}^*\psi_{\hby_{\hm}\hphi_{\hm}}\cdots\psi_{\hby_{1}\hphi_{1}}+\psi_{\hby_{1}\hphi_{1}}^*\cdots\psi_{\hby_{\hm}\hphi_{\hm}}^*\psi_{\hbx_{\hm}\hxi_{\hm}}\cdots\psi_{\hbx_1\hxi_1}\>_L\right|\\
&\quad\le (4^{\hm+1}-\hm 4^{2\hm+1}\log(1-R))\\
&\qquad\cdot  F_{t,t',d}\left(\frac{\pi}{2\beta}\right)^{-\frac{1}{4ed}\sum_{p=1}^d\left|\frac{e^{i2\pi\<\sum_{j=1}^{\hm}\hbx_j-\sum_{j=1}^{\hm}\hby_j,\be_p\>/L}-1}{2\pi/L}\right|}.
\end{split}
\end{equation*}
Finally, by Lemma \ref{lem_thermodynamic_limit} we can take the limit $L\to +\infty$ in the both sides to obtain the claimed inequality.
\end{proof}

\begin{proof}[Proof of Theorem \ref{thm_exponential_decay_hubbard_model}]
Since the construction of Theorem \ref{thm_exponential_decay_hubbard_model} is parallel to that of Theorem \ref{thm_exponential_decay}, we only outline the proof. The main difference is that we use Proposition \ref{prop_schwinger_functional_hubbard_model_bound} in place of Proposition \ref{prop_schwinger_functional_bound}.

(Step 1) For any $m,n\in\N\cup\{0\}$ let $c_{n,m}$ denote the coefficient of $\eta^m$ in the Taylor expansion of 
\begin{equation*}
\eta\longmapsto \left(\frac{e^{i\frac{2\pi}{L}\<\hbx_1+\hbx_2-\hby_1-\hby_2,\be_p\>}-1}{2\pi/L}\right)^nS_{(\hbx_1,\hbx_2), (\hby_1,\hby_2),(\ua,\da), (\ua,\da)}(C_h,\eta)
\end{equation*}
around $0\in\C$. By repeating the same argument as in Lemma \ref{lem_schwinger_functional_integral_bound} using Proposition \ref{prop_schwinger_functional_hubbard_model_bound} in stead of Proposition \ref{prop_schwinger_functional_bound} we obtain
\begin{equation*}
\begin{split}
|c_{n,m}|\le &2^n n^n \left(\log F_{t,t',d}\left(\frac{\pi}{2\beta}\right)\right)^{-n}\\
&\cdot\frac{64}{3m+4}\left(\begin{array}{c} 3m+4 \\ m\end{array}\right)\left(16\beta \left(
\frac{F_{t,t',d}\left(\frac{\pi}{2\beta}\right)^{1/(2e\pi d)}+1}{F_{t,t',d}\left(\frac{\pi}{2\beta}\right)^{1/(2e\pi d)}-1}\right)^d|U|\right)^m.
\end{split}
\end{equation*}

(Step 2) Recall the fact that the radius of convergence of the power series 
$$\sum_{m=0}^{\infty}\frac{4}{3m+4}\left(\begin{array}{c} 3m+4 \\ m\end{array}\right)x^m$$
is $4/27$ and
\begin{equation*}
\sum_{m=0}^{\infty}\frac{4}{3m+4}\left(\begin{array}{c} 3m+4 \\ m\end{array}\right)\left(\frac{4}{27}\right)^m=\frac{81}{16}
\end{equation*}
(see \cite[\mbox{Lemma 4.9}]{K}).

Therefore, the argument involving the identity theorem parallel to the proof of Corollary \ref{cor_analytic_continuation} shows that for all $U\in\R$ satisfying 
\begin{equation*}
|U|<(108\beta)^{-1}\left(
\frac{F_{t,t',d}\left(\frac{\pi}{2\beta}\right)^{1/(2e\pi d)}+1}{F_{t,t',d}\left(\frac{\pi}{2\beta}\right)^{1/(2e\pi d)}-1}\right)^{-d}
\end{equation*}
and for sufficiently large $h\in 2\N/\beta$ the equality
\begin{equation*}
 \left(\frac{e^{i\frac{2\pi}{L}\<\hbx_1+\hbx_2-\hby_1-\hby_2,\be_p\>}-1}{2\pi/L}\right)^nS_{(\hbx_1,\hbx_2),(\hby_1,\hby_2),(\ua,\da),(\ua,\da)}(C_h,1)=\sum_{m=0}^{\infty}{c}_{n,m}
\end{equation*}
and the inequality 
\begin{equation}\label{eq_pre_bound_Hubbard}
\begin{split}
& \left|\frac{e^{i\frac{2\pi}{L}\<\hbx_1+\hbx_2-\hby_1-\hby_2,\be_p\>}-1}{2\pi/L}\right|^n\left|S_{(\hbx_1,\hbx_2),(\hby_1,\hby_2),(\ua,\da),(\ua,\da)}\left(C_h,1\right)\right|\\
&\quad\le 2^{n}n^n \left(\log F_{t,t',d}\left(\frac{\pi}{2\beta}\right)\right)^{-n}\cdot 81
\end{split}
\end{equation}
hold. By using the fact that $S_{(\hbx_1,\hbx_2),(\hby_1,\hby_2),(\ua,\da),(\ua,\da)}\left(C_h,1\right)$ is a rational function of $U$, we can prove \eqref{eq_pre_bound_Hubbard} for $U\in\R$ with \eqref{eq_assumption_U_convergence_hubbard_model} as well.

(Step 3) By repeating the same calculation as in the proof of Theorem \ref{thm_exponential_decay} or Proposition \ref{prop_covariance_exponential_decay} we reach the bound
\begin{equation}\label{eq_decay_bound_almost}
\begin{split}
&\left|\<\psi_{\hbx_1\ua}^*\psi_{\hbx_{2}\da}^*\psi_{\hby_{2}\da}\psi_{\hby_{1}\ua}+\psi_{\hby_{1}\ua}^*\psi_{\hby_{2}\da}^*\psi_{\hbx_{2}\da}\psi_{\hbx_1\ua}\>_L\right|\\
&\quad\le 324\cdot F_{t,t',d}\left(\frac{\pi}{2\beta}\right)^{-\frac{1}{4ed}\sum_{p=1}^d\left|\frac{e^{i2\pi\<\hbx_1+\hbx_2-\hby_1-\hby_2,\be_p\>/L}-1}{2\pi/L}\right|}
\end{split}
\end{equation}
for $U\in\R$ satisfying \eqref{eq_assumption_U_convergence_hubbard_model} from \eqref{eq_pre_bound_Hubbard}.

(Step 4) By writing
\begin{equation*}
\begin{split}
&U\sum_{\bx\in\G}\psi_{\bx\ua}^*\psi_{\bx\da}^*\psi_{\bx\da}\psi_{\bx\ua}\\
&\quad =\sum_{\bx_1,\bx_2\in\G}\sum_{\xi_1,\xi_2,\phi_1,\phi_2\in\spin}U_o((\bx_1,\bx_2),(\xi_1,\xi_2),(\phi_1,\phi_2))\psi_{\bx_1\xi_1}^*\psi_{\bx_2\xi_2}^*\psi_{\bx_2\phi_2}\psi_{\bx_1\phi_1}
\end{split}
\end{equation*}
with $U_o((\bx_1,\bx_2),(\xi_1,\xi_2),(\phi_1,\phi_2)):=U\delta_{\bx_1,\bx_2}\delta_{\xi_1,\ua}\delta_{\xi_2,\da}\delta_{\xi_1,\phi_1}\delta_{\xi_2,\phi_2}$,
one can translate the proof of Lemma \ref{lem_thermodynamic_limit} to confirm the existence of the thermodynamic limit \eqref{eq_thermodynamic_limit_hubbard}.  Then, by sending $L\to +\infty$ in \eqref{eq_decay_bound_almost} we obtain the claimed bound for $U$ satisfying \eqref{eq_assumption_U_convergence_hubbard_model}.
\end{proof}

\appendix

\section{Anti-symmetrization of the coefficients}\label{app_anti_symmetricity}
For simplicity we use the notations $X^l,Y^l\in\G^l$, $\Xi^l,\Phi^l\in\spin^l$ $(l\in\N)$ defined in the subsection \ref{subsec_correlation_functions}. Moreover, let $X^l_{\Xi}$, $Y^l_{\Phi}$ denote 
$$((\bx_1,\xi_1),(\bx_2,\xi_2),\cdots, (\bx_l,\xi_l)),\ ((\by_1,\phi_1),(\by_2,\phi_2),\cdots, (\by_l,\phi_l))\in (\G\times\spin)^l, 
$$
respectively. Let $S_l$ be the set of permutations over $l$ elements. For $\pi,\tau\in S_l$ we write
\begin{equation*}
\begin{split}
&\pi(X^l_{\Xi}):=((\bx_{\pi(1)},\xi_{\pi(1)}),\cdots,(\bx_{\pi(l)},\xi_{\pi(l)})),\\
&\tau(Y^l_{\Phi}):=((\by_{\tau(1)},\phi_{\tau(1)}),\cdots,(\by_{\tau(l)},\phi_{\tau(l)})).
\end{split}
\end{equation*}

\begin{lemma}\label{lem_anti_symmetrization}
Assume that $l\in\{1,\cdots,2L^d\}$. For any function $g:\G^l\times\spin^{2l}\to\C$ there uniquely exists a function $f:(\G\times\spin)^{2l}\to \C$ such that
\begin{equation}\label{eq_anti_symmetricity}
f(\pi(X^l_{\Xi}),\tau(Y^l_{\Phi}))=\sgn(\pi)\sgn(\tau)f(X^l_{\Xi},Y^l_{\Phi})
\end{equation}
for all $\pi,\tau\in S_l$, $X^l_{\Xi}, Y^l_{\Phi}\in (\G\times\spin)^l$ and
\begin{equation}\label{eq_anti_symmetric_modification}
\begin{split}
&\sum_{X^l\in\G^l}\sum_{\Xi^l,\Phi^l\in\spin^l}g(X^l,\Xi^l,\Phi^l)\psi^*_{\bx_1\xi_1}\cdots\psi^*_{\bx_l\xi_l}\psi_{\bx_l\phi_l}\cdots\psi_{\bx_1\phi_1}\\
&\quad = \sum_{X_{\Xi}^l,Y_{\Phi}^l\in(\G\times\spin)^l}f(X_{\Xi}^l,Y_{\Phi}^l)\psi^*_{\bx_1\xi_1}\cdots\psi^*_{\bx_l\xi_l}\psi_{\by_1\phi_1}\cdots\psi_{\by_l\phi_l}.
\end{split}
\end{equation}
Moreover, the following inequality holds.
\begin{equation}\label{eq_anti_symmetricity_inequality}
\begin{split}
&\max\Bigg\{\max_{(\bx_1,\xi_1)\in\G\times\spin}\sum_{(\bx_k,\xi_k)\in\G\times\spin\atop \forall k\in\{2,\cdots,l\}}\sum_{Y_{\Phi}^l\in(\G\times\spin)^l}|f(X^l_{\Xi}, Y^l_{\Phi})|,\\
&\quad\qquad \max_{(\by_1,\phi_1)\in\G\times\spin}\sum_{X_{\Xi}^l\in(\G\times\spin)^l}\sum_{(\by_k,\phi_k)\in\G\times\spin\atop \forall k\in\{2,\cdots,l\}}|f(X^l_{\Xi}, Y^l_{\Phi})|\Bigg\}\\
&\quad \le \max\Bigg\{\max_{j\in\{1,\cdots,l\}}\max_{(\bx_j,\xi_j)\in\G\times\spin}\sum_{(\bx_k,\xi_k)\in\G\times\spin\atop \forall k\in\{1,\cdots,l\}\backslash\{j\}}\sum_{\Phi^l\in\spin^l}|g(X^l,\Xi^l, \Phi^l)|,\\
&\quad\qquad\qquad \max_{j\in\{1,\cdots,l\}}\max_{(\bx_j,\phi_j)\in\G\times\spin}\sum_{(\bx_k,\phi_k)\in\G\times\spin\atop \forall k\in\{1,\cdots,l\}\backslash\{j\}}\sum_{\Xi^l\in\spin^l}|g(X^l,\Xi^l, \Phi^l)|\Bigg\}.
\end{split}
\end{equation}
\end{lemma}
\begin{proof}
By setting $u(X^l_{\Xi}, Y^l_{\Phi}):=g(X^l,\Xi^l,\Phi^l)\delta_{X^l,Y^l}(-1)^{l(l-1)/2}$, we see that
\begin{equation*}
\begin{split}
&\sum_{X^l\in\G^l}\sum_{\Xi^l,\Phi^l\in\spin^l}g(X^l,\Xi^l,\Phi^l)\psi^*_{\bx_1\xi_1}\cdots\psi^*_{\bx_l\xi_l}\psi_{\bx_l\phi_l}\cdots\psi_{\bx_1\phi_1}\\
&\quad= \sum_{X_{\Xi}^l,Y_{\Phi}^l\in(\G\times\spin)^l}u(X_{\Xi}^l,Y_{\Phi}^l)\psi^*_{\bx_1\xi_1}\cdots\psi^*_{\bx_l\xi_l}\psi_{\by_1\phi_1}\cdots\psi_{\by_l\phi_l}\\
&\quad= \sum_{X_{\Xi}^l,Y_{\Phi}^l\in(\G\times\spin)^l}\sum_{\pi,\tau\in S_l}\frac{\sgn(\pi)\sgn(\tau)}{(l!)^2}u(\pi(X_{\Xi}^l),\tau(Y_{\Phi}^l))\\
&\quad\quad\cdot\psi^*_{\bx_1\xi_1}\cdots\psi^*_{\bx_l\xi_l}\psi_{\by_1\phi_1}\cdots\psi_{\by_l\phi_l}.
\end{split}
\end{equation*}
If we define $f:(\G\times\spin)^{2l}\to \C$ by 
$$f(X_{\Xi}^l,Y_{\Phi}^l):=\sum_{\pi,\tau\in S_l}\frac{\sgn(\pi)\sgn(\tau)}{(l!)^2}u(\pi(X_{\Xi}^l),\tau(Y_{\Phi}^l)),$$
$f$ satisfies \eqref{eq_anti_symmetricity}-\eqref{eq_anti_symmetric_modification} and 
\begin{equation}\label{eq_one_inequality}
\begin{split}
&\max_{(\bx_1,\xi_1)\in\G\times\spin}\sum_{(\bx_k,\xi_k)\in\G\times\spin\atop \forall k\in\{2,\cdots,l\}}\sum_{Y_{\Phi}^l\in(\G\times\spin)^l}|f(X^l_{\Xi}, Y^l_{\Phi})|\\
&\quad \le \max_{j\in\{1,\cdots,l\}}\max_{(\bx_j,\xi_j)\in\G\times\spin}\sum_{(\bx_k,\xi_k)\in\G\times\spin\atop \forall k\in\{1,\cdots,l\}\backslash\{j\}}\sum_{Y_{\Phi}^l\in(\G\times\spin)^l}|u(X^l_{\Xi}, Y^l_{\Phi})|\\
&\quad= \max_{j\in\{1,\cdots,l\}}\max_{(\bx_j,\xi_j)\in\G\times\spin}\sum_{(\bx_k,\xi_k)\in\G\times\spin\atop \forall k\in\{1,\cdots,l\}\backslash\{j\}}\sum_{\Phi^l\in\spin^l}|g(X^l,\Xi^l, \Phi^l)|.
\end{split}
\end{equation}
The inequality
\begin{equation}\label{eq_the_other_inequality}
\begin{split}
&\max_{(\by_1,\phi_1)\in\G\times\spin}\sum_{X_{\Xi}^l\in(\G\times\spin)^l}\sum_{(\by_k,\phi_k)\in\G\times\spin\atop \forall k\in\{2,\cdots,l\}}|f(X^l_{\Xi}, Y^l_{\Phi})|\\
&\quad \le \max_{j\in\{1,\cdots,l\}}\max_{(\bx_j,\phi_j)\in\G\times\spin}\sum_{(\bx_k,\phi_k)\in\G\times\spin\atop \forall k\in\{1,\cdots,l\}\backslash\{j\}}\sum_{\Xi^l\in\spin^l}|g(X^l,\Xi^l, \Phi^l)|
\end{split}
\end{equation}
can be confirmed similarly. The inequalities \eqref{eq_one_inequality}-\eqref{eq_the_other_inequality} yield \eqref{eq_anti_symmetricity_inequality}.

To prove the uniqueness of such $f$, we number each element of $\G\times\spin$ so that $\G\times\spin=\{(\bx\xi)_j\ |\ j=1,\cdots,2L^d\}$. Suppose that
$$\sum_{X_{\Xi}^l,Y_{\Phi}^l\in(\G\times\spin)^l}f(X_{\Xi}^l,Y_{\Phi}^l)\psi^*_{\bx_1\xi_1}\cdots\psi^*_{\bx_l\xi_l}\psi_{\by_1\phi_1}\cdots\psi_{\by_l\phi_l}=0.$$
The anti-symmetricity \eqref{eq_anti_symmetricity} implies that
\begin{equation}\label{eq_ordered_vanish}
\begin{split}
\sum_{1\le j_1<\cdots<j_l\le 2L^d\atop 1\le k_1<\cdots<k_l\le 2L^d}&f((\bx\xi)_{j_1},\cdots,(\bx\xi)_{j_l},(\bx\xi)_{k_1},\cdots,(\bx\xi)_{k_l}) \\
&\cdot\psi^*_{(\bx\xi)_{j_1}}\cdots\psi^*_{(\bx\xi)_{j_l}}\psi_{(\bx\xi)_{k_1}}\cdots\psi_{(\bx\xi)_{k_l}}=0.
\end{split}
\end{equation}
Since 
$$\{\psi_{(\bx\xi)_{j_1}}^*\cdots\psi_{(\bx\xi)_{j_l}}^*\psi_{(\bx\xi)_{k_1}}\cdots\psi_{(\bx\xi)_{k_l}}\}_{1\le j_1<\cdots<j_l\le 2L^d,\ 1\le k_1<\cdots<k_l\le 2L^d}$$
are linearly independent in the complex vector space of linear operators on the Fermionic Fock space on $\G\times\spin$, the equality \eqref{eq_ordered_vanish} deduces that 
\begin{equation}\label{eq_coefficients_vanish}
f((\bx\xi)_{j_1},\cdots,(\bx\xi)_{j_l},(\bx\xi)_{k_1},\cdots,(\bx\xi)_{k_l})=0\end{equation}
for all $j_m,k_m\in\{1,\cdots,2L^d\}$ $(m\in\{1,\cdots,l\})$ with $j_1<\cdots<j_l$, $k_1<\cdots<k_l$.  Again by \eqref{eq_anti_symmetricity} the equality \eqref{eq_coefficients_vanish} ensures that $f(X_{\Xi}^l,Y_{\Phi}^l)=0$ for all $X_{\Xi}^l,Y_{\Phi}^l\in(\G\times\spin)^l$, which concludes the uniqueness of $f$ satisfying \eqref{eq_anti_symmetricity} and \eqref{eq_anti_symmetric_modification}.
\end{proof}

\section{Proof of Lemma \ref{lem_thermodynamic_limit}}\label{app_thermodynamic_limit}
By Corollary \ref{cor_analytic_continuation} with the same notation there exists $N_0\in\N$ such that
\begin{equation}\label{eq_perturbation_expansion_max}
S_{\hX^{\hm},\hY^{\hm},\hXi^{\hm},\hPhi^{\hm}}(C_h,1)=\sum_{m=0}^{\infty}{b}_{0,m}
\end{equation}
for all $h\in 2\N/\beta$ with $h\ge 2N_0/\beta$. By Lemma \ref{lem_schwinger_functional_integral_bound} $|b_{0,0}|\le 4^{\hm}$ and
\begin{equation}\label{eq_estimate_every_term}
|b_{0,m}|\le \frac{\hm 16^{\hm}}{m}R^m\ (\forall m\in\N).
\end{equation}

Assume that the limits
\begin{equation}\label{eq_each_term_limits}
\lim_{h\to +\infty\atop h\in 2\N/\beta}b_{0,m},\ \lim_{L\to +\infty\atop L\in\N}\lim_{h\to +\infty\atop h\in 2\N/\beta}b_{0,m}
\end{equation}
exist for any $m\in\N\cup\{0\}$. Since the right hand side of the inequality \eqref{eq_estimate_every_term} is summable over $m\in \N$, we can apply the dominated convergence theorem for $l^1(\N\cup\{0\})$ to verify that the equalities
\begin{equation}\label{eq_limit_exchange}
\begin{split}
&\lim_{h\to +\infty\atop h\in 2\N/\beta}\sum_{m=0}^{\infty}b_{0,m}=\sum_{m=0}^{\infty}\lim_{h\to +\infty\atop h\in 2\N/\beta}b_{0,m},\\
&\lim_{L\to +\infty\atop L\in\N}\lim_{h\to +\infty\atop h\in 2\N/\beta}\sum_{m=0}^{\infty}b_{0,m}=\sum_{m=0}^{\infty}\lim_{L\to +\infty\atop L\in\N}\lim_{h\to +\infty\atop h\in 2\N/\beta}b_{0,m}
\end{split}
\end{equation}
hold in $\C$.
Combining \eqref{eq_limit_exchange} with \eqref{eq_perturbation_expansion_max} proves the existence of 
$$\lim_{L\to +\infty\atop L\in\N}\lim_{h\to +\infty\atop h\in 2\N/\beta}S_{\hX^{\hm},\hY^{\hm},\hXi^{\hm},\hPhi^{\hm}}(C_h,1)$$
for all $(\hX^{\hm},\hY^{\hm},\hXi^{\hm},\hPhi^{\hm})\in(\Z^d)^{2\hm}\times \spin^{2\hm}$. Proposition \ref{prop_correlation_grassmann_integral} then concludes that the correlation function converges to a finite value as $L\to +\infty$.
Hence it suffices to prove the existence of the limits \eqref{eq_each_term_limits}.

By \eqref{eq_bound_for_m_1} for $G_h=C_h$,  $b_{0,0}=\det(C(\hbx_j\hxi_j0,\hby_k\hphi_k 0))_{1\le j,k\le \hm}$. From \eqref{eq_covariance} it is clear that $\lim_{L\to +\infty,L\in\N}b_{0,0}$ exists. 

Let us derive an expression of $b_{0,m}$ $(m\ge 1)$ from the tree formula \eqref{eq_tree_expansion}. We consider that the root of any tree $T\in \T(\{0,1,\cdots,m\})$ is the vertex $0$. For any $j\in \{0,\cdots,m\}$ let $\dis_T(j)$ $(\in\{0,\cdots,m\})$ denote the distance between the root $0$ and the vertex $j$ along the unique path of $T$ connecting $0$ with $j$. We write any line of $T$ as $(p,q)$ $(p,q\in\{0,\cdots,m\})$ with $\dis_T(q)=\dis_T(p)+1$. Set $\max(T):=\max_{j\in\{1,\cdots,m\}}\dis_T(j)$ and $D_j(T):=\{(p,q)\in T\ |\ \dis_T(p)+1=\dis_T(q)=j\}$ for $j\in\{1,\cdots,\max(T)\}$.

By letting $\prod_{(p,q)\in T}(\Delta_{p,q}+\Delta_{q,p})$ act on the monomial \eqref{eq_grassmann_monomial} first and then applying $e^{\Delta(M(T,\pi,\bs))}$ to the rest of the monomials in \eqref{eq_tree_expansion}, we have 
\begin{equation}\label{eq_another_representation}
\begin{split}
b_{0,m}&=\\
&\frac{1}{\beta h}\sum_{s_0\in [0,\beta)_h}\prod_{k=1}^{m}\left(\frac{1}{h}\sum_{s_k\in[0,\beta)_h}\sum_{l_k=1}^{\tn}\sum_{(X^{l_k}_k,\Xi^{l_k}_k,\Phi^{l_k}_k)\in\G^{l_k}\times\spin^{2l_k}}U_{L,l_k}(X_k^{l_k},\Xi_k^{l_k},\Phi_k^{l_k})\right)\\
&\qquad\cdot\sum_{T\in\T(\{0,1,\cdots,m\})}\prod_{(0,q)\in D_1(T)}\left(\sum_{j_{0,q}=1}^{\hm}\sum_{k_{0,q}=1}^{l_q}\sum_{a_{0,q}\in\{-1,1\}}C_h^{j_{0,q},k_{0,q},a_{0,q}}(\bx_{q,k_{0,q}})\right)\\
&\qquad\cdot\prod_{r=2}^{\max(T)}\prod_{(p,q)\in D_r(T)}\left(\sum_{j_{p,q}=1}^{l_p}\sum_{k_{p,q}=1}^{l_q}\sum_{a_{p,q}\in\{-1,1\}}C_h^{j_{p,q},k_{p,q},a_{p,q}}(\bx_{p,j_{p,q}},\bx_{q,k_{p,q}})\right)\\
&\qquad\cdot f(T,\{j_{p,q},k_{p,q},a_{p,q}\}_{(p,q)\in T},C_h),
\end{split}
\end{equation}
where for $(0,q)\in D_1(T)$
\begin{equation*}
\begin{split}
&C_h^{j_{0,q},k_{0,q},-1}(\bx):=C_h(\hbx_{j_{0,q}}\hxi_{j_{0,q}}s_0, \bx\phi_{q,k_{0,q}}s_q),\\
&C_h^{j_{0,q},k_{0,q},1}(\bx):=C_h(\bx\xi_{q,k_{0,q}}s_q, \hby_{j_{0,q}}\hphi_{j_{0,q}}s_0),
\end{split}
\end{equation*}
for $(p,q)\in D_r(T)$ $(r\ge 2)$
\begin{equation*}
\begin{split}
&C_h^{j_{p,q},k_{p,q},-1}(\bx,\by):=C_h(\bx \xi_{p,j_{p,q}}s_p,\by \phi_{q,k_{p,q}}s_q),\\
&C_h^{j_{p,q},k_{p,q},1}(\bx,\by):=C_h(\by \xi_{q,k_{p,q}}s_q,\bx \phi_{p,j_{p,q}}s_p),
\end{split}
\end{equation*}
and the term $f(T,\{j_{p,q},k_{p,q},a_{p,q}\}_{(p,q)\in T},C_h)$ satisfies that if there exist $(v,q),(v,q')$ $\in T$ such that $q\neq q'$ and $(j_{v,q}, a_{v,q})=(j_{v,q'}, a_{v,q'})$ or if there exist $(p,v),(v,q)\in T$ such that $(k_{p,v},a_{p,v})=(j_{v,q},-a_{v,q})$, then $f(T,\{j_{p,q},k_{p,q},a_{p,q}\}_{(p,q)\in T},C_h)=0$. All the $C_h$-dependent terms contained in $f(T,\{j_{p,q},k_{p,q},a_{p,q}\}_{(p,q)\in T},C_h)$ are the remains left after $e^{\Delta(M(T,\pi,\bs))}$ acting on Grassmann monomials.

Moreover, by setting 
\begin{equation*}
\begin{split}
&f'(T,\{j_{p,q},k_{p,q},a_{p,q}\}_{(p,q)\in T},C_h,\{U_{L,l_k}\}_{k=1}^m)\\
&\qquad:=f(T,\{j_{p,q},k_{p,q},a_{p,q}\}_{(p,q)\in T},C_h)\prod_{k=1}^{m}U_{L,l_k}(X_k^{l_k},\Xi_k^{l_k},\Phi_k^{l_k})^{1_{l_k= 1}},
\end{split}
\end{equation*}
the expression \eqref{eq_another_representation} is rewritten as follows.
\begin{equation}\label{eq_further_transformation}
\begin{split}
&b_{0,m}=\\
&\frac{1}{\beta h}\sum_{s_0\in [0,\beta)_h}\prod_{k=1}^{m}\left(\frac{1}{h}\sum_{s_k\in[0,\beta)_h}\sum_{l_k=1}^{\tn}\sum_{(X^{l_k}_k,\Xi^{l_k}_k,\Phi^{l_k}_k)\in\G^{l_k}\times\spin^{2l_k}}\right)\sum_{T\in\T(\{0,1,\cdots,m\})}\\
&\cdot\prod_{(0,q)\in D_1(T)}\left(U_{L,l_q}(X_q^{l_q},\Xi_q^{l_q},\Phi_q^{l_q})^{1_{l_q\ge 2}}\sum_{j_{0,q}=1}^{\hm}\sum_{k_{0,q}=1}^{l_q}\sum_{a_{0,q}\in\{-1,1\}}C_h^{j_{0,q},k_{0,q},a_{0,q}}(\bx_{q,k_{0,q}})\right)\\
&\cdot\prod_{r=2}^{\max(T)}\prod_{(p,q)\in D_r(T)}\Bigg(U_{L,l_q}(X_q^{l_q},\Xi_q^{l_q},\Phi_q^{l_q})^{1_{l_q\ge 2}}\sum_{j_{p,q}=1}^{l_p}\sum_{k_{p,q}=1}^{l_q}\sum_{a_{p,q}\in\{-1,1\}}\\
&\cdot C_h^{j_{p,q},k_{p,q},a_{p,q}}(\bx_{p,j_{p,q}},\bx_{q,k_{p,q}})\Bigg) f'(T,\{j_{p,q},k_{p,q},a_{p,q}\}_{(p,q)\in T},C_h,\{U_{L,l_k}\}_{k=1}^m).
\end{split}
\end{equation}

 Note that by Proposition \ref{prop_extended_determinant_bound} the bound of the form \eqref{eq_determinant_bound_application} for $\cB=4$ is valid. By using this bound, the equality \eqref{eq_phi_equality} and the inequality that
$$|U_{L,l_k}(X_k^{l_k},\Xi_k^{l_k},\Phi_k^{l_k})^{1_{l_k= 1}}|\le \|U_1\|_1^{1_{l_k= 1}},
$$
we can find a non-negative function $g(T,\{j_{p,q},k_{p,q},a_{p,q}\}_{(p,q)\in T}, \{l_k\}_{k=1}^m)$, which is independent of $L$, $h$ and the variables $\{(X^{l_k}_k,\Xi^{l_k}_k,\Phi^{l_k}_k)\}_{k=1}^m$, such that 
\begin{equation}\label{eq_dominant_function}
\begin{split}
&|f'(T,\{j_{p,q},k_{p,q},a_{p,q}\}_{(p,q)\in T},C_h,\{U_{L,l_k}\}_{k=1}^m)|\\
&\qquad \le g(T,\{j_{p,q},k_{p,q},a_{p,q}\}_{(p,q)\in T},\{l_k\}_{k=1}^m).
\end{split}
\end{equation}

By using the periodicity and the translation invariance of  $C_h$ and $U_{L,l}$ and by arguing recursively with respect to $r$ running from $\max(T)$ to $1$, we see that
\begin{equation}\label{eq_recurcive_transformation}
\begin{split}
&\prod_{k=1}^{m}\left(\sum_{X^{l_k}_k\in\G^{l_k}}\right)\\
&\cdot\prod_{(0,q)\in D_1(T)}\left(U_{L,l_q}(X_q^{l_q},\Xi_q^{l_q},\Phi_q^{l_q})^{1_{l_q\ge 2}}\sum_{j_{0,q}=1}^{\hm}\sum_{k_{0,q}=1}^{l_q}\sum_{a_{0,q}\in\{-1,1\}}C_h^{j_{0,q},k_{0,q},a_{0,q}}(\bx_{q,k_{0,q}})\right)\\
&\cdot\prod_{r=2}^{\max(T)}\prod_{(p,q)\in D_r(T)}\Bigg(U_{L,l_q}(X_q^{l_q},\Xi_q^{l_q},\Phi_q^{l_q})^{1_{l_q\ge 2}}\sum_{j_{p,q}=1}^{l_p}\sum_{k_{p,q}=1}^{l_q}\sum_{a_{p,q}\in\{-1,1\}}\\
&\qquad\cdot C_h^{j_{p,q},k_{p,q},a_{p,q}}(\bx_{p,j_{p,q}},\bx_{q,k_{p,q}})\Bigg)f'(T,\{j_{p,q},k_{p,q},a_{p,q}\}_{(p,q)\in T},C_h,\{U_{L,l_k}\}_{k=1}^m)\\
&=\prod_{k=1}^{m}\left(\sum_{X^{l_k}_k\in\G^{l_k}}\right)\prod_{(0,q)\in D_1(T)}\Bigg(\sum_{j_{0,q}=1}^{\hm}\sum_{k_{0,q}=1}^{l_q}\sum_{a_{0,q}\in\{-1,1\}} \tilde{C}_h^{j_{0,q},k_{0,q},a_{0,q}}(\bx_{q,k_{0,q}})\\
&\qquad\cdot U_{L,l_q}((\bx_{q,1},\cdots,\bx_{q,k_{0,q}-1},\bx_{q,k_{0,q}+1},\cdots,\bx_{q,l_q},\b0),\Xi_q^{l_q},\Phi_q^{l_q})^{1_{l_q\ge 2}}\Bigg)\\
&\quad\cdot\prod_{r=2}^{\max(T)}\prod_{(p,q)\in D_r(T)}\Bigg(\sum_{j_{p,q}=1}^{l_p}\sum_{k_{p,q}=1}^{l_q}\sum_{a_{p,q}\in\{-1,1\}}C_h^{j_{p,q},k_{p,q},a_{p,q}}(\b0,\bx_{q,k_{p,q}})\\
&\qquad\cdot U_{L,l_q}((\bx_{q,1},\cdots,\bx_{q,k_{p,q}-1},\bx_{q,k_{p,q}+1},\cdots,\bx_{q,l_q},\b0),\Xi_q^{l_q},\Phi_q^{l_q})^{1_{l_q\ge 2}}\Bigg)\\
&\quad\cdot f'(T,\{j_{p,q},k_{p,q},a_{p,q}\}_{(p,q)\in T},C_h,\{U_{L,l_k}\}_{k=1}^m),
\end{split}
\end{equation}
where
\begin{equation*}
\begin{split}
&\tilde{C}_h^{j_{0,q},k_{0,q},-1}(\bx):=C_h(\b0\hxi_{j_{0,q}}s_0, \bx\phi_{q,k_{0,q}}s_q),\\
&\tilde{C}_h^{j_{0,q},k_{0,q},1}(\bx):=C_h(\bx\xi_{q,k_{0,q}}s_q, \b0\hphi_{j_{0,q}}s_0).
\end{split}
\end{equation*}
Remark that the dependency of $f'(T,\{j_{p,q},k_{p,q},a_{p,q}\}_{(p,q)\in T},C_h,\{U_{L,l_k}\}_{k=1}^m)$ on the variables $(X_1^{l_1},\cdots,X_m^{l_m})$ is changed in every step of the derivation of \eqref{eq_recurcive_transformation}. However we used the same notation for simplicity. 

By substituting \eqref{eq_recurcive_transformation} into \eqref{eq_further_transformation} we obtain
\begin{equation}\label{eq_each_term_decomposition}
\begin{split}
b_{0,m}=&\prod_{k=1}^{m}\left(\sum_{l_k=1}^{\tn}\sum_{(X^{l_k}_k,\Xi^{l_k}_k,\Phi^{l_k}_k)\in(\Z^d)^{l_k}\times\spin^{2l_k}}\right)\\
&\cdot B_{m,\{l_k\}_{k=1}^m}^{L,h}((X_1^{l_1},\Xi_1^{l_1},\Phi_1^{l_1}),\cdots,(X_m^{l_m},\Xi_m^{l_m},\Phi_m^{l_m})),
\end{split}
\end{equation}
where
\begin{equation*}
\begin{split}
&B_{m,\{l_k\}_{k=1}^m}^{L,h}((X_1^{l_1},\Xi_1^{l_1},\Phi_1^{l_1}),\cdots,(X_m^{l_m},\Xi_m^{l_m},\Phi_m^{l_m}))\\
&\quad:=\frac{1}{\beta}\prod_{k=0}^m\left(\frac{1}{h}\sum_{s_k\in[0,\beta)_h}\right)\prod_{k=1}^m1_{X_k^{l_k}\in((\{-\lfloor L/2\rfloor,\cdots,-\lfloor L/2\rfloor+L-1\})^d)^{l_k}}\sum_{T\in\T(\{0,1,\cdots,m\})}\\
&\quad\qquad\cdot \prod_{(0,q)\in D_1(T)}\Bigg(\sum_{j_{0,q}=1}^{\hm}\sum_{k_{0,q}=1}^{l_q}\sum_{a_{0,q}\in\{-1,1\}}\tilde{C}_h^{j_{0,q},k_{0,q},a_{0,q}}(\bx_{q,k_{0,q}})\\
&\qquad\qquad\qquad\cdot U_{l_q}((\bx_{q,1},\cdots,\bx_{q,k_{0,q}-1},\bx_{q,k_{0,q}+1},\cdots,\bx_{q,l_q},\b0),\Xi_q^{l_q},\Phi_q^{l_q})^{1_{l_q}\ge 2}\Bigg)\\
&\qquad\quad\cdot\prod_{(p,q)\in T\backslash D_1(T)}\Bigg(\sum_{j_{p,q}=1}^{l_p}\sum_{k_{p,q}=1}^{l_q}\sum_{a_{p,q}\in\{-1,1\}} C_h^{j_{p,q},k_{p,q},a_{p,q}}(\b0,\bx_{q,k_{p,q}})\\
&\qquad\qquad\qquad\cdot U_{l_q}((\bx_{q,1},\cdots,\bx_{q,k_{p,q}-1},\bx_{q,k_{p,q}+1},\cdots,\bx_{q,l_q},\b0),\Xi_q^{l_q},\Phi_q^{l_q})^{1_{l_q}\ge 2}\Bigg)\\
&\qquad\quad\cdot f'(T,\{j_{p,q},k_{p,q},a_{p,q}\}_{(p,q)\in T},C_h,\{U_{L,l_k}\}_{k=1}^m).
\end{split}
\end{equation*} 
By the properties of $C_h$ and the equality \eqref{eq_limit_equivalence} for $l=1$, we observe that the limits
$$\lim_{h\to +\infty\atop h\in2\N/\beta}B_{m,\{l_k\}_{k=1}^m}^{L,h},\ \lim_{L\to +\infty\atop L\in\N}\lim_{h\to +\infty\atop h\in2\N/\beta}B_{m,\{l_k\}_{k=1}^m}^{L,h}$$
exist. 

By using the inequalities \eqref{eq_covariance_exponential_decay_limit} and \eqref{eq_dominant_function} we have 
\begin{equation}\label{eq_dominant_L1_function}
\begin{split}
&|B_{m,\{l_k\}_{k=1}^m}^{L,h}((X_1^{l_1},\Xi_1^{l_1},\Phi_1^{l_1}),\cdots,(X_m^{l_m},\Xi_m^{l_m},\Phi_m^{l_m}))|\\
&\quad\le B_{m,\{l_k\}_{k=1}^m}((X_1^{l_1},\Xi_1^{l_1},\Phi_1^{l_1}),\cdots,(X_m^{l_m},\Xi_m^{l_m},\Phi_m^{l_m}))
\end{split}
\end{equation}
with
\begin{equation*}
\begin{split}
&B_{m,\{l_k\}_{k=1}^m}((X_1^{l_1},\Xi_1^{l_1},\Phi_1^{l_1}),\cdots,(X_m^{l_m},\Xi_m^{l_m},\Phi_m^{l_m}))\\
&\quad:= \beta^m \sum_{T\in\T(\{0,1,\cdots,m\})}\prod_{(p,q)\in T}\Bigg(\sum_{j_{p,q}=1}^{l_p}\sum_{k_{p,q}=1}^{l_q}\sum_{a_{p,q}\in\{-1,1\}}\\
&\qquad\cdot |U_{l_q}((\bx_{q,1},\cdots,\bx_{q,k_{p,q}-1},\bx_{q,k_{p,q}+1},\cdots,\bx_{q,l_q},\b0),\Xi_q^{l_q},\Phi_q^{l_q})|^{1_{l_q}\ge 2}\\
&\qquad\cdot 2 F_{t,t',d}\left(\frac{\pi}{2\beta}\right)^{-\frac{1}{2e \pi d}\sum_{v=1}^d|\<\bx_{q,k_{p,q}},\be_v\>|}\Bigg)g(T,\{j_{p,q},k_{p,q},a_{p,q}\}_{(p,q)\in T},\{l_k\}_{k=1}^m),
\end{split}
\end{equation*}
where $l_0:=\hm$. 
Note that $B_{m,\{l_k\}_{k=1}^m}$ is independent of $L$ and $h$. Moreover, by calculating in the same way as in the proof of Corollary \ref{cor_covariance_L1_bound} we obtain the bound that for any $m\in\N$ 
\begin{equation}\label{eq_dominant_bound_m}
\begin{split}
&\prod_{k=1}^m\left(\sum_{l_k=1}^{\tn}\sum_{(X^{l_k}_k,\Xi^{l_k}_k,\Phi^{l_k}_k)\in(\Z^d)^{l_k}\times\spin^{2l_k}}\right)B_{m,\{l_k\}_{k=1}^m}\\
&\quad\le \prod_{k=1}^m\left(2\beta\left(\frac{F_{t,t',d}\left(\frac{\pi}{2\beta}\right)^{\frac{1}{2e \pi d}}+1}{F_{t,t',d}\left(\frac{\pi}{2\beta}\right)^{\frac{1}{2e \pi d}}-1}\right)^{d}\sum_{l_k=1}^{\tn}(2\|U_{l_k}\|_{l_k})^{1_{l_k}\ge 2}4^{1_{l_k}=1}\right)\sum_{T\in\T(\{0,1,\cdots,m\})}\\
&\qquad\cdot\prod_{(p,q)\in T}\Bigg(\sum_{j_{p,q}=1}^{l_p}\sum_{k_{p,q}=1}^{l_q}\sum_{a_{p,q}\in\{-1,1\}}\Bigg) g(T,\{j_{p,q},k_{p,q},a_{p,q}\}_{(p,q)\in T},\{l_k\}_{k=1}^m)<+\infty.
\end{split}
\end{equation}

By \eqref{eq_each_term_decomposition}-\eqref{eq_dominant_bound_m} we can apply the dominant convergence theorem to ensure that the limits \eqref{eq_each_term_limits} exist and satisfy
\begin{equation*}
\begin{split}
&\lim_{h\to+\infty\atop h\in2\N/\beta}b_{0,m}=\prod_{k=1}^m\left(\sum_{l_k=1}^{\tn}\sum_{(X^{l_k}_k,\Xi^{l_k}_k,\Phi^{l_k}_k)\in(\Z^d)^{l_k}\times\spin^{2l_k}}\right)\lim_{h\to+\infty\atop h\in2\N/\beta}B_{m,\{l_k\}_{k=1}^m}^{L,h},\\
&\lim_{L\to+\infty\atop L\in\N}\lim_{h\to+\infty\atop h\in2\N/\beta}b_{0,m}=\prod_{k=1}^m\left(\sum_{l_k=1}^{\tn}\sum_{(X^{l_k}_k,\Xi^{l_k}_k,\Phi^{l_k}_k)\in(\Z^d)^{l_k}\times\spin^{2l_k}}\right)\lim_{L\to+\infty\atop L\in\N}\lim_{h\to+\infty\atop h\in2\N/\beta}B_{m,\{l_k\}_{k=1}^m}^{L,h},
\end{split}
\end{equation*}
which completes the proof.

\section*{Acknowledgments}
The author wishes to thank W. Pedra and M. Salmhofer for valuable comments as well as V. Bach for organizing opportunities for communication.

\end{document}